\pdfoutput=1
\documentclass[11pt,a4paper]{article}
\usepackage[a4paper,margin=25mm]{geometry}
\usepackage[T1]{fontenc}
\usepackage{lmodern}
\usepackage{microtype}
\usepackage{amsmath,amssymb,amsthm,mathtools}
\usepackage{graphicx,booktabs,multirow,array}
\usepackage{xcolor}
\usepackage{enumitem}
\usepackage[numbers,square,sort&compress]{natbib}
\usepackage[colorlinks=true,linkcolor=blue!50!black,citecolor=blue!50!black,urlcolor=blue!50!black]{hyperref}
\hypersetup{pdftitle={Finite-step statistical mechanics of scale-invariant learning: spectral equations of state and curvature transport},
  pdfauthor={Alexander Itin},
  pdfsubject={Statistical mechanics of learning dynamics},
  pdfkeywords={Stationary states; Fluctuation phenomena; Stochastic processes; Machine learning}}
\usepackage{url}
\usepackage[capitalize,nameinlink]{cleveref}
\usepackage{fancyhdr}
\newtheorem{theorem}{Theorem}
\newtheorem{proposition}[theorem]{Proposition}
\newtheorem{corollary}[theorem]{Corollary}
\newtheorem{lemma}[theorem]{Lemma}
\theoremstyle{definition}

\newtheorem{approximation}[theorem]{Approximation}

\theoremstyle{remark}

\newtheorem{observation}{Observation}
\crefname{conjecture}{Conjecture}{Conjectures}
\crefname{approximation}{Approximation}{Approximations}
\crefname{observation}{Observation}{Observations}

\newcommand{\aeos}{a_{\mathrm{EoS}}}
\newcommand{\E}{\mathbb{E}}
\newcommand{\Tr}{\operatorname{Tr}}
\newcommand{\R}{\mathbb{R}}
\newcommand{\cE}{c_{\mathrm{E}}}
\newcommand{\cR}{c_{\mathrm{R}}}
\newcommand{\Rtop}{R_{\mathrm{top}}}
\newcommand{\cA}{\mathcal A}
\renewcommand{\DH}{\mathcal D_H}
\newcommand{\Rem}{\mathcal R}
\newcommand{\cK}{\mathcal K}
\newcommand{\cM}{\mathcal M}
\newcommand{\Ng}{N_{\mathrm g}}
\newcommand{\Hreg}{H_{\mathrm{reg}}}
\newcommand{\Hpop}{H_{\mathrm{pop}}}
\newcommand{\Hsing}{H_{\mathrm{sing}}}
\newcommand{\rhoH}{\rho_H}
\newcommand{\sH}{\mathsf H}
\newcommand{\sS}{\mathsf S}
\newcommand{\ncheck}[1]{check~\ref{num:#1} in \cref{app:numerics}}
\title{Finite-step statistical mechanics of scale-invariant learning:\\ spectral equations of state and curvature transport}
\author{Alexander Itin\textsuperscript{1,2}\thanks{E-mail: \texttt{alexander.itin@ru.nl}, \texttt{alx.itin2@gmail.com}; website: \href{https://ai4science.de}{AI4science.de}}\\[4pt]
{\small \textsuperscript{1}Radboud University, Nijmegen, The Netherlands}\\
{\small \textsuperscript{2}Independent researcher, Hamburg, Germany}}
\date{}

\begin{document}
\maketitle
\thispagestyle{fancy}

\begin{abstract}
Stationary learning admits an ideal-gas description under restrictive noise assumptions. In neural networks whose gradient noise departs from these assumptions, however, an entropy-scaling relation resembling a Maxwell relation remains approximately valid. This raises two questions: which thermodynamic-looking relations follow from scale symmetry rather than equilibrium, and how can departures from the ideal-gas description be quantified? We address these questions for discrete stochastic gradient dynamics with memory. An exact rescaling of trajectories explains the entropy relation without assuming thermal equilibrium. Within a local harmonic approximation, self-adjusting weight norms yield a spectral equation of state whose virial coefficients are set by moments of the curvature spectrum. Coupled parameter groups can nevertheless sustain stationary probability currents. Our central result beyond this approximation is an exact curvature-transport identity. Together with a finite-step loss balance, it decomposes corrections to the harmonic fluctuation relation into measurable curvature changes, finite-step remainders, gradient--memory correlations and finite-window drift terms. A related displacement-weighted curvature change provides a diagnostic of the arrow of time. Its mean and second moment give a lower bound on the relative entropy between forward and reversed state pairs. We also calculate the activation-boundary contribution to curvature for rectified teacher--student networks with Gaussian inputs. Tests in sphere models, analytically tractable networks and convolutional networks on CIFAR-10 and CIFAR-100 support the transport description while identifying limits of the harmonic and cubic-remainder approximations. In the tested Softplus convolutional networks, negative curvature transport is consistent with the update-curvature deficit. Rectified networks show activation-boundary contributions on CIFAR-10 and, in one setting per dataset, a resolved part of the deficit that the identity assigns to the mixed gradient--memory term. The measurements also reveal time-reversal asymmetry over finite observation windows. These results separate symmetry-based relations from equilibrium assumptions and provide quantitative diagnostics of nonideal fluctuations in learning dynamics.
\end{abstract}

\noindent\textbf{Keywords:} Stationary states; Fluctuation phenomena; Stochastic processes; Machine learning.

\section{Introduction}\label{sec:intro}

A recently proposed thermodynamic description of stationary learning raises a concrete statistical-mechanical question. \citet{sadrtdinov2025idealgas} derived an ideal-gas equation of state for a scale-invariant learning model with isotropic noise of fixed strength. In this description, weight decay plays the role of pressure and half the squared weight norm that of volume. In their networks, however, the effective noise strength varies with the hyperparameters, mainly through the product of learning rate and weight decay. The ideal-gas law with a fixed noise strength therefore does not hold. With the measured noise strength, their prediction of the stationary radius remains accurate over much of the parameter range. An entropy-scaling relation resembling a Maxwell relation agrees with their measurements to within about 3\% in the fully scale-invariant ResNet-18 experiment on CIFAR-10. Why does the entropy relation remain robust when the fixed-noise assumptions are relaxed, and how should the fluctuation description be extended beyond the continuum approximation?

The distinction matters because an equation of state, an entropy identity and time-reversal symmetry express different properties. The failure of an ideal-gas law does not by itself rule out useful relations between collective variables. Conversely, agreement with a Maxwell-type relation does not establish thermal equilibrium: a stationary stochastic system can obey exact balance relations while carrying probability currents \citep{zia2007probability,seifert2012stochastic}. Our aim is to determine which relations in stationary learning follow from symmetry, which require approximations to the dynamics, and how their departures can be measured. The ideal-gas description is a reference for this investigation, not an equilibrium assumption imposed on the networks.

Learning at a fixed update size provides a setting in which these questions can be studied. The trainable parameters evolve under noisy gradient updates \citep{mandt2017sgd,li2017sme,yaida2019fdr}. Weight decay contracts their magnitudes. A common modification averages the gradients over preceding steps, which introduces a memory variable into the dynamics. Unlike a continuous-time diffusion, the resulting process evolves by finite displacements. Its stationary fluctuations can therefore depend on the update rule as well as on the loss landscape and the noise. At a fixed learning rate, SGD has also been described as minimizing a free energy whose temperature is set by the learning rate \citep{sadrtdinov2025free}. Even the discretization of a reversible diffusion can break time-reversal symmetry \citep{katsoulakis2014irreversibility}.

Scale-invariant parameter groups give this problem a useful geometry. Write the weights of a group as $w_a=r_an_a$, where $r_a$ is its norm and $n_a$ its direction. When rescaling $w_a$ leaves the predictions unchanged, the loss depends on $n_a$ but not on $r_a$. The norm nevertheless affects the dynamics: the effective angular step is proportional to $\eta/r_a^2$. For stochastic gradient descent, the gradient of an exactly invariant group is tangent to the sphere. Its finite displacement increases the squared norm, whereas weight decay contracts it. Their balance can support a fluctuating stationary regime in which the radii adjust together with the angular fluctuations \citep{vanlaarhoven2017l2,li2020exponential,li2020intrinsic,wan2021spherical}. The angular dynamics can also show periodic instabilities, distinct regimes of the effective learning rate and a dependence of the solution on the initial learning rate \citep{lobacheva2021periodic,kodryan2022three,sadrtdinov2024where}.

Such groups arise in normalized networks, for example in every convolution followed by batch normalization \citep{ioffe2015bn,vanlaarhoven2017l2}. They may be coupled to parameters that are not scale-invariant, including a trainable output layer. Training those additional parameters does not destroy the invariance of the hidden groups. We therefore distinguish groupwise scale symmetry from invariance of the entire model. We first treat the exact symmetry and then examine its violations in implemented normalized networks. This setting combines an explicit geometric constraint with anisotropic noise, nonlinear interactions and memory.

Several ingredients limit a simple continuum description. Finite updates modify the fluctuations even for a quadratic loss, especially when a curvature mode approaches its linear stability threshold \citep{liu2021noise}. Training tends to drive the sharpest modes towards this threshold \citep{cohen2021eos,cohen2022adaptive,andreyev2024eoss}. Gradient memory correlates the current update with earlier gradients, which were evaluated at different points of the landscape. Nonlinear variation of the curvature then enters these correlations. Nonsmooth activations add a contribution when an update crosses an activation boundary. Coupling between parameter groups also prevents their fluctuations from being described, in general, by independent scalar temperatures. These effects need not be captured by a single effective noise strength.

The continuum ideal-gas construction and the harmonic approximation used here should be distinguished. The former follows from the noise assumptions of \citet{sadrtdinov2025idealgas} and does not require a quadratic angular loss. We introduce a local harmonic approximation for a different purpose: it permits an explicit calculation of finite-step fluctuations. The stationary covariance of discrete stochastic gradient descent and momentum on quadratic losses is known \citep{liu2021noise}, as are exact mean-square stability conditions for stochastic gradient descent \citep{mulayoff2024exact}. Stationary fluctuation identities provide a complementary route that does not replace the update by a diffusion \citep{yaida2019fdr}. Our loss balance is of this type, with the finite-step remainder written out. Our analysis combines these approaches with scale symmetry and with observables that measure the failure of the local quadratic description.

The first result answers the question about the robustness of the entropy relation. An exact rescaling maps trajectories with one set of learning-rate and decay parameters to trajectories with another (\cref{prop:symmetry}). For corresponding stationary states, this fixes the change of the entropy without requiring a Gibbs distribution or detailed balance (\cref{cor:maxwell}). The map remains exact when non-invariant parameters are trained, provided their optimizer settings are held fixed while those of the invariant groups are rescaled. The entropy relation is therefore a consequence of a symmetry of the dynamics, rather than evidence that the stationary state is an equilibrium state. This extends the effective-learning-rate perspective of earlier work on normalized learning \citep{li2020exponential,wan2021spherical} and symmetry analyses of training dynamics \citep{kunin2021neural,ziyin2024symmetry}. Every invariant group also obeys an exact stationary norm balance, with and without gradient memory and with a trained head (\cref{prop:radial}).

The second result concerns the collective fluctuations. We couple the known finite-step mode variances \citep{liu2021noise} to the self-adjusting norms. Within the harmonic approximation, this gives a spectral equation of state for a group with isotropic angular noise, whose virial coefficients are set by moments of the curvature spectrum (\cref{sec:eos}). The approximation supplies calculable deviations from the ideal-gas reference. For SGD, this spectral equation of state contains no attractive term. Under the additional isotropic-noise approximation used for AdamW, the predicted stationary-radius correction instead has the opposite sign. \Citet{sadrtdinov2025idealgas} named such ``real-gas'' corrections, gradient memory and adaptive optimizers as open directions. For interacting groups, the covariance must instead be determined jointly. In the Gaussian harmonic surrogate of SGD with frozen radii, the stationary state is reversible if and only if the transformed curvature and noise matrices commute (\cref{prop:coupled}). Exact groupwise norm balances can therefore coexist with stationary probability currents. These calculations separate the existence of a stationary balance from the stronger conditions required for equilibrium.

Our central result beyond the harmonic approximation is an exact curvature-transport identity for dynamics with gradient memory (\cref{prop:transport}). The update direction is built from gradients evaluated at earlier states, while the curvature changes as the parameters move. The identity relates the departure of the current update curvature from its harmonic reference to this curvature change, a mixed gradient--memory correlation and finite-window drift. Together with the finite-decay identity derived from the loss balance (\cref{prop:energy,thm:finite}), which retains weight decay and the finite-step remainder, it provides a measurable decomposition of the fluctuation correction. Both relations hold for any network trained by SGD or gradient memory with hyperparameters shared by all parameters, without assuming a quadratic loss or Gaussian noise. Their interpretation in terms of angular fluctuations uses the additional scale symmetry.

\begin{figure}[t]
\centering
\includegraphics[width=\linewidth]{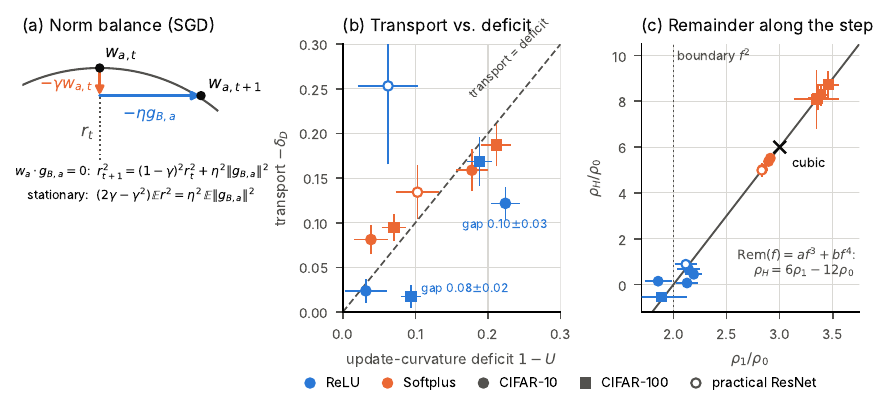}
\caption{The main results. (a)~Norm balance of an invariant group $w_a$, $r=\|w_a\|$ (\cref{prop:radial}; drawn to scale). The tangential step adds $\eta^2\|g_{B,a}\|^2$ to $r^2$, and decay removes $(2\gamma-\gamma^2)r^2$. The diagram shows SGD. Gradient-memory dynamics obeys the distinct norm balance \eqref{eq:radial-ema}. Both balances remain valid for an invariant group coupled to a trained head. Neither balance implies equilibrium. (b)~Transport against the update-curvature deficit, \eqref{eq:accounting} (seed means for the thin ResNet; runs with collapsed channels omitted). Softplus settings lie on the diagonal within 1.3 joint standard errors of the gap. The labelled ReLU settings lie below it. The practical ReLU network at $\gamma=10^{-4}$ (open symbol) lies above it. (c)~Endpoint ratios of the runs with $\rho_0$ resolved. Softplus runs lie on $\rhoH=6\rho_1-12\rho_0$, the line of a remainder $af^3+bf^4$ along the step. On CIFAR-10 they lie near the cubic point. ReLU runs sit near the boundary value 2 (\cref{lem:kink}), also on the line. However, only Softplus runs approximately satisfy it for individual sampled transitions (\cref{app:cubic}). Error bars: one standard error.}
\label{fig:overview}
\end{figure}

A related displacement-weighted curvature change provides a diagnostic of the arrow of time (\cref{prop:antisym}). It changes sign when the order of two consecutive parameter states is reversed. A nonzero mean therefore demonstrates time-reversal asymmetry of their joint distribution. Its mean and second moment give a lower bound on the relative entropy between the forward and reversed pair distributions. This connects a local curvature measurement with the trajectory-based description of irreversibility \citep{katsoulakis2014irreversibility,roldan2010estimating,roldan2012entropy}. \Citet{ziyin2026irreversibility} relate backward-error and time-reversal corrections to a regularized entropy-production construction in a small-step expansion. This result uses smoothness and symmetric-Jacobian assumptions and is distinct from the finite-step pair-distribution bound considered here. The bound concerns parameter pairs; it is not identified with the thermodynamic entropy production of the full optimizer. At finite weight decay, the mean pair observable also differs from the memory-based curvature transport by explicit terms of first and second order in the decay rate.

For nonsmooth models, the curvature measured within an activation region does not account for all finite-displacement effects. We distinguish smooth nonlinear remainders from contributions of activation-boundary crossings (\cref{lem:kink}). If the pre-activation density is smooth, activation-boundary crossings can produce an ensemble-averaged contribution of order $f^2$, where $f$ is the step fraction. For a smooth loss, the remainder after the quadratic expansion is $O(f^3)$; its leading cubic coefficient may vanish. For normalized committees of rectified linear units with Gaussian inputs, the boundary contribution to the curvature can be calculated analytically (\cref{prop:relu}). Competing interactions among the units and with the target function give it no universal sign. For smooth losses, we examine cubic and cubic-plus-quartic approximations to the remainder. These approximations reduce the number of observables needed to reconstruct the correction, but their validity must be tested separately from the exact identities.

The numerical investigation uses three complementary systems (\cref{sec:experiments}). Sphere models isolate the harmonic predictions and geometric effects. Teacher--student networks, whose targets are generated by a fixed reference network, provide an analytic population loss and directly evaluable derivatives. These models connect the present problem with the classical statistical mechanics of soft committee machines, from online learning \citep{saad1995online} to the equilibrium specialization transition \citep{biehl1998phase}. Convolutional networks on CIFAR-10 and CIFAR-100 then test which relations remain informative in higher-dimensional learning dynamics. In 28 runs with gradient memory, every term of the exact identities is measured at both ends of updates sampled from the optimizer's transition law. The measurements distinguish the general transport description from the simpler remainder approximations (\cref{fig:overview}). In the Softplus networks, negative curvature transport is consistent with the update-curvature deficit on both datasets. The cubic approximation is less general: the CIFAR-100 endpoint measurements require a substantial fourth-order contribution within the smooth-remainder description. On CIFAR-10, the remainder of rectified networks is dominated by activation boundaries. In one setting per dataset these networks also leave a resolved part of the deficit, which the identity assigns to the mixed contribution.

The purpose of these experiments is to test fluctuation relations and identify the contributions responsible for their departure from the harmonic reference. Agreement with an exact identity checks the estimators and the accounting of its terms; it does not by itself validate a predictive approximation. We therefore test the harmonic and remainder approximations separately. We also retain finite-window drift terms and examine time-reversal asymmetry without assuming that every network trajectory has reached a fully stationary distribution. The resulting description identifies which aspects of the gas analogy are protected by symmetry and which depend on the dynamics.

\Cref{sec:setting} defines the stochastic dynamics and notation (\cref{tab:notation}). \Cref{sec:exact} derives the exact scaling, balance and transport relations, together with the time-reversal diagnostic. \Cref{sec:linear} develops the harmonic fluctuation theory and the spectral equation of state. \Cref{sec:remainder} treats smooth nonlinearities, activation boundaries and non-invariant parameters. \Cref{sec:experiments} presents the numerical tests, and \cref{sec:status} discusses the conclusions and the remaining limitations.
\section{Stochastic dynamics of networks with scale-invariant groups}\label{sec:setting}
\paragraph{Scale-invariant groups.} The parameters $\theta=(w_1,\dots,w_{\Ng},\psi)\in\R^d$ consist of weight groups $w_a\in\R^{d_a}$ and possibly parameters $\psi$ that are not scale-invariant. Examples are an output layer or normalization coefficients. These parameters may be trained. The minibatch loss satisfies $L_B(\dots,\alpha w_a,\dots,\psi)=L_B(\theta)$ for every $\alpha>0$, group $a$, batch $B$ and value of $\psi$. Indeed, rescaling a normalized filter leaves the features seen by the downstream layers unchanged. A network is \emph{fully} scale-invariant if there is no $\psi$. We then write $w=\theta$. The objective is $L(\theta)=\E_BL_B(\theta)$. With batch normalization in training mode, this is the expected loss of a random batch. It is not the loss of the training set treated as one batch. We write $r_a=\|w_a\|$ and $n_a=w_a/r_a$ for the norm and direction of a group. We also write $g(\theta)=\nabla L(\theta)$ for the \emph{mean gradient}, $g_B=\nabla L_B(\theta)$ for a minibatch gradient, $\xi=g_B-g$ for the gradient noise and $\Sigma=\mathrm{Cov}_B(g_B)$ for its covariance. Euler's theorem for functions homogeneous of degree zero gives, for every batch and group,
\begin{equation}\label{eq:si}
w_a\cdot\nabla_{w_a}L_B=0,\qquad \nabla_{w_a}L_B(\dots,\alpha w_a,\dots)=\alpha^{-1}\nabla_{w_a}L_B(\theta),\qquad H_Bw_a=-g_{B,a},
\end{equation}
where $H_B$ is the Hessian of $L_B$. The last relation follows by differentiating $\sum_{i\in a}w_i\partial_iL_B=0$. Here $w_a$ and $g_{B,a}$ are embedded in $\R^d$. For a fully scale-invariant network, summing over groups gives $H_Bw=-g_B$. Rescaling a group also leaves $\nabla_\psi L_B$ unchanged.

\paragraph{Two Hessians for piecewise-smooth networks.} For smooth activations $H=\nabla^2L$. Networks with rectified linear units (ReLU) are only piecewise smooth. For them, two curvatures must be distinguished. The first is the \emph{regular} Hessian $\Hreg=\E_B[\nabla^2_{\mathrm{reg}}L_B]$. It is the second derivative with the activation pattern held fixed (what automatic differentiation returns), averaged over batches. The second is the population Hessian $\Hpop=\nabla^2\E_BL_B$, where it exists. Their difference $\Hsing=\Hpop-\Hreg$ is the \emph{singular curvature} of activation boundaries (\cref{sec:kinks}). Throughout, $H$ denotes $\Hreg$. For smooth activations it coincides with $\Hpop$. \Cref{eq:si} holds for the regular Hessian sample by sample.

\paragraph{Dynamics with and without gradient memory.} We compare stochastic gradient descent (SGD) with a gradient-memory dynamics, commonly called exponential-moving-average (EMA) momentum. With decoupled weight decay $\lambda$, learning rate $\eta$ and $\gamma=\eta\lambda$, it reads
\begin{equation}\label{eq:update}
m_{t+1}=\beta m_t+(1-\beta)\,g_{B_t}(\theta_t),\qquad \theta_{t+1}=(1-\gamma)\,\theta_t-\eta\,m_{t+1},
\end{equation}
with common hyperparameters for all parameters unless stated otherwise. \Cref{prop:symmetry} allows separate hyperparameters per group. The batch $B_t$ is drawn afresh at every step, so that the noise is conditionally unbiased, $\E[g_{B_t}(\theta_t)\,|\,\mathcal F_t]=g(\theta_t)$, where $\mathcal F_t$ is the history up to step $t$. The averaged identities of \cref{sec:exact} use this property. Independent sampling at every step satisfies it. Epoch-wise reshuffling does not, because the next batch then depends on the examples already used. All runs of this paper sample independently at every step: the networks draw each batch uniformly with replacement from the training set, and the sphere and teacher--student models draw fresh noise and fresh inputs.
The memory variable $m_t$ is an exponentially weighted average of past stochastic gradients. It is not a mechanical momentum. The parameter $0\le\beta<1$ sets the persistence of past gradients. The value $\beta=0$ is SGD, $\theta_{t+1}=(1-\gamma)\theta_t-\eta g_{B_t}$. We write $\Delta\theta_t=\theta_{t+1}-\theta_t$ ($\Delta w$ for a fully scale-invariant network) and use
\begin{equation}\label{eq:kappa}
\kappa=\frac{1-\beta}{1+\beta},\qquad q=\beta(1-\gamma),\qquad \aeos=\frac{2(1+\beta)}{1-\beta}.
\end{equation}
$\aeos$ is the value of $\eta h$ at which a quadratic mode of curvature $h$ becomes linearly unstable when $\gamma\to0$ ($\aeos=2$ for SGD). Full-batch training drives the sharpness to this threshold, where it oscillates \citep{cohen2021eos,arora2022eos,damian2023selfstab}. Analogues exist for adaptive methods \citep{cohen2022adaptive} and for SGD \citep{andreyev2024eoss}. Momentum beyond the threshold is studied in \citep{itin2026physics}. The heavy-ball normalization $m\leftarrow\beta m+g_B$ has the threshold $2(1+\beta)$ and, in the continuum limit, a temperature larger by $1/(1-\beta)$; at equal $\eta$ the finite-step factors differ as well. We do not use it except where stated. AdamW \citep{kingma2015adam,loshchilov2019adamw} replaces $m_{t+1}$ by $\hat m_{t+1}/(\sqrt{\hat v_{t+1}}+\epsilon)$ with the usual bias-corrected moments. In the networks, $\beta=\beta_1=0.9$, $\beta_2=0.999$ and $\epsilon=10^{-8}$. We treat AdamW as an extension. Consider the preconditioned coordinates $y=P^{1/2}\theta$, with $P=\sqrt{\hat v}+\epsilon$ \emph{frozen} at a diagnostic point. In these coordinates AdamW is EMA momentum with $g\to P^{-1/2}g$, $H\to P^{-1/2}HP^{-1/2}$ and $\Sigma\to P^{-1/2}\Sigma P^{-1/2}$. This is an approximation. The actual preconditioner changes along the trajectory and is correlated with the current batch.

\paragraph{Angular variables.} On a sphere of radius $r$ the relevant step is the angular one, $\eta_0=\eta/r^2$. We describe the stationary angular fluctuations of a group by tangent modes $k$. Each mode has a curvature $h_k$ of the loss on the unit sphere and an angular noise variance $s_k$. At a reference critical direction, $\mu_k=h_k/r^2$ are the corresponding tangent eigenvalues of $H$ at radius $r$, so that $\eta_0h_k=\eta\mu_k$. Away from that reference the identification is a local approximation, and the radial--tangent coupling $Hw=-g$ must be kept where it matters. Stationary harmonic modes have $h_k>0$; negative-curvature and flat directions need separate treatment. The \emph{edge ratio} $R_k=\eta_0h_k/\aeos$ is the fraction of the linear stability threshold reached by mode $k$. The distance to the threshold is $1-R_k$. We write $\Rtop$ for the edge ratio of the sharpest mode.

\paragraph{Finite-step observables.} For an update $\Delta\theta$ and an evaluation batch $B'$ independent of it, the remainder of the quadratic loss model along a fraction $f$ of the step is
\begin{equation}\label{eq:remainder}
\Rem(f)=L_{B'}(\theta+f\Delta\theta)-L_{B'}(\theta)-f\,g_{B'}\cdot\Delta\theta-\tfrac{f^2}{2}\,\Delta\theta^\top H_{B'}\Delta\theta,
\end{equation}
with the regular Hessian $H_{B'}$. With $Q=\frac12\E[\Delta\theta^\top H\Delta\theta]$ we use the normalized remainder, its endpoint derivative and the curvature change along the step,
\begin{equation}\label{eq:rhos}
\rho_0=\frac{\E\Rem(1)}{Q},\qquad \rho_1=\frac{\E\Rem'(1)}{Q},\qquad \rhoH=\frac{\E\big[\Delta\theta^\top\big(H(\theta_{t+1})-H(\theta_t)\big)\Delta\theta\big]}{Q},
\end{equation}
where, sample by sample, $\Rem'(1)=[g_{B'}(\theta+\Delta\theta)-g_{B'}(\theta)-H_{B'}\Delta\theta]\cdot\Delta\theta$. Since $B'$ is independent of the step, the conditional mean $\Rem_L(f)=\E_{B'}[\Rem(f)\,|\,\theta,\Delta\theta]$ is the remainder of the objective $L$, and $\Rem_L'(1)=(g_{t+1}-g_t-H\Delta\theta)\cdot\Delta\theta$ with $g_t=g(\theta_t)$ and the regular Hessian $H$. We also use the remainder profile $\rho(f)=\E\Rem(f)/(\frac{f^2}{2}\E[\Delta\theta^\top H\Delta\theta])$, so that $\rho_0=\rho(1)$. For smooth losses $\rhoH=\E\Rem''(1)/Q$. For ReLU networks $\rho_0$ and $\rho_1$ include the activation boundaries. By contrast, $\rhoH$ uses the regular Hessian at both ends and is not the second derivative of the mean remainder (\cref{sec:kinks}).

\begin{table}[t]
\centering\footnotesize
\caption{Notation. Expectations are over the stationary law (or the time marginals of a finite window) and over the noise.}
\label{tab:notation}
\setlength{\tabcolsep}{4pt}
\begin{tabular}{@{}>{\raggedright\arraybackslash}p{0.17\linewidth}>{\raggedright\arraybackslash}p{0.29\linewidth}>{\raggedright\arraybackslash}p{0.24\linewidth}>{\raggedright\arraybackslash}p{0.24\linewidth}@{}}
\toprule
$\eta,\ \lambda,\ \gamma=\eta\lambda$ & learning rate, weight decay, decay per step & $G=\E\|g\|^2$ & mean-gradient power\\
$\beta,\ \kappa,\ q$ & memory parameter, \eqref{eq:kappa} & $B_H=\E[g^\top Hg]$ & curvature along $g$\\
$\aeos$ & stability threshold of $\eta h$ & $\Xi_0=\E\Tr\Sigma$ & noise power\\
$\theta=(w_a,\psi)$; $r_a,\ n_a$ & parameters; norm and direction of group $a$ & $\Xi_H=\E\Tr(\Sigma H)$ & noise curvature\\
$\eta_0=\eta/r^2$ & angular step & $\cK=\E[m_{t+1}^\top H(\theta_t)m_{t+1}]$ & update curvature\\
$h_k,\ s_k$ & curvature and angular noise of mode $k$ & $\cK_{\mathrm{lin}}$, $U=\cK/\cK_{\mathrm{lin}}$ & its linear value, ratio\\
$R_k=\eta_0h_k/\aeos$ & edge ratio & $c$ & linear-FDR ratio, \eqref{eq:F2}\\
$T^{(0)}_k,\ T_k$ & mode temperatures, \eqref{eq:modes} & $\Rem(f)$, $Q$; $\rho_0,\rho_1,\rhoH$ & remainder, \eqref{eq:remainder}--\eqref{eq:rhos}\\
$H,\ \Hpop,\ \Hsing$ & regular, population and singular Hessian & $\cA$; $e$; $\mathcal S$ & pair observable \eqref{eq:A}; rescaling response \eqref{eq:e}; entropy\\
\bottomrule
\end{tabular}
\end{table}
\section{Exact relations}\label{sec:exact}

The statements of this section are exact for the update \eqref{eq:update} under the stated assumptions. They also assume that the moments involved are finite and that the trajectories are bounded away from the singular origin $w_a=0$. The statements differ in scope. The norm balance of a group and the rescaling symmetry require the scale invariance of that group, but not of the rest of the network. The following relations hold for the full parameter vector $\theta$ of any network: the loss balance, the finite-decay identity (with an explicit term for non-invariant parameters and common $\eta$, $\lambda$, $\beta$), the curvature-transport identity (common $\beta$) and the antisymmetry of the pair observable. Their reading as a correction to the linear fluctuation--dissipation relation, \eqref{eq:F2}, uses full scale invariance. ``Stationary'' means that the laws of $(\theta_t,m_t)$ at consecutive times coincide. When they do not, the drift terms $d_\bullet=\E[\bullet_{t+1}-\bullet_t]$ are kept. The expectations can then be read as averages over the time marginals of a finite window.

\subsection{Norm balance and two radius baselines}

\begin{proposition}[Norm balance]\label{prop:radial}
Let a parameter group $a$ be updated as $w_{a,t+1}=(1-\gamma_a)w_{a,t}-\eta_au_{a,t}$, with $u_a=g_{B,a}$ (SGD), $u_a=m_{a,t+1}$ (memory) or the AdamW direction. In any stationary state,
\begin{equation}\label{eq:radial-adam}
(2\gamma_a-\gamma_a^2)\,\E r_a^2=\eta_a^2\,\E\|u_a\|^2-2\eta_a(1-\gamma_a)\,\E[w_a\cdot u_a].
\end{equation}
If $L_B$ is invariant under rescaling of $w_a$ (for any values of the other parameters, which may be trained), then
\begin{align}
\text{SGD:}\quad &(2\gamma-\gamma^2)\,\E r_a^2=\eta^2\,\E\|g_{B,a}\|^2=\eta^2\big(\E\|g_a\|^2+\E\Tr\Sigma_{aa}\big),\label{eq:radial-sgd}\\
\text{gradient memory:}\quad &(2\gamma-\gamma^2)\,\E r_a^2=\eta^2\,\E\|m_a\|^2\,\frac{1+q}{1-q},\label{eq:radial-ema}
\end{align}
with $\gamma$, $\eta$, $q$ those of the group.
\end{proposition}
\begin{proof}
\eqref{eq:radial-adam} is the expansion of $\|w_{a,t+1}\|^2$. For an invariant group, $w_a\perp g_{B,a}$ by \eqref{eq:si}. The cross term therefore vanishes for SGD. With memory, let $\varpi_t=w_{a,t}\cdot m_{a,t}$. Then $w_{a,t}\cdot m_{a,t+1}=\beta\varpi_t$, $\varpi_{t+1}=q\,\varpi_t-\eta\|m_{a,t+1}\|^2$ and $\E\varpi=-\eta\E\|m_a\|^2/(1-q)$. Inserting $\E[w_{a,t}\cdot m_{a,t+1}]=\beta\E\varpi$ into \eqref{eq:radial-adam} gives \eqref{eq:radial-ema}.
\end{proof}
Consider a group whose invariance holds only approximately, for example a channel normalized as $(z-\mu_B)/\sqrt{v_B+\epsilon_{\mathrm{BN}}}$ with fixed $\epsilon_{\mathrm{BN}}>0$. For such a group the violation is measured by $\E[w_a\cdot g_{B,a}]$. This quantity vanishes for an exactly invariant group with any optimizer. For SGD it is the cross term of \eqref{eq:radial-adam}. With memory or AdamW, however, the cross term is nonzero even for an exactly invariant group.

Weight decay contracts the norm deterministically. The stochastic gradient, orthogonal to the weights, increases it. Two baselines for the radius follow, and they must be kept apart. The \emph{measured-noise} baseline of \citep{sadrtdinov2025idealgas} drops the gradient signal: $r^2_{\mathrm{SDE}}=\eta\Xi_0/(2\lambda)$. \Cref{prop:radial} gives exactly
\begin{equation}\label{eq:baselines}
\frac{\E r^2}{r^2_{\mathrm{SDE}}}=\frac{1+G/\Xi_0}{1-\gamma/2}\quad\text{(SGD)},\qquad
\frac{\E r^2}{r^2_{\mathrm{SDE}}}=\frac{1+q}{1-q}\,\frac{\E\|m\|^2/\Xi_0}{1-\gamma/2}\quad\text{(memory)},
\end{equation}
and in the harmonic $\gamma\to0$ limit the memory ratio becomes $1+\kappa G/\Xi_0$. The \emph{fixed-angular-noise} baseline $r_0$ of the equation of state (\cref{sec:eos}) is a different object. It fixes the noise per tangent direction instead of using its measured trace.

\subsection{Scaling symmetry and the entropy relation}

\begin{proposition}[Scaling symmetry]\label{prop:symmetry}
Let every invariant group $a$ have its own learning rate $\eta_a$, decay $\lambda_a$ and memory parameter (SGD, EMA momentum or heavy ball). Let the non-invariant parameters $\psi$ be trained with fixed hyperparameters by any optimizer whose update of $\psi$ depends only on $\psi$, its own optimizer state and $\nabla_\psi L_B$. Clipping by a global gradient norm, or trust ratios computed jointly across groups, break the map below. For any $\alpha_a>0$, consider the map
\[
w_a\mapsto\alpha_a^{1/2}w_a,\qquad m_a\mapsto\alpha_a^{-1/2}m_a,\qquad (\eta_a,\lambda_a)\mapsto(\alpha_a\eta_a,\lambda_a/\alpha_a),
\]
with $\psi$, its optimizer state and its hyperparameters unchanged. This map sends a trajectory with minibatch sequence $(B_t)$ to the trajectory with the transformed hyperparameters and the same batches. For AdamW with $\epsilon=0$ the same holds with $w_a\mapsto\alpha_aw_a$ and the moments $(m_a,v_a)\mapsto(m_a/\alpha_a,v_a/\alpha_a^2)$. For a fully scale-invariant network and a common $\alpha$ this is the map $(\eta,\lambda)\mapsto(\alpha\eta,\lambda/\alpha)$.
\end{proposition}
\begin{proof}
$\alpha_a\eta_a\cdot\lambda_a/\alpha_a=\eta_a\lambda_a$, so the decay factors are unchanged. By \eqref{eq:si}, $\nabla_{w_a}L_B$ at the rescaled point is $\alpha_a^{-1/2}\nabla_{w_a}L_B$. Hence $\alpha_a\eta_a$ times it is $\alpha_a^{1/2}$ times the original step, and the memory scales as stated. $L_B$ is unchanged as a function of $\psi$. Hence $\nabla_\psi L_B$, and with it the whole trajectory of $\psi$, is unchanged. For AdamW the update is invariant under a rescaling of the gradients of a group when $\epsilon=0$.
\end{proof}
Rescaling arguments of this kind underlie the effective learning rate of normalized networks \citep{li2020exponential,wan2021spherical}. What matters here is that the map is exact for the discrete stochastic dynamics, with memory and with a trained head. A trained non-invariant block breaks only the \emph{common} rescaling of all learning rates, including its own. In the teacher--student network with a trained output layer, the mapped trajectories coincide to rounding error while the head moves by $O(1)$. A common rescaling of all learning rates, in contrast, changes them at the per-cent level (\ncheck{head}).

\begin{corollary}[Dependence on $\eta\lambda$; exact entropy relation]\label{cor:maxwell}
Keep the model, the batch size and the noise law fixed. Let the invariant groups share $(\eta,\lambda)$, and keep the hyperparameters of $\psi$ fixed. In the variables $z=w/\sqrt\eta$ and $p=\sqrt\eta\,m$ of the invariant part, the updates read $z_+=(1-\gamma)z-\nabla_zL_B(z,\psi)$ (SGD) and $p_+=\beta p+(1-\beta)\nabla_zL_B(z,\psi)$, $z_+=(1-\gamma)z-p_+$ (memory). In addition, $\nabla_\psi L_B(w,\psi)=\nabla_\psi L_B(z,\psi)$. For AdamW with $\epsilon=0$ the same statements hold for $z=w/\eta$. The dynamics therefore involves $(\eta,\lambda)$ only through $\gamma$, and the rescaling of \cref{prop:symmetry} maps every stationary law at $(\eta,\lambda)$ to a stationary law at $(\alpha\eta,\lambda/\alpha)$. The following statements hold for stationary laws that correspond in this way. (i)~The directions $(n_a)_a$ and the parameters $\psi$ have identical distributions. Every observable unchanged by the rescaling depends on $\eta$ and $\lambda$ only through their product. These include the loss, the accuracy, the edge ratios and, for fully scale-invariant networks, the ratios studied below. (ii)~At fixed $\gamma$, $r_a\propto\eta^\ell$ in law, with $\ell=\frac12$ for SGD and memory and $\ell=1$ for AdamW ($\epsilon\to0$). (iii)~Let $\mathcal S=\mathcal S(n,\psi)+\sum_a(d_a-1)\E\log r_a$ be the per-group thin-shell entropy, and assume that it is finite for the stationary laws considered. Then $\mathcal S(\alpha\eta,\lambda/\alpha)-\mathcal S(\eta,\lambda)=\ell\sum_a(d_a-1)\log\alpha$. Wherever the derivatives exist,
\begin{equation}\label{eq:maxwell}
\Big(\frac{\partial \mathcal S}{\partial\log\eta}\Big)_\lambda-\Big(\frac{\partial \mathcal S}{\partial\log\lambda}\Big)_\eta=\ell\sum_a(d_a-1),
\end{equation}
exactly, including all discrete-time and edge-of-stability effects. The derivatives are taken at fixed hyperparameters of $\psi$. For the entropy built from one global radius of the invariant weights \citep{sadrtdinov2025idealgas} the right-hand side is $\ell(\sum_ad_a-1)$. For the differential entropy of the full density of $w$ it is $\ell\sum_ad_a$. With separate hyperparameters per group the relation holds group by group, with right-hand side $\ell(d_a-1)$.
\end{corollary}
\begin{proof}
By \cref{prop:symmetry}, the shift $(\log\eta,\log\lambda)\mapsto(\log\eta+\delta,\log\lambda-\delta)$ leaves the joint law of the directions and of $\psi$ unchanged. It multiplies every radius by $e^{\ell\delta}$, so $\mathcal S$ changes by $\ell\delta\sum_a(d_a-1)$. Hence $\mathcal S-\ell\sum_a(d_a-1)\log\eta$ is a function of $\log\eta+\log\lambda$ alone, and \eqref{eq:maxwell} follows by differentiation.
\end{proof}
Neither uniqueness of the stationary law nor detailed balance is assumed. Uniqueness is needed only to identify stationary states reached independently, for example from different seeds. Such states need not correspond. For SGD with one group of dimension $d_1$, \eqref{eq:maxwell} is the Maxwell relation with coefficient $(d_1-1)/2$. This is the relation that \citet{sadrtdinov2025idealgas} derived from the ideal gas and confirmed within about 3\% for ResNet-18 on CIFAR-10. The corollary explains why it survives where the gas law with a fixed noise strength fails. It also explains why their noise level depends on $\eta\lambda$ only. For AdamW the coefficient is $d-1$, exact as $\epsilon\to0$. If the learning rate and decay shared by the head are varied as well, the relation acquires a correction set by the response of the head (\cref{sec:mixture}). The relation follows from a symmetry of the dynamics, not from a Gibbs form of the stationary law.

\subsection{Loss balance}

\begin{proposition}[Loss balance]\label{prop:energy}
Let the update be $\Delta\theta=-\eta v$ with any $v$ independent of the evaluation batch $B'$. For \eqref{eq:update}, $v=m_{t+1}+\lambda\theta$. Then, for any network,
\begin{equation}\label{eq:energy-general}
\E[g\cdot v]=\frac\eta2\,\E[v^\top Hv]+\frac{\E\Rem(1)-d_L}{\eta},\qquad d_L=\E[L(\theta_{t+1})-L(\theta_t)].
\end{equation}
For a fully scale-invariant network, $g\cdot w=0$ and $Hw=-g$ give $(1+\gamma)\E[g\cdot m_{t+1}]=\frac\eta2\cK+(\E\Rem(1)-d_L)/\eta$. For SGD this becomes
\begin{equation}\label{eq:energy-sgd}
(1+\gamma)\,G=\frac\eta2\,\big(\Xi_H+B_H\big)+\frac{\E\Rem(1)-d_L}{\eta}.
\end{equation}
\end{proposition}
\begin{proof}
\eqref{eq:remainder} at $f=1$ is the definition of $\Rem$. Average over $B'$, which is independent of $v$ and of $\theta_{t+1}$. Then use $\E_{B'}L_{B'}(\theta_{t+1})=L(\theta_{t+1})$. For the scale-invariant form insert \eqref{eq:si}.
\end{proof}
We call $\cE=\E[g\cdot v]/(\frac\eta2\E[v^\top Hv])$ the \emph{loss-balance ratio} and $\cR=1+\E\Rem(1)/Q$ its \emph{remainder prediction}. Since \eqref{eq:energy-general} is an identity, $\cE=\cR-d_L/Q$ for any network. The agreement $\cE\approx\cR$ in data is a check of stationarity and of independent estimators (\cref{app:estimators}). The physics is in the decomposition: how much of the deviation from the quadratic model sits in $\Rem$, and how $\Rem$ scales with $f$. For SGD with $\gamma\to0$ and $\Rem=0$, \eqref{eq:energy-sgd} is the linear fluctuation--dissipation relation (FDR) \eqref{eq:fdr} with $\aeos=2$, including its edge term $\frac\eta2B_H$. Every deviation from the linear FDR is then the finite-step remainder.

\subsection{Gradient memory: finite-decay identity and curvature transport}

With memory, define the drifts $d_{gm}=\E[g_{t+1}\cdot m_{t+1}-g_t\cdot m_t]$ and $d_{\cM}=\E[m_{t+1}^\top H(\theta_{t+1})m_{t+1}-m_t^\top H(\theta_t)m_t]$, the remainder expectations $\bar\Rem_i=Q\rho_i$, and the response of the loss to a global rescaling of all parameters,
\begin{equation}\label{eq:e}
e(\theta)=\theta\cdot\nabla L(\theta)=\frac{d}{ds}L(e^s\theta)\Big|_{s=0},\qquad \bar e=\E\,e(\theta_t),\qquad d_e=\E[e(\theta_{t+1})-e(\theta_t)].
\end{equation}
For a fully scale-invariant network $e\equiv0$. For exactly invariant groups coupled to non-invariant parameters, $e=\psi\cdot\nabla_\psi L$. The fresh noise $\xi_t$ is conditionally unbiased given the past.

\begin{theorem}[Finite-decay, finite-window identity]\label{thm:finite}
For gradient memory \eqref{eq:update} with common $\eta$, $\gamma$ and $\beta$ for all parameters, on any network,
\begin{equation}\label{eq:finite}
\eta(1-\beta)\,\big(G+\lambda\bar e\big)=(1+q)\,Q+(1-q)\,(\bar\Rem_0-d_L)+q\,\bar\Rem_1+\eta\beta\,d_{gm}+\beta\gamma\,d_e .
\end{equation}
Define the linear-FDR ratio $c=G/[\frac\eta2(\Xi_H+\kappa B_H)]$, the linear value $\cK_{\mathrm{lin}}=\kappa\Xi_H+\kappa^2B_H$ of the update curvature (\cref{prop:fdr}) and $U=\cK/\cK_{\mathrm{lin}}$. For a fully scale-invariant network, in a stationary state,
\begin{equation}\label{eq:F2}
c=U\,\frac{(1+q)+(1-q)\rho_0+q\rho_1}{(1+\beta)\,[1+\gamma(1+\rho_0)]},\qquad \big[1+\gamma(1+\rho_0)\big]\,Q=\frac{\eta^2\cK}{2}.
\end{equation}
\end{theorem}
The proof (\cref{app:proofs}) uses only the loss expansion along the step, the definition of $e$ and the unbiasedness of the fresh noise. No quadratic approximation, Gaussian noise or $\gamma\to0$ limit is needed. Scale invariance enters only through $e$. The relation \eqref{eq:F2} additionally uses $\theta\cdot g=0$ and $H\theta=-g$. A non-invariant block therefore enters through the measurable term $\lambda\bar e$ rather than as an unspecified correction. Because \eqref{eq:finite} and \eqref{eq:transport} are identities, evaluating them numerically validates the estimators and the accounting of the terms, not a predictive model. The content lies in the size and sign of the individual terms, which the identities do not fix. In the teacher--student network with a trained output layer and common $(\eta,\lambda,\beta)$, \eqref{eq:finite} holds within its statistical error of $3\cdot10^{-4}$. The form without the head term, in contrast, misses by up to 1.5\%, 60 standard errors (\ncheck{head}). For $\beta=0$, $\cK=\Xi_H+B_H$ and $U=1$ exactly. Then \eqref{eq:F2} reduces to the stationary form of \eqref{eq:energy-sgd}, $c=(1+\rho_0)/[1+\gamma(1+\rho_0)]$. For $\gamma\to0$,
\begin{equation}\label{eq:decomp}
c=U\Big[1+\frac{(1-\beta)\rho_0+\beta\rho_1}{1+\beta}\Big].
\end{equation}
When $Q$ or a reference denominator is poorly resolved, the unnormalized form \eqref{eq:finite} is the one to test. For a fully scale-invariant network, an intermediate step of the proof gives the first-moment relation $\E[g\cdot m_t]+(1-\gamma)\E[g\cdot m_{t+1}]=(1-\gamma)(2\bar\Rem_0-\bar\Rem_1-2d_L)/\eta-d_{gm}$. Without remainders and drifts it gives $\E[g\cdot m_t]=-(1-\gamma)(1-\beta)G/(1+q)\to-\kappa\,G$. On average, the memory is thus anti-aligned with the current mean gradient.

\begin{proposition}[Curvature transport]\label{prop:transport}
With $E_X=\E[g^\top Hm_t]+\kappa B_H$ and the \emph{curvature transport}
\begin{equation}\label{eq:DH}
\DH=\E\big[m_{t+1}^\top\big(H(\theta_{t+1})-H(\theta_t)\big)m_{t+1}\big],
\end{equation}
gradient memory obeys, on any network and without any stationarity assumption,
\begin{equation}\label{eq:transport}
\cK=\cK_{\mathrm{lin}}+\frac{2\beta}{1+\beta}E_X+\frac{\beta^2}{1-\beta^2}\,\big(\DH-d_{\cM}\big).
\end{equation}
\end{proposition}
\begin{proof}
Let $\cM=\E[m_t^\top H(\theta_t)m_t]$ be the pre-update curvature. Since the fresh noise is conditionally unbiased, $\cK=\beta^2\cM+2\beta(1-\beta)\E[m_t^\top Hg]+(1-\beta)^2(B_H+\Xi_H)$. By definition $\cM=\cK+\DH-d_{\cM}$. Eliminate $\cM$ and use $\kappa B_H-\frac{2\beta\kappa}{1+\beta}B_H=\kappa^2B_H$.
\end{proof}
Divided by $\cK_{\mathrm{lin}}$, \eqref{eq:transport} is an exact accounting of the update-curvature deficit into a mixed term and a transport term, which includes the drift of the memory curvature:
\begin{equation}\label{eq:accounting}
U-1=\delta_X+\delta_D,\qquad \delta_X=\frac{2\beta}{1+\beta}\,\frac{E_X}{\cK_{\mathrm{lin}}},\qquad \delta_D=\frac{\beta^2}{1-\beta^2}\,\frac{\DH-d_{\cM}}{\cK_{\mathrm{lin}}}.
\end{equation}
The identity needs neither scale invariance, nor Gaussian noise, nor a fixed Hessian, nor a quadratic loss, nor $\gamma=0$. It distinguishes the curvature $\cK$ of the current update from the curvature $\cM$ of the memory before the update. Apart from the mixed term $E_X$, the update curvature falls short of its linear value exactly when the curvature seen by the same direction decreases after the step. The factor $\beta^2/(1-\beta^2)$ ($4.26$ for $\beta=0.9$) quantifies the amplification of curvature transport by gradient memory. The average keeps re-using its direction while the weights move along it. With $\gamma\to0$, $\Delta\theta=-\eta m_{t+1}$ and $\DH=Q\rhoH/\eta^2=\rhoH\cK/2$. At stationarity, where $d_{\cM}=0$, \eqref{eq:transport} then gives $U=(1+\delta_X)/(1-\frac{\beta^2}{2(1-\beta^2)}\rhoH)$. Otherwise the numerator contains the additional term $-\frac{\beta^2}{1-\beta^2}\,d_{\cM}/\cK_{\mathrm{lin}}$. Together with \eqref{eq:decomp}, this expresses $c$ through $\rho_0$, $\rho_1$, $\rhoH$ and $\delta_X$. At finite $\gamma$, $\DH$ and $\rhoH$ differ by the terms of \eqref{eq:A-decay} below. Both contain geometric contributions of order $\gamma$ from the change of the radius along the step. In the teacher--student runs of \cref{sec:ts}, $\DH\approx-4\gamma\cK$ and $\rhoH\approx-(0.6\text{--}1)\gamma$ for $\beta\ge0.9$ with negligible anharmonicity. In addition, the radial part of the memory contributes to $E_X$ (\cref{app:lyapunov}). The relations to use are then \eqref{eq:F2} and \eqref{eq:transport}. The geometric parts must be separated before $\DH$ or $\rhoH$ is read as transport of the landscape curvature.

\subsection{Curvature changes and the arrow of time}\label{sec:antisym}
The displacement-weighted change in curvature provides a diagnostic of the arrow of time. It changes sign when the order of two consecutive parameter states is reversed. A nonzero mean therefore demonstrates time-reversal asymmetry of their joint distribution. Its mean and second moment give a lower bound on the relative entropy between the forward and reversed pair distributions.

\begin{proposition}[Curvature change along a step]\label{prop:antisym}
For consecutive states define
\begin{equation}\label{eq:A}
\cA(\theta,\theta')=(\theta'-\theta)^\top\big[H(\theta')-H(\theta)\big](\theta'-\theta).
\end{equation}
(i)~On any network $\cA(\theta',\theta)=-\cA(\theta,\theta')$. Hence every reversal-symmetric law of consecutive states, $\pi(\theta,\theta')=\pi(\theta',\theta)$, has $\E\cA=0$. A nonzero mean therefore witnesses time-reversal asymmetry. The converse does not hold: for a constant Hessian $\cA\equiv0$ in any process. (ii)~For a fully scale-invariant loss, $\cA$ is invariant under a global rescaling of $\theta$ and is compatible with \cref{prop:symmetry}. (iii)~For gradient memory with decay, $\E\cA=Q\rhoH$ and
\begin{equation}\label{eq:A-decay}
\E\cA=\eta^2\DH+2\gamma\eta\,\E\big[\theta_t^\top(H_+-H)m_{t+1}\big]+\gamma^2\,\E\big[\theta_t^\top(H_+-H)\theta_t\big],
\end{equation}
with $H_+=H(\theta_{t+1})$. Hence $\DH$ alone is a pair-reversal observable only for $\gamma=0$. (iv)~Let $\pi$ be a law of consecutive pairs $(\theta_t,\theta_{t+1})$, such as the stationary law or the transition law averaged over a finite window. Let $\pi^{\mathrm R}(\theta,\theta')=\pi(\theta',\theta)$ and $I_2=D_{\mathrm{KL}}(\pi\,\|\,\pi^{\mathrm R})$. If $0<\E\cA^2<\infty$,
\begin{equation}\label{eq:I2}
I_2\ \ge\ \frac{2\,(\E\cA)^2}{\E\cA^2}.
\end{equation}
\end{proposition}
\begin{proof}
(i)~follows from exchanging $\theta$ and $\theta'$ in \eqref{eq:A}. (ii)~$H(\alpha\theta)=\alpha^{-2}H(\theta)$ for a scale-invariant loss, and the two displacement factors contribute $\alpha^2$. (iii)~Insert $\theta'-\theta=-\gamma\theta-\eta m_{t+1}$. (iv)~is proved in \cref{app:proofs} with the inequality $(x-y)\log(x/y)\ge2(x-y)^2/(x+y)$ and Cauchy--Schwarz.
\end{proof}
The bound \eqref{eq:I2} concerns pairs of parameter states. It needs no stationarity; for a window average it measures the pair asymmetry over the window. It is not the thermodynamic entropy production of the full optimizer. For that quantity, the memory variables, the reversal operation and the trajectory-length limit would have to be specified \citep{roldan2010estimating,roldan2012entropy}. Discretizations of reversible diffusions are in general irreversible themselves \citep{katsoulakis2014irreversibility}. What the proposition adds is the following. For $\gamma\to0$, the curvature transport that controls the memory correction to the FDR is the mean of this observable. The anharmonicity of the loss thus enters the memory dynamics through a time-reversal asymmetry. At finite decay the two differ by the terms of \eqref{eq:A-decay}. In the teacher--student runs, $\DH\approx-4\gamma\cK$, whereas $\E\cA/\eta^2\approx-(0.3\text{--}0.5)\gamma\cK$ for $\beta\ge0.9$. Because the transport is the mean of a pair observable, Gaussian noise does not by itself make it vanish. The relevant condition is the symmetry of the consecutive-state law, not of one fresh kick. For scale-invariant SGD with $0<\gamma<1$, the radial update constrains the support of the pair law (\cref{app:proofs}). If the radius of an invariant group changes with nonzero probability, the pair law is not absolutely continuous with respect to its reversal, and $I_2=\infty$. Dynamics with constant radius are not covered by this argument. Squaring a noisy estimate of $\E\cA$ biases \eqref{eq:I2} upwards by $2\,\mathrm{se}(\E\cA)^2/\E\cA^2$. This is the leading contribution of the noise in the numerator; estimating $\E\cA^2$ adds smaller covariance terms. We report it as the noise floor of the bound.
\section{Harmonic fluctuations and the spectral equation of state}\label{sec:linear}
We now linearize the angular dynamics around a stationary direction $n^*$ and solve the resulting models exactly. Write the tangent displacement along fixed tangent vectors $e_k$ (the Hessian eigenvectors on the unit sphere at $n^*$) as $x_k$. The corresponding angular noise $\zeta_k$ has variance $s_k$.

\subsection{Mode temperatures and the linear FDR}

\begin{proposition}[Mode temperatures]\label{prop:modes}
For a quadratic mode with $h_k>0$ and $R_k<1$ in the surrogate $x_{k,t+1}=(1-\eta_0h_k)x_{k,t}-\eta_0\zeta_{k,t}$ (SGD; \cref{approx:geo} for $\gamma>0$), and for its gradient-memory analogue with $\gamma\to0$, the stationary variance is
\begin{equation}\label{eq:modes}
\mathrm{Var}\,x_k=\frac{T_k}{h_k},\qquad T_k=h_k\mathrm{Var}\,x_k=\frac{T^{(0)}_k}{1-R_k},\qquad T^{(0)}_k=\frac{\eta_0s_k}{2}.
\end{equation}
\end{proposition}
$T^{(0)}_k$ is the temperature of mode $k$ in the continuous-time (SDE) limit. It is the same for all modes only for isotropic noise. The factor $(1-R_k)^{-1}$ is the finite-step correction. It breaks equipartition even for isotropic noise. The temperature also diverges as the mode approaches the stability threshold. For Gaussian additive noise, each stable harmonic mode has a Gaussian stationary law with the variance \eqref{eq:modes}. To first order in $\eta_0$ this modal law is the Gibbs law of the \emph{shadow loss} $L-\frac{\eta_0}{2\aeos}\|\nabla L\|^2$ at the temperature $T^{(0)}_k$. The correction in the shadow loss has the opposite sign to that of the modified loss $L+\frac\eta4\|\nabla L\|^2$ of backward error analysis \citep{barrett2021igr,smith2021origin}. The modified loss governs the mean drift rather than the fluctuations. A joint Gibbs law of the shadow loss at a single temperature requires, in addition, isotropic noise (equal $s_k$ and no covariance between modes). For non-Gaussian noise \eqref{eq:modes} still holds, but the variance does not determine the stationary law. For example, for $x_{t+1}=ax_t-\eta_0\zeta_t$ with $|a|<1$ and independent innovations, the stationary third cumulant is $-\eta_0^3\kappa_3(\zeta)/(1-a^3)$. It is nonzero for skewed noise. For SGD the factor $(1-\eta h/2)^{-1}$ appears in \citep{liu2021noise}. With heavy-ball normalization, $T^{(0)}_k$ is larger by $1/(1-\beta)$.

\begin{proposition}[Linear FDR and memory relations]\label{prop:fdr}
For a quadratic loss with constant $H$ and $\Sigma$ (not necessarily commuting), $\gamma\to0$, and all modes below the threshold,
\begin{equation}\label{eq:fdr}
\E\|g\|^2=\frac\eta2\Tr\Sigma H+\frac{\eta}{\aeos}\E[g^\top Hg]\qquad\text{(with and without memory)},
\end{equation}
and with memory
\begin{equation}\label{eq:ema-identities}
\E[g\cdot m]=-\kappa\,\E\|g\|^2,\qquad \E[g^\top Hm]=-\kappa\,\E[g^\top Hg],\qquad \E[m^\top Hm]=\kappa\Tr\Sigma H+\kappa^2\E[g^\top Hg].
\end{equation}
\end{proposition}
\begin{proof}
For SGD, multiply the Lyapunov equation $HC+CH-\eta HCH=\eta\Sigma$ by $H$ and take the trace. In the eigenbasis of $H$, all traces involve only diagonal entries of $C$. These depend only on the diagonal of $\Sigma$. It therefore suffices to solve the $2\times2$ Lyapunov equation of one memory mode (\cref{app:lyapunov}).
\end{proof}
The continuous-time (ideal-gas) value of $\E\|g\|^2$ is $\frac\eta2\Tr\Sigma H$. The second term of \eqref{eq:fdr} is the edge term. It is significant for SGD and negligible with memory at $\beta=0.9$ unless $R\to1$. With memory and small $R$, $\frac\eta2\Tr\Sigma H$ consists of two parts. The fresh noise in the current update gives $\frac\eta2\kappa\Tr\Sigma H$. The correlation between the current gradient and the noise stored in the memory gives $\frac{\eta\beta}{1+\beta}\Tr\Sigma H$. For $\beta=0.9$ this is $2\beta/(1+\beta)=95\%$ of the total. This second part relies on the linear response of the gradient to the steps taken while the noise is stored. This is where curvature transport enters (\cref{sec:closure}).

\subsection{Weight decay in the angular dynamics: a geometric mean-field approximation}\label{sec:geo}
Weight decay enters the tangent dynamics twice, in addition to the radial effects discussed below. First, renormalization contracts every $x_k$ by $(1-\gamma)$ per step on average. Second, the curvature of the sphere adds $-\omega(x)\,x_k$ to the gradient component $g\cdot e_k$. Here $\omega(x)=\sum_jh_jx_j^2$ (over the coordinates of one group) is a fluctuating quantity. Indeed, with $n=\sqrt{1-|x|^2}\,n^*+\sum_jx_je_j$ and $g\perp n$, one has $g\cdot e_k=\partial_{x_k}L+x_k\,(g\cdot n^*)/(n^*\!\cdot n)$. In addition, $g_B\cdot n^*\approx-\omega(x)-\sum_jx_j\zeta_j$.

\begin{approximation}[Geometric mean field]\label{approx:geo}
Replace $\omega(x)x_k$ by $\E[\omega]\,x_k$ and the radius ratio $r_t/r_{t+1}$ by one. Stationarity of $\E[n\cdot n^*]$ then fixes $\eta_0\E\omega=\gamma+O(\gamma^2,\gamma\E|x|^2)$. This holds with and without memory, because the memory averages past gradients and $\E[m\cdot n^*]=\E[g\cdot n^*]=-\E\omega$ in a stationary state (\cref{app:proofs}). For SGD the resulting surrogate is
\begin{equation}\label{eq:geo-sgd}
x_{k,t+1}=(1-\gamma)x_{k,t}-\eta_0\big((h_k-\E\omega)x_{k,t}+\zeta_{k,t}\big)=(1-\eta_0h_k)x_{k,t}-\eta_0\zeta_{k,t},
\end{equation}
in which the two terms cancel. This cancellation holds within the surrogate. Beyond the surrogate, the radial step adds the $O(\gamma)$ effect \eqref{eq:c-sgd}. With memory the surrogate is
\begin{equation}\label{eq:geo-ema}
m_{t+1}=\beta m_t+(1-\beta)\big((h_k-\E\omega)x_t+\zeta_t\big),\qquad x_{t+1}=(1-\gamma)x_t-\eta_0m_{t+1},
\end{equation}
in which the two terms do not cancel. Here the restoring force acts through the memory, while the contraction acts directly.
\end{approximation}

\begin{lemma}[Exact covariance of the surrogate]\label{lem:geo}
The stationary variance of \eqref{eq:geo-ema} with $\E\omega=\gamma/\eta_0$ is
\begin{equation}\label{eq:modes-q}
\mathrm{Var}\,x_k=\frac{T^{q}_k}{h_k(1-R^q_k)},\qquad T^q_k=\frac{1-\beta}{1-q}\,\frac{\eta_0s_k}{2},\qquad R^q_k=\frac{\eta_0h_k(1-\beta)}{2(1+q)}.
\end{equation}
\end{lemma}
The proof is the discrete Lyapunov solution (\cref{app:lyapunov}; \ncheck{lyap}). Weight decay thus lowers the temperature factor $T^q_k$ of every mode by $(1-\beta)/(1-q)=1/(1+\beta\gamma/(1-\beta))$. It also lowers the threshold from $2(1+\beta)/(1-\beta)$ to $2(1+q)/(1-\beta)$, so that $R^q_k=R_k(1+\beta)/(1+q)>R_k$. The variance itself changes by the factor $[(1-\beta)/(1-q)]\,(1-R_k)/(1-R^q_k)$. It decreases far from the threshold but increases for $R_k>(1+q)/2$, as long as the mode remains stable, $R_k<(1+q)/(1+\beta)$. For $\beta=0.9$, $\gamma=10^{-3}$ and $R_k=0.5$ the factor is $0.992$. The same $q$ appears in \eqref{eq:radial-ema}. As $\gamma\to0$, \eqref{eq:modes-q} reduces to \eqref{eq:modes}. The factor is the part of the decay effect that is enhanced by the memory time $1/(1-\beta)$.

The surrogate does not determine the remaining $O(\gamma)$ terms. They come from the radial parts of the step and of the memory, which the tangent description drops. They enter the exact relation \eqref{eq:F2} through $U$ and through geometric contributions to the remainders. In the memory-free limit the radial part separates exactly. For SGD on a fully scale-invariant network and $0\le\gamma<1$, $w_{t+1}=(1-\gamma)(w_t+\Delta_{\tan})$ with the tangent displacement $\Delta_{\tan}=-\eta g_B/(1-\gamma)$. Hence $L(w_{t+1})=L(w_t+\Delta_{\tan})$. Let $\Rem_{\tan}$ be the remainder of the quadratic model along $\Delta_{\tan}$, with normalization $\rho_{\tan}=\E\Rem_{\tan}/Q_{\tan}$, $Q_{\tan}=\frac12\E[\Delta_{\tan}^\top H\Delta_{\tan}]$. The stationary loss balance then gives exactly
\begin{equation}\label{eq:c-sgd}
c=\frac{1+\rho_{\tan}}{1-\gamma}\qquad(\beta=0),
\end{equation}
which complements $c=(1+\rho_0)/[1+\gamma(1+\rho_0)]$. In the latter, the full-step remainder $\rho_0$ contains the radial contribution $2\gamma+O(\gamma^2)$ (\cref{app:lyapunov}). The radial factor $1/(1-\gamma)$ raises $c$. The tangent remainder is of the same order at stationarity, because the norm balance makes the squared angular step $\approx2\gamma$. Its sign is set by the anharmonicity of the loss along the step. The total change of $c$ is therefore not universally positive. With memory, the radial memory component is correlated with the step length. For this case we have no closed form. In exact sphere simulations and in the teacher--student network (\cref{sec:sphere,sec:ts}), $\rho_{\tan}$ vanishes to this order for a quadratic angular loss and is negative for the network. With memory, the terms beyond the per-mode factor are of order $\gamma$ in these simulations. Their coefficient depends on the noise statistics through the correlation between the radial memory and the step length. Its magnitude decreases with the number of tangent dimensions. At long memory it is several times smaller than the effect of the factor itself. The per-mode factor therefore gives the leading behaviour for long memory. The exact relation \eqref{eq:F2} is the consistent route to the $O(\gamma)$ terms. The reason is that the per-mode FDR of the surrogate uses $g=hx$ and misses the radial parts of the gradient and memory. This includes the radial part of the memory in $\E[g^\top Hm]$ (\cref{app:lyapunov}).

The approximation is not an exact statement about the nonlinear sphere dynamics. Its first omitted correlation is visible already for a centred Gaussian local state with covariance $C$. For this state, $\E[(x^\top Hx)xx^\top]=\Tr(HC)\,C+2CHC$. The second term is what replacing $\omega(x)$ by its mean drops. Radial fluctuations enter at the same order. In sphere models, \eqref{eq:modes-q} nevertheless gives every mode variance to within a few per cent where weight decay lowers the temperatures by a third (\cref{sec:sphere}). In the networks of \cref{sec:networks} the factor matters only at the largest decay rates.

\subsection{Coupled groups: a covariance matrix, not independent gases}\label{sec:coupled}
Independent scale invariances identify norms and directions on a product of spheres. However, they do not make the groups statistically independent. Collect the tangent coordinates of all groups in $x$ and freeze the radii over a measurement window. Let $\Lambda=\mathrm{diag}(\eta/r_1^2\,I,\dots,\eta/r_{\Ng}^2\,I)$ be the block-diagonal angular step. The angular Hessian $\sH$ on the product of spheres keeps its cross-group blocks, and so does the angular noise covariance $\sS$.

\begin{proposition}[Coupled-group covariance and reversibility]\label{prop:coupled}
Consider the SGD surrogate $x_+=(I-\Lambda\sH)x-\Lambda\zeta$, $\mathrm{Cov}\,\zeta=\sS$. In this surrogate the contraction and the curvature shift of every group cancel as in \eqref{eq:geo-sgd}. (i)~The stationary covariance satisfies
\begin{equation}\label{eq:M1}
C=(I-\Lambda\sH)\,C\,(I-\sH\Lambda)+\Lambda\sS\Lambda.
\end{equation}
With $y=\Lambda^{-1/2}x$, $\tilde\sH=\Lambda^{1/2}\sH\Lambda^{1/2}$ and $\tilde\sS=\Lambda^{1/2}\sS\Lambda^{1/2}$, in an eigenbasis of $\tilde\sH$ with stable eigenvalues $0<\nu_i<2$,
\begin{equation}\label{eq:M2}
(C_y)_{ij}=\frac{\tilde\sS_{ij}}{\nu_i+\nu_j-\nu_i\nu_j}.
\end{equation}
(ii)~For Gaussian noise and nondegenerate $C$, the consecutive-state law of $y$ is reversible if and only if $[\tilde\sH,\tilde\sS]=0$. The lag-antisymmetric covariance $J_1=\E[y_+y^\top-yy_+^\top]=(I-\tilde\sH)C_y-C_y(I-\tilde\sH)$ measures the stationary probability current. (iii)~In particular, two coupled scalar groups with steps $\alpha_1\ne\alpha_2$, diagonal noise $s_1,s_2$ and nonzero Hessian coupling are reversible only if $\alpha_1s_1=\alpha_2s_2$, i.e.\ if their local noise temperatures agree. Hessian-aligned noise, $\sS\propto\sH$, gives $\tilde\sS\propto\tilde\sH$ and a reversible surrogate for \emph{any} radii.
\end{proposition}
\begin{proof}
(i)~is the discrete Lyapunov equation of the surrogate. Transforming to $y$ gives $C_y=(I-\tilde\sH)C_y(I-\tilde\sH)+\tilde\sS$. Its solution in the eigenbasis of the symmetric $\tilde\sH$ is \eqref{eq:M2}. (ii)~A stationary Gaussian autoregression $y_+=My+\zeta'$ is reversible iff $\E[y_+y^\top]=MC_y$ is symmetric. With $M=I-\tilde\sH$ symmetric, this is $[\tilde\sH,C_y]=0$. Then $\tilde\sS=(2\tilde\sH-\tilde\sH^2)C_y$ commutes with $\tilde\sH$. Conversely, $[\tilde\sH,\tilde\sS]=0$ makes \eqref{eq:M2} diagonal in a common eigenbasis. (iii)~follows by writing out $[\tilde\sH,\tilde\sS]$.
\end{proof}
The stationary discrete covariance of a single group is known \citep{liu2021noise}. The content here is that unequal radii and correlated groups do not admit independent scalar temperatures or a common angular learning rate. It is also that exact norm balances do not imply an equilibrium mixture. Coupled groups can carry stationary probability currents. Reversibility alone does not imply a common mode temperature: in the coordinates $y$, a common temperature, $C_y=T\tilde\sH^{-1}$, requires $\tilde\sS=T(2I-\tilde\sH)$, which is stronger than $[\tilde\sH,\tilde\sS]=0$. At the teacher of a realizable teacher--student model the noise is Hessian-aligned. There the surrogate is reversible even with unequal radii. Away from the teacher, or for an unrealizable teacher, the alignment is only approximate (\cref{sec:ts}). With memory the same construction applies in phase space $(x,m)$ with the shifts of \cref{approx:geo} taken per group. There we treat $x$ and $m$ as even variables and report the lag-antisymmetric covariance in these coordinates. A reversal with odd memory would be a different test, and neither is the entropy production of the optimizer.

\subsection{Self-consistent radius and the spectral equation of state}\label{sec:eos}
\paragraph{Self-consistent radius.} In this subsection we consider one group and write $d$ for its dimension. For isotropic noise with angular variance $\sigma_{\mathrm a}^2$ per tangent direction and SGD, \eqref{eq:radial-sgd} with \eqref{eq:modes} gives the self-consistency condition
\begin{equation}\label{eq:thermostat}
\frac{2\gamma}{\sigma_{\mathrm a}^2\eta_0^2}=\sum_{k=1}^{d-1}\frac{1}{1-R_k},\qquad\text{i.e.}\qquad r^4=r^4_0\Big\langle\frac{1}{1-R}\Big\rangle,
\end{equation}
to leading order in $\gamma$. The radius $r$ here is a representative radius. The exact balance fixes $\E r^2$, and evaluating averages that contain $r^{-2}$ at this radius is part of the approximation. With gradient memory the same condition follows from \eqref{eq:radial-ema} at leading order in $\gamma$, because the factors $\kappa$ and $1/\kappa$ cancel; this requires $\gamma/(1-\beta)$ to be small. Here $\langle\cdot\rangle$ averages over the $d-1$ tangent modes, and $r_0^4=\eta^2\sigma_{\mathrm a}^2(d-1)/(2\gamma)$ is the ideal-gas radius at fixed angular noise. Since $R_k\propto\eta_0\propto r^{-2}$, the radius adjusts until \eqref{eq:thermostat} holds. The ratio $r^4/r_0^4$ is the mean normalized mode temperature $\langle(1-R)^{-1}\rangle$, so that $r^2/r_0^2=\langle(1-R)^{-1}\rangle^{1/2}$. It is dominated by the modes near the threshold. The right-hand side diverges as $\Rtop\to1$. Within this self-consistent harmonic model, at leading order in $\gamma$, the stationary $\Rtop$ therefore stays below 1 for every $\gamma$. With anisotropic noise we use \eqref{eq:radial-sgd} with $\E\|g\|^2$ from \eqref{eq:fdr} and the measured noise traces. Only $\E\|g\|^2$ is then predicted.

\paragraph{Spectral virial law.} Define $V=r^2/2$, the pressure $p=\lambda$, $R_{\mathrm{gas}}=(d-1)/2$, the angular temperature $T_{\mathrm{ang}}=\eta\sigma_{\mathrm a}^2/(2r^2)$ and the \emph{excluded volumes} $b_k=\eta h_k/(2\aeos)$. Then $R_k=b_k/V$, and \eqref{eq:thermostat} reads
\begin{equation}\label{eq:virial}
pV=R_{\mathrm{gas}}T_{\mathrm{ang}}\,\Phi(V),\qquad \Phi(V)=\Big\langle\frac{1}{1-b/V}\Big\rangle .
\end{equation}
For equal curvatures this is $p(V-b)=R_{\mathrm{gas}}T_{\mathrm{ang}}$, the excluded-volume part of the van der Waals law. With the nominal temperature $T_{\mathrm{nom}}=\sqrt{\eta\lambda\sigma_{\mathrm a}^2/(2(d-1))}$ of \citep{sadrtdinov2025idealgas},
\[
Z=\frac{pV}{R_{\mathrm{gas}}T_{\mathrm{nom}}}=\Phi^{1/2}=1+\frac{B_2}{V}+\frac{B_3}{V^2}+\dots,\qquad B_2=\frac{\langle b\rangle}2,\qquad B_3=\frac{\langle b^2\rangle}2-\frac{\langle b\rangle^2}8\ \ge\ \frac32B_2^2,
\]
with equality for equal curvatures. The virial coefficients are thus moments of the curvature spectrum rather than fitted constants. There is no attractive term. At fixed $\eta$, $p(V)=\frac{(d-1)\eta\sigma_{\mathrm a}^2}{8}\langle[V(V-b)]^{-1}\rangle$ decreases monotonically on the stable branch $V>b_{\max}$. Hence there is no loop and no coexistence. Since $b_k\propto\eta$, these coefficients are not material constants. Under the rescaling of \cref{prop:symmetry}, $b_k\to\alpha b_k$, and the coefficients scale with the corresponding powers, $B_n\to\alpha^{n-1}B_n$. The law follows from the harmonic approximation at leading order in $\gamma$. It is not a finite-decay identity. The two radius baselines of \eqref{eq:baselines} give different ratios in this model: $r^2/r_0^2=\Phi^{1/2}$, whereas $r^2/r^2_{\mathrm{SDE}}=\Phi$. The reason is that the measured noise trace itself scales as $r^{-2}$. In networks with anisotropic noise the moments are noise-weighted. To first order, $r^2/r^2_{\mathrm{SDE}}-1=\frac\eta2\Tr\Sigma H/\Tr\Sigma=\langle b\rangle_\Sigma/V$.

\paragraph{Spectral entropy and stored-gradient information.} By \cref{cor:maxwell}, for corresponding stationary laws the shell entropy of SGD and gradient memory has the form $\mathcal S=\frac D2\log\eta+F(\log\gamma)$, with $D=\sum_a(d_a-1)$ and a function $F$ that the symmetry leaves undetermined. An estimator bias that depends on $\gamma$ alone also cancels in \eqref{eq:maxwell}. Agreement with \eqref{eq:maxwell} therefore does not test the $\gamma$-dependence of an entropy estimate. Starting from the covariance \eqref{eq:modes} \citep{liu2021noise}, the harmonic theory identifies a spectral contribution to this dependence.

\begin{proposition}[Spectral entropy and stored-gradient information]\label{prop:spec-entropy}
Consider the zero-decay harmonic surrogate of \cref{prop:modes} with i.i.d.\ centred Gaussian additive noise, for $D$ independent modes (angular curvature and noise covariance diagonal in a common basis) with $h_k,s_k>0$ and $0<R_k<1$. Let $C$ be the stationary covariance of the modes and $C^{(0)}$ the continuous-time covariance, $C^{(0)}_k=\eta_0s_k/(2h_k)$. Then:
(i)~the finite-step entropy excess is
\begin{equation}\label{eq:spec-entropy}
\Delta S_{\mathrm{spec}}=\frac12\log\frac{\det C}{\det C^{(0)}}=-\frac12\sum_k\log(1-R_k)=\frac12\sum_{n\ge1}\frac1n\Big(\frac{\eta_0}{\aeos}\Big)^n\sum_kh_k^n ;
\end{equation}
(ii)~for one group with isotropic noise ($D=d-1$), the function $\Phi$ of \eqref{eq:virial} obeys, at fixed $\eta$, angular spectrum and $\beta$,
\begin{equation}\label{eq:spec-eos}
\Phi=1+\frac2D\,\frac{\partial\Delta S_{\mathrm{spec}}}{\partial\log\eta_0}=1-\frac{2V}{D}\,\frac{\partial\Delta S_{\mathrm{spec}}}{\partial V};
\end{equation}
(iii)~with the memory $m_t$ before the update, at equal times every mode has $\mathrm{Corr}(x_k,m_k)^2=R_k$ and $\mathrm{Var}(x_k\,|\,m_k)=C^{(0)}_k$. Hence
\begin{equation}\label{eq:spec-info}
I(x_t;m_t)=\Delta S_{\mathrm{spec}} .
\end{equation}
\end{proposition}
\begin{proof}
(i) follows from \eqref{eq:modes}; the noise amplitudes cancel in the ratio. (ii) $R_k\propto\eta_0\propto V^{-1}$ at fixed $\eta$, and $\partial_{\log\eta_0}[-\frac12\log(1-R_k)]=\frac12[(1-R_k)^{-1}-1]$. (iii) In one mode, $g_k=h_kx_k$, and \eqref{eq:ema-identities} gives $\mathrm{Cov}(x_k,m_k)=-\kappa h_k\mathrm{Var}\,x_k$ and $\mathrm{Var}\,m_k=\kappa s_k+\kappa^2h_k^2\mathrm{Var}\,x_k=\kappa s_k/(1-R_k)$, using $R_k=\kappa\eta_0h_k/2$. Therefore $\mathrm{Corr}^2=\kappa\eta_0h_k/2=R_k$, and the Gaussian conditional variance is $(1-R_k)\mathrm{Var}\,x_k=C^{(0)}_k$. The mutual information of a Gaussian pair is $-\frac12\log(1-\mathrm{Corr}^2)$, and independent modes add.
\end{proof}
The finite-step Gaussian entropy excess is therefore equal to the mutual information between the current position and the pre-update stored gradient: conditioning on $m_t$ reduces every mode to its continuous-time reference variance. For $\beta=0$, $m_t$ is the preceding stochastic gradient, so the identity does not single out persistent momentum. It does not extend unchanged to non-commuting curvature and noise, to non-Gaussian noise with the same covariance (which preserves these second moments but need not preserve the entropies) or to finite decay, where the surrogate \eqref{eq:modes-q} gives an entropy excess different from $I(x;m)$. The position process of independent Gaussian modes is reversible. Neither $\Delta S_{\mathrm{spec}}$ nor $I(x;m)$ therefore measures irreversibility; the pair observable of \cref{sec:antisym} provides a separate diagnostic of time asymmetry.

\Cref{eq:spec-eos} is a spectral derivative identity at fixed angular spectrum, $\eta$ and $\beta$. It does not establish thermodynamic conjugacy between $p$ and $V$: the equation of state \eqref{eq:virial} is derived from the stationary norm balance, without a Gibbs assumption. For one group with isotropic angular noise $\sigma_{\mathrm a}^2$, the local Gaussian approximation gives $\mathcal S\simeq\frac D2\log(\pi e\eta\sigma_{\mathrm a}^2)-\frac12\sum_k\log h_k+\Delta S_{\mathrm{spec}}$. Within the harmonic approximation, therefore, $F=\frac D2\log(\pi e\sigma_{\mathrm a}^2)-\frac12\sum_k\log h_k+\Delta S_{\mathrm{spec}}$. The noise strength and the angular spectrum may themselves depend on $\gamma$. With both held fixed, the self-consistency condition \eqref{eq:thermostat} determines $\eta_0$ as a function of $\gamma$ and hence predicts the spectral contribution to $F$. The spectral term is the part of $F$ produced by the finite step.

\paragraph{AdamW.} For isotropic noise and a random basis, Adam's second moment is nearly uniform, $\hat v\approx\bar s/r^2$ with $\bar s=\frac{\sigma_{\mathrm a}^2}{d}\sum_k[1+\frac\aeos2\frac{R_k}{1-R_k}]$. In this regime AdamW acts as gradient memory with angular step $\eta/(r\sqrt{\bar s})$. Then
\begin{equation}\label{eq:adam-radius}
r^2=r^2_{\mathrm{A}}\,\frac{\langle(1-R)^{-1}\rangle}{\langle1+\frac\aeos2\frac{R}{1-R}\rangle},\qquad r^2_{\mathrm{A}}=\frac{\eta d}{2\lambda}.
\end{equation}
The correction has the \emph{opposite sign} to SGD, with coefficient $\aeos/2-1=2\beta/(1-\beta)=18$. The gradient signal inflates the second moment, which shrinks the effective step and the radius.

\paragraph{Concentration at the threshold.} Whether the excess temperature then concentrates in one mode depends on the spectrum near its upper edge. Suppose that the top curvature $h_*$ is isolated from a bulk with limiting density $\varrho$. Suppose further that the angular noise scales as $\sigma_{\mathrm a}^2=\sigma_0^2/D$ with the number $D$ of tangent modes. Then the harmonic self-consistency \eqref{eq:thermostat} has a transition at $\gamma_c=\sigma_0^2\aeos^2I_c/(2h_*^2)$, with $I_c=\int\varrho(h)(1-h/h_*)^{-1}dh$. The transition is analogous to Bose--Einstein condensation. For $\gamma>\gamma_c$ the top mode carries the share $1-\gamma_c/\gamma$ of $\sum_k(1-R_k)^{-1}$ (\cref{prop:condensation} in \cref{app:proofs}). A remark there treats a bulk extending to $h_*$. The transition concerns the harmonic model only. The macroscopic top-mode fluctuation it predicts can invalidate the quadratic approximation. The predicted transition is therefore not by itself evidence for condensation in a network. In a sphere model with an isolated top mode the simulated top-mode temperature follows this prediction across the transition (\cref{sec:sphere}).
\section{Beyond the harmonic approximation}\label{sec:remainder}

\subsection{Smooth nonlinearity and activation boundaries}\label{sec:kinks}
The remainder profile $\rho(f)$ distinguishes two mechanisms. For a smooth loss the remainder is of third order in the step, $\Rem(f)\approx f^3\,\nabla^3L[\Delta\theta,\Delta\theta,\Delta\theta]/6$. Therefore $\rho(f)\propto f$ unless the cubic coefficient vanishes by symmetry. A ReLU network is piecewise smooth. Units whose pre-activation changes sign along the step contribute curvature that the regular Hessian does not contain.

\begin{lemma}[Expected Taylor remainder of an activation]\label{lem:kink}
Let $z$ be a pre-activation and $\delta$ its change along the step. Let the conditional density $p=p(\cdot\,|\,\delta)$ of $z$ given $\delta$ be twice continuously differentiable, with $p$, $p'$ and $p''$ bounded and integrable. Let $\phi$ grow at most linearly, with $\phi''$ a finite measure (for ReLU the Dirac mass at $0$), so that the integrals below exist and the boundary terms of the integrations by parts vanish. Let $F(s)=\int p(z)\,\phi(z+s)\,dz$. Then $F''(0)=\int p\,\phi''$ in the sense of distributions, and
\begin{equation}\label{eq:kink}
\E\big[\phi(z+\delta)-\phi(z)-\delta\phi'(z)\,\big|\,\delta\big]=F(\delta)-F(0)-\delta F'(0)=\frac{\delta^2}{2}\int p\,\phi''\,dz-\frac{\delta^3}{6}\int p'\,\phi''\,dz+O(\delta^4).
\end{equation}
For ReLU, $\int p\phi''=p(0)$, while the regular second derivative vanishes almost everywhere. The quadratic model therefore misses $\frac{\delta^2}2p(0)$, a remainder $\propto f^2$. For a smooth activation $\phi_\tau$ of any width $\tau>0$ (e.g.\ the shifted Softplus $\phi_\tau(z)=\tau\log(1+e^{z/\tau})-\tau\log2$), the regular second-order term has conditional expectation exactly $\frac{\delta^2}2\int p\phi_\tau''$. The averaged quadratic term is then correct through second order. Since $\phi_\tau''\ge0$ has unit integral, the remainder after it obeys $|F(\delta)-F(0)-\delta F'(0)-\frac{\delta^2}2F''(0)|\le|\delta|^3\|p'\|_\infty/6$ uniformly in $\tau$.
\end{lemma}
\begin{proof}
$F(s)=\int p(y-s)\phi(y)\,dy$, so $F^{(n)}(0)=(-1)^n\int p^{(n)}\phi=(-1)^n\int p^{(n-2)}\phi''$ after two integrations by parts, and $F'(0)=\int p\,\phi'$.
\end{proof}
The cubic coefficient $-\frac16\int p'\phi_\tau''$ depends on $\tau$ in general. When the pre-activation density varies slowly across the smoothing region, it approaches the width-independent limit $-\frac16p'(0)$. Two consequences are testable. (a)~In the ensemble regime considered here, ReLU boundary crossings can give $\rho(f)$ a nonzero plateau at small $f$. Its value, the singular curvature along the step relative to $\Delta\theta^\top H\Delta\theta$, can have either sign or vanish. (b)~Replacing ReLU by Softplus removes this plateau at \emph{any} $\tau$, including $\tau$ much smaller than the step in pre-activation units. The $f^2$ law is an ensemble statement. For a fixed finite batch and a small enough step, no unit crosses its threshold. For a finite empirical loss the boundary term is likewise understood distributionally: at a generic point, the ordinary Hessian of a finite average of piecewise-smooth losses equals its regular Hessian. The network remainder measurements capture the ensemble-averaged effect over finite displacements, and the $f^2$ regime need not persist to arbitrarily small $f$ for a fixed finite set of states and examples. In the Gaussian-input model of \cref{sec:relu}, input averaging gives an ordinary population-Hessian difference. The lemma also shows why, for ReLU, $\rhoH$ is not the second derivative of the mean remainder. Already for $F(s)=\E\,\mathrm{ReLU}(z+s)$, $F''(s)=p(-s)$, while the regular samplewise second derivative vanishes almost everywhere. Of the relations derived from a single power law of $\E\Rem(f)$ (\cref{sec:closure}), $\rho_1=n\rho_0$ therefore holds for kink remainders too, since $\frac{d}{df}\E\Rem=\E\Rem'$. However, $\rhoH=n(n-1)\rho_0$ holds only for smooth remainders.

\subsection{Singular curvature of normalized ReLU committees}\label{sec:relu}
For Gaussian inputs the singular curvature can be computed in closed form. Let $x\sim\mathcal N(0,I_{N_{\mathrm{in}}})$. The student is $f_w(x)=\sum_jo_j\,\mathrm{ReLU}(n_j\!\cdot x)$ with $n_j=w_j/r_j$, and the teacher is $f_*(x)=\sum_lo^*_l\,\mathrm{ReLU}(v_l\!\cdot x)$ with unit $v_l$. The loss is the half-squared error, with independent zero-mean label noise. This noise does not change the following. Here $\Hsing=\Hpop-\Hreg$ is the population Hessian minus the expected regular samplewise Hessian.

\begin{proposition}[Singular curvature of a ReLU committee]\label{prop:relu}
Let $P_j=I-n_jn_j^\top$, $\varsigma_j(v)=\sqrt{1-(n_j\!\cdot v)^2}$ and $F_j(v)=\varsigma_j(v)P_j+(P_jv)(P_jv)^\top/\varsigma_j(v)$, with the continuous limit $F_j=0$ at $\varsigma_j=0$. Then $\Hsing$ is block diagonal in the raw hidden weights, with
\begin{equation}\label{eq:R2}
(\Hsing)_{jj}=\frac{o_j}{2\pi r_j^2}\Big[\sum_{i\ne j}o_iF_j(n_i)-\sum_lo^*_lF_j(v_l)\Big].
\end{equation}
Consider a single unit with $n\cdot v=\cos\theta$, $0<\theta<\pi$, and let $e=(v-\cos\theta\,n)/\sin\theta$. A tangent perturbation $t\perp n$ then has
\begin{equation}\label{eq:R1}
t^\top\Hsing t=-\frac{oo_*\sin\theta}{2\pi r^2}\big[\|t\|^2+(t\cdot e)^2\big],
\end{equation}
negative semidefinite for $oo_*>0$ and vanishing at alignment and anti-alignment.
\end{proposition}
The proof (\cref{app:proofs}) conditions on the hyperplane $n_j\cdot x=0$, on which the student unit vanishes. Off-diagonal singular blocks vanish because each gate depends only on its own hidden weights. Ordinary cross-block curvature remains. Student--student and student--teacher terms compete, so a committee has no universal sign. At an exactly matched teacher they cancel. The single-unit case \eqref{eq:R1} shows that batch normalization is not necessary for negative singular curvature. It does not imply that the singular curvature of an arbitrary network is negative.

We evaluate \eqref{eq:R2} at fixed states with nonzero residual. These are committees with $N_{\mathrm s}=1,3,4$ units, unequal radii and mixed-sign output weights. The reference is the population Hessian of the arc-cosine kernel \citep{cho2009kernel} minus a Monte Carlo estimate of the expected regular Hessian ($4\cdot10^6$ inputs; \ncheck{relu}). The two agree to within $4\cdot10^{-3}$ on entries of size 0.47--0.78, which is the Monte Carlo error. The off-diagonal singular blocks vanish to the same accuracy. The committee eigenvalues of $\Hsing$ have both signs ($-0.95$ to $+0.07$).

\subsection{Curvature transport and smooth-remainder approximations}\label{sec:closure}
By \eqref{eq:F2} and \eqref{eq:transport}, what the single-step remainder does not explain enters in two ways. It enters through $U$, i.e.\ through the curvature transport $\DH$ amplified by $\beta^2/(1-\beta^2)$. It also enters through the mixed term. To predict $c$ from the single-step remainder alone, the endpoint quantities $\rho_1$ and $\rhoH$ must be expressed through $\rho_0$ (a moment closure). Suppose that the mean remainder along the actual step is a single power near $f=1$, $\E\Rem(f)\approx a_nf^n$. Then $\rho_1=n\rho_0$ and $\rhoH=n(n-1)\rho_0$. The relation for $\rhoH$ holds for smooth remainders only (\cref{sec:kinks}). With $\delta_X\approx0$ and $\gamma\to0$ it follows that
\begin{equation}\label{eq:closure}
c\approx\frac{1+\frac{1+(n-1)\beta}{1+\beta}\,\rho_0}{1-\frac{n(n-1)\beta^2}{2(1-\beta^2)}\,\rho_0}.
\end{equation}
We call the case $n=3$ the \emph{cubic-remainder approximation}. To first order in $\rho_0$ it gives
\begin{equation}\label{eq:amplification}
c\approx1+a_{\mathrm{mem}}(\beta)\,\rho_0,\qquad a_{\mathrm{mem}}(\beta)=\frac{1+2\beta}{1+\beta}+\frac{3\beta^2}{1-\beta^2}=\frac{1+\beta+\beta^2}{1-\beta^2}\ \xrightarrow{\ \beta\to1\ }\ \frac{3}{2(1-\beta)},
\end{equation}
with $a_{\mathrm{mem}}=14.3$ for $\beta=0.9$ (2.3, 29.3 and 149 for $\beta=0.5$, 0.95 and 0.99). Within this approximation, gradient memory therefore amplifies the single-step anharmonicity by an order of magnitude. The nonlinear form matters already at $\rho_0=-0.02$, where it gives $c=0.773$ against 0.715 from \eqref{eq:amplification}.

\paragraph{A one-parameter extension.} The cubic approximation fixes both endpoint ratios by $\rho_0$. A remainder $af^3+bf^4$ along the step relaxes this. It gives $\rho_1/\rho_0=3+\varphi$ and $\rhoH/\rho_0=6+6\varphi$, with the effective quartic fraction $\varphi=b/(a+b)$. For every $\varphi$ the two ratios obey
\begin{equation}\label{eq:quartic}
\rhoH=6\rho_1-12\rho_0 .
\end{equation}
For fully scale-invariant networks and a measured $\rho_1$, inserting \eqref{eq:quartic} into the exact stationary $\gamma\to0$ relation from \eqref{eq:decomp} and \eqref{eq:transport} gives
\begin{equation}\label{eq:quartic-c}
c\simeq(1+\delta_X)\,\frac{1+[(1-\beta)\rho_0+\beta\rho_1]/(1+\beta)}{1-\frac{\beta^2}{2(1-\beta^2)}\,(6\rho_1-12\rho_0)} .
\end{equation}
The approximation lies in \eqref{eq:quartic} and in neglecting finite decay and drift. In the Softplus networks \eqref{eq:quartic} holds on both datasets to within 1\% of $\rhoH$ (\cref{sec:momentum}, \cref{app:cubic}). The cubic relation $\rhoH\approx6\rho_0$ holds approximately on CIFAR-10, where $\varphi\approx-0.1$. It fails on CIFAR-100, where $\varphi\approx0.4$. The general transport identity is unaffected by this difference; only the local approximation changes.

\paragraph{Conditions of validity.} The error of the approximation can be stated in measurable quantities. Define $\Delta_1=\rho_1-3\rho_0$, $\Delta_2=\rhoH-6\rho_0$, $N_0=1+\frac{1+2\beta}{1+\beta}\rho_0$, $D_0=1-\frac{3\beta^2}{1-\beta^2}\rho_0$, $a_1=\beta/(1+\beta)$, $b_2=\beta^2/(2(1-\beta^2))$ and $c_{\mathrm{cub}}=N_0/D_0$. The exact stationary $\gamma\to0$ relation is $c=(1+\delta_X)\big[1+\frac{(1-\beta)\rho_0+\beta\rho_1}{1+\beta}\big]/\big[1-b_2\rhoH\big]$. Whenever $D_0\ne0$ and $|b_2\Delta_2|\le|D_0|/2$, it implies
\begin{equation}\label{eq:C1}
|c-c_{\mathrm{cub}}|\le|\delta_X|\,|c_{\mathrm{cub}}|+2(1+|\delta_X|)\Big[\frac{a_1|\Delta_1|}{|D_0|}+\frac{|N_0|\,b_2|\Delta_2|}{D_0^2}\Big]
\end{equation}
(checked on $2\cdot10^5$ random parameter sets). To see this, write $N=N_0+a_1\Delta_1$ and $D=D_0-b_2\Delta_2$, so that $c=(1+\delta_X)N/D$. Then $c-c_{\mathrm{cub}}=\delta_Xc_{\mathrm{cub}}+(1+\delta_X)(N/D-N_0/D_0)$, with $N/D-N_0/D_0=(a_1\Delta_1D_0+N_0b_2\Delta_2)/(DD_0)$ and $|D|\ge|D_0|/2$. Long memory increases the sensitivity to an inaccurate endpoint curvature change through $b_2$. Finite-decay and finite-window terms must be added through \cref{thm:finite}, not hidden in the error. The approximation does not predict the remainder from the Hessian. Additional information about the landscape is unavoidable. The transport identity does not require Gaussian noise, a fixed Hessian or free displacement. The reduction to the single-step remainder additionally requires control of the endpoint derivatives of the mean remainder and of the mixed term. A one-dimensional toy (\cref{app:toy}) shows both regimes. With skewed kicks and a tilted log-cosh potential, whose expansion about its minimum has a cubic term, soft modes ($\eta h\lesssim1-\beta$) have $\rho_1/\rho_0\approx2.8$--$2.9$ and $\rhoH/\rho_0\approx5.0$--$5.6$, close to the cubic values. For these modes the approximation is accurate. Stiffer modes have a smaller transport, and the approximation overshoots. With Gaussian kicks the transport is small but nonzero, and the single-power form fails. The validity of the approximation is therefore decided by endpoint measurements, not by a noise-skewness or stiffness criterion. In smooth networks the endpoint measurements give $\Delta_2=6\Delta_1$ up to less than 1\% of $\rhoH$. This is \eqref{eq:quartic}. One measured ratio then fixes both (\cref{fig:overview}c). The approximation concerns smooth remainders. A boundary remainder $\propto f^2$ (\cref{lem:kink}) has no counterpart in the change of the regular Hessian. For ReLU networks the transport must therefore be measured, not inferred from $\rho_0$ (\cref{sec:momentum}).

\subsection{Networks with non-invariant parameters}\label{sec:mixture}
In a batch-normalized network each normalized output channel is its own sphere. Tensors that write to a residual stream that is always read through a normalization are invariant only jointly. The normalization affine parameters and a trained head are not invariant. Three statements remain exact in this setting. (i)~The norm balance \eqref{eq:radial-sgd}--\eqref{eq:radial-ema} of every invariant group holds while the other parameters are trained. (ii)~The rescaling of \cref{prop:symmetry} remains a symmetry when it acts on the invariant groups and their hyperparameters only. What a trained head breaks is the \emph{common} rescaling of all learning rates, including its own. This common rescaling is the variation used by \citet{sadrtdinov2025idealgas} and in our campaigns. The coefficient of \eqref{eq:maxwell} then deviates from $\ell\sum_a(d_a-1)$ by an amount set by the response of the head. This is a mechanism for the larger deviations that \citet{sadrtdinov2025idealgas} report for partially invariant models. At fixed head hyperparameters the relation is exact. (iii)~Let the update be \eqref{eq:update} with $\eta$, $\lambda$ and $\beta$ common to all parameters. (AdamW has a different, adaptive update.) Then the finite-decay identity \eqref{eq:finite} holds with the measurable head term $\lambda\bar e$. Here $e=\psi\cdot\nabla_\psi L$ for exactly invariant groups and $e=\theta\cdot\nabla L$ in general, including groups whose invariance is broken.

Locally the radii determine a block-diagonal angular step. The Hessian and the noise covariance, however, retain cross-group couplings. The stationary covariance therefore follows from the matrix Lyapunov equation of \cref{prop:coupled}. The Gaussian surrogate of that proposition is reversible precisely when the transformed curvature and noise covariance commute. Reversibility alone does not imply a common mode temperature. A Gibbs law of the harmonic energy requires an additional compatibility between noise and dissipation. The non-invariant parameters are harmonic modes. For gradient memory with decoupled decay, a mode of curvature $h>0$ is linearly stable for $\eta h<(1+\beta)(2-\gamma)/(1-\beta)$ ($0\le\gamma<1$). This is the period-two boundary of $m_+=\beta m+(1-\beta)hx$, $x_+=(1-\gamma)x-\eta m_+$. The condition reduces to $\eta(h+\lambda)<2$ for SGD and to $\eta h<\aeos$ as $\gamma\to0$. A head-like scale $\vartheta$ sets the curvature scale $\vartheta^2$ seen by the spheres. When the angular variables relax much faster than the head, their conditional stationary statistics define an effective averaged force on the head. A potential description requires an additional integrability condition. Invariance can also be broken inside a group. A normalization $(z-\mu_B)/\sqrt{v_B+\epsilon_{\mathrm{BN}}}$ with fixed $\epsilon_{\mathrm{BN}}>0$ maps a rescaling $z\to cz$ to $\epsilon_{\mathrm{BN}}\to\epsilon_{\mathrm{BN}}/c^2$. The invariance is therefore approximate wherever the batch variance $v_B$ is not large compared with $\epsilon_{\mathrm{BN}}$. The violation is measured by $\E[w_a\cdot g_{B,a}]$. When the gradient and memory of a group vanish (a dead channel), the update \eqref{eq:update} contracts its norm by the factor $1-\gamma$ per step. This decay drives the group into that regime.
\section{Numerical tests: sphere models, teacher--student networks and convolutional networks}\label{sec:experiments}

\subsection{Sphere models}\label{sec:sphere}
A single sphere tests the harmonic theory without estimator issues (\cref{fig:toy}; \ncheck{sphere}). It has dimension $d=51$, Hessian eigenvalues log-spaced in $[5,100]$, isotropic angular noise $\sigma_{\mathrm a}\in\{0.2,0.1\}$ and $\lambda=0.01$.
\begin{itemize}[leftmargin=*,itemsep=1pt]
\item \emph{SGD}, $\Rtop$ up to 0.95: the radius is within 0.03\% of \eqref{eq:thermostat}, where the fixed-angular-noise ideal-gas radius $r_0$ is off by up to 48\%. The mode variances are within 0.7\% (top mode 2\% at $R=0.95$), against up to a factor 20 for the continuous-time prediction.
\item \emph{Gradient memory} ($\beta=0.9$): the radius is within 0.03--2.5\% and the mode variances within 2.4\% of \eqref{eq:modes-q} (medians within 0.5\%) up to $R\approx0.7$. This holds although $\beta\gamma/(1-\beta)$ reaches 0.56 and weight decay lowers the temperatures by up to 36\%. The contraction alone, without the curvature shift, leaves deviations of up to 5\%. Closer to the threshold ($R\gtrsim0.85$) the chains make intermittent large excursions of the top mode. At the largest settings they settle in a large-amplitude state outside the harmonic description.
\item \emph{AdamW} (rotated basis): second moment within 0.1--1.9\%, mode variances within 2--7\% and radius within 1.4--5.0\% of the self-consistent prediction. The squared radius falls to 0.27 of its ideal-gas value $r_{\mathrm A}^2$ at $\Rtop=0.4$, the opposite sign to SGD.
\item \emph{Concentration} (SGD, 49 bulk modes in $[5,50]$ and an isolated top mode at 100; \ncheck{condense}): as $\gamma$ goes from $3\cdot10^{-5}$ to $3\cdot10^{-2}$ the top-mode temperature grows from $1.4\,T^{(0)}$ to $3685\,T^{(0)}$. It matches the large-$D$ formula ($\gamma_c=5.0\cdot10^{-4}$) to within 3\% for $\gamma\ge3\gamma_c$, with a smooth finite-$D$ crossover near $\gamma_c$. As $\gamma$ grows, $\Rtop$ saturates just below 1, and the bulk is unchanged. The radius agrees within 0.11\%.
\item \emph{Finite decay} (quadratic angular loss, isotropic tangent noise, $\Rtop\le0.1$; \ncheck{zero}): for SGD $c-1=(1.00\pm0.02)\gamma$, the radial factor of \eqref{eq:c-sgd} with $\rho_{\tan}=0$. With memory, $c$ follows the per-mode factor $(1-\beta)/(1-q)$ of \cref{lem:geo}. This factor departs from 1 by up to $-18\gamma$. The remaining difference is between $+1.0\gamma$ and $-3.1\gamma$ and shrinks with the number of tangent dimensions (\ref{num:zero}). \Cref{eq:F2} holds in every run to $2.3\cdot10^{-4}$.
\item \emph{Scaling symmetry}: $(\eta,\lambda)\to(4\eta,\lambda/4)$ leaves the angular statistics unchanged and gives $r^2$ ratios 4.000 (SGD), 4.001 (memory) and 16.007 (AdamW), as \cref{prop:symmetry} requires.
\item \emph{Geometric terms} (\cref{approx:geo}): with Hessian-aligned noise and $\gamma=0.03$, the one-step regression coefficient of a soft SGD mode is 0.9812, against $1-\eta_0h=0.9814$ and $1-\gamma-\eta_0h=0.951$. In addition, $\eta_0\E\omega/\gamma=1.003$ (SGD) and 1.002--1.006 (memory). With $g$ and $H$ evaluated exactly at the current $w$, the memory FDR with the per-mode factor alone has residuals of 1.3\% and 0.5\% at $\gamma=10^{-2}$ and $3\cdot10^{-3}$. These are the remaining $O(\gamma)$ terms of \cref{sec:geo}.
\item \emph{Mixture} (three spheres of 20 modes with different noise and a head, $L=\vartheta^2\sum_a\bar L_a+\frac{c_\vartheta}2(\vartheta-\vartheta_0)^2$, SGD on all parameters): group radii within 0.6\%, head mean within 0.1\%, mode variances within 1\%, head variance within 10--37\% (adiabatic approximation). The coefficient of \eqref{eq:maxwell} under a common rescaling of all learning rates is 30.00 with the head frozen (exact: 30) and 34.96 with a trained head, against 35.09 predicted. The head thus breaks the common rescaling by 17\%, as predicted. At fixed head hyperparameters the relation is exact.
\end{itemize}
\begin{figure}[t]
\centering
\includegraphics[width=\linewidth]{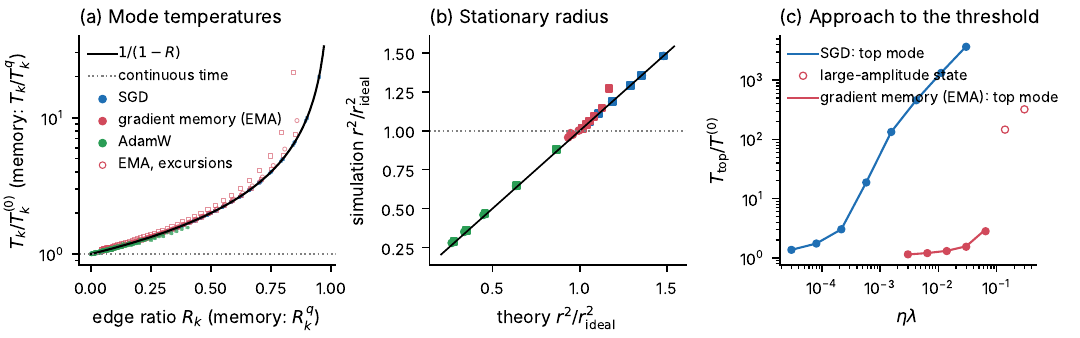}
\caption{Sphere model ($d=51$). (a)~Mode temperature against the edge ratio. SGD, gradient memory and AdamW follow $1/(1-R)$ (the continuous-time limit predicts 1). With memory, $T^q_k$ and $R^q_k$ of \eqref{eq:modes-q} are plotted. Open symbols mark runs with intermittent excursions ($\Rtop>0.75$). (b)~Stationary squared radius relative to the square of the ideal-gas radius at fixed angular noise ($r_0$ for SGD and memory, $r_{\mathrm A}$ for AdamW). Simulation against theory, no fitted parameters. (c)~Approach to the threshold at fixed learning rate. For SGD the temperature of the top mode grows with $\eta\lambda$ (concentration in the top mode). With memory the chains leave the harmonic regime (open symbols: large-amplitude states).}
\label{fig:toy}
\end{figure}

\subsection{A teacher--student network with an analytic population loss}\label{sec:ts}
The normalized soft committee machine \citep{saad1995online} has an analytic population loss. The loss and all its derivatives at $w_t$ and $w_{t+1}$ are therefore evaluated without Monte Carlo error. Hence so is every term of \cref{prop:energy,prop:transport,thm:finite,prop:antisym}. Only stationary averages (and estimated noise covariances) carry sampling error. The student is $f_w(x)=\sum_{i=1}^{N_{\mathrm s}}o_i\,\phi_\tau(n_i\!\cdot x)$ with $n_i=w_i/\|w_i\|$ and $\phi_\tau(z)=\mathrm{erf}(z/(\sqrt2\,\tau))$. Every hidden unit is therefore a scale-invariant group. The teacher $f_*(x)=\sum_{l=1}^{N_{\mathrm t}}o^*_l\,\phi_\tau(v_l\!\cdot x)$ has orthonormal directions $v_l$ and the same activation. Changing the width $\tau$ thus changes teacher and student together. Inputs are $x\sim\mathcal N(0,I_{N_{\mathrm{in}}})$, labels $y=f_*(x)+\sigma\varepsilon$ with unit-variance noise $\varepsilon$, and the loss is $\frac12(f_w-y)^2$. With the overlaps $n_i\cdot n_j$ and $n_i\cdot v_l$ and the Gaussian kernel $k_\tau(u\cdot u')=\E\,\phi_\tau(u\cdot x)\phi_\tau(u'\cdot x)=\frac2\pi\arcsin\frac{u\cdot u'}{1+\tau^2}$ for unit vectors $u$, $u'$,
\begin{equation}\label{eq:ts-loss}
L(w)=\frac12\sum_{i,j}o_io_j\,k_\tau(n_i\!\cdot n_j)-\sum_{i,l}o_io^*_l\,k_\tau(n_i\!\cdot v_l)+\frac12\sum_{l,l'}o^*_lo^*_{l'}\,k_\tau(v_l\!\cdot v_{l'})+\frac{\sigma^2}{2}.
\end{equation}
We use $N_{\mathrm s}=N_{\mathrm t}=4$, $N_{\mathrm{in}}=16$ ($d=64$) and width $\tau=1$ (wide) or $0.1$ (narrow, anharmonic on a smaller angular scale). The output weights are equal, $o_i=o^*_i=1$, for SGD and graded, $o_i=o^*_i=1,\frac12,\frac14,\frac18$, with memory. Derivatives are automatic derivatives of \eqref{eq:ts-loss} in float64. At the realizable teacher, under these assumptions, the noise covariance is exactly $\Sigma=\sigma^2H/B$. This alignment is a property of the model at the teacher, not a general network identity. Chains run in parallel (128 chains, 20,000 measured steps after 5,000--20,000 steps of burn-in). By \cref{cor:maxwell} only $\gamma$ matters, so we set $\eta=1$.

\paragraph{A symmetry of the label noise.} For this model with inversion-symmetric inputs and no biases, the residual $\Delta f(x)=f_w(x)-f_*(x)$ and the Jacobian $J_w(x)=\nabla_wf_w(x)$ are odd in $x$. Hence the sample gradient $[\Delta f(x)-\sigma\varepsilon]J_w(x)$ is unchanged when $(x,\varepsilon)$ is replaced by $(-x,-\varepsilon)$. Flipping both with an independent random sign therefore leaves every sample gradient, and every trajectory, unchanged. Under the flip the input law stays Gaussian and the label noise becomes symmetrized. The dynamics depends only on the symmetrized label-noise law. Our runs with centred-exponential label noise are therefore runs with a symmetric, non-Gaussian (Laplace-like) noise, not a skewness intervention. Skewness is tested in the toy of \cref{app:toy}, where the kicks are injected directly.

\begin{table}[t]
\centering
\caption{Teacher--student, SGD (7 of 16 runs; $B=16$, Gaussian label noise of level $\sigma$, equal output weights). Mode variances relative to the surrogate \eqref{eq:geo-sgd}: stiff modes (the two upper Hessian levels) and soft modes (44 of 60). In parentheses, soft modes against the contraction-only form $x_+=(1-\gamma-\eta_0h)x-\eta_0\zeta$.}
\label{tab:ts}
\small
\begin{tabular}{ccccccc}
\toprule
$\gamma$ & $\sigma$ & $\Rtop$ & $\|g\|^2/$continuum & $\|g\|^2/$linear FDR & stiff var/pred & soft var/pred\\
\midrule
$10^{-3}$ & 0.30 & 0.142 & 1.157 & $1.0008\pm0.0003$ & 1.000 & 0.997 (1.091)\\
$10^{-3}$ & 0.10 & 0.307 & 1.419 & $1.0005\pm0.0001$ & 0.999 & 0.999 (1.042)\\
$3\cdot10^{-3}$ & 0.10 & 0.389 & 1.598 & $1.0010\pm0.0001$ & 1.002 & 1.011 (1.114)\\
$10^{-2}$ & 0.30 & 0.316 & 1.431 & $1.0044\pm0.0001$ & 1.004 & 1.034 (1.466)\\
$10^{-2}$ & 0.10 & 0.447 & 1.743 & $1.0018\pm0.0001$ & 1.008 & 1.044 (1.351)\\
$3\cdot10^{-2}$ & 0.30 & 0.399 & 1.595 & $1.0104\pm0.0001$ & 1.015 & 1.106 (2.188)\\
$3\cdot10^{-2}$ & 0.10 & 0.475 & 1.795 & $1.0045\pm0.0001$ & 1.021 & 1.108 (2.015)\\
\bottomrule
\end{tabular}
\end{table}

\begin{observation}[SGD: finite-step fluctuations in a smooth nonlinear network]\label{obs:ts-sgd}
With SGD (\cref{tab:ts}), $\|g\|^2$ exceeds the continuous-time value $\frac\eta2\Tr\Sigma H$ by 16--80\%. However, the linear FDR with the edge term holds to 0.05--1\%. The loss-balance ratio equals $1+\rho_0$, and the norm balance holds to $10^{-4}$. The remaining excess, $c-1=(0.15\text{--}0.75)\gamma$, corresponds by \eqref{eq:c-sgd} to a negative tangent remainder, $\rho_{\tan}=-(0.25\text{--}0.86)\gamma$. The top modes ($\Rtop$ up to 0.48) follow $T^{(0)}/(1-R)$ to within 0.1--2\%, where the continuous-time prediction is off by 16--94\%. The soft modes follow the surrogate \eqref{eq:geo-sgd} to within 1.1\% for $\gamma\le3\cdot10^{-3}$ and 3--11\% for $\gamma=10^{-2}$--$3\cdot10^{-2}$. These deviations are the neglected higher-order geometric terms of \cref{approx:geo}. The contraction-only form is off by up to a factor 2.2.
\end{observation}

\begin{figure}[t]
\centering
\includegraphics[width=\linewidth]{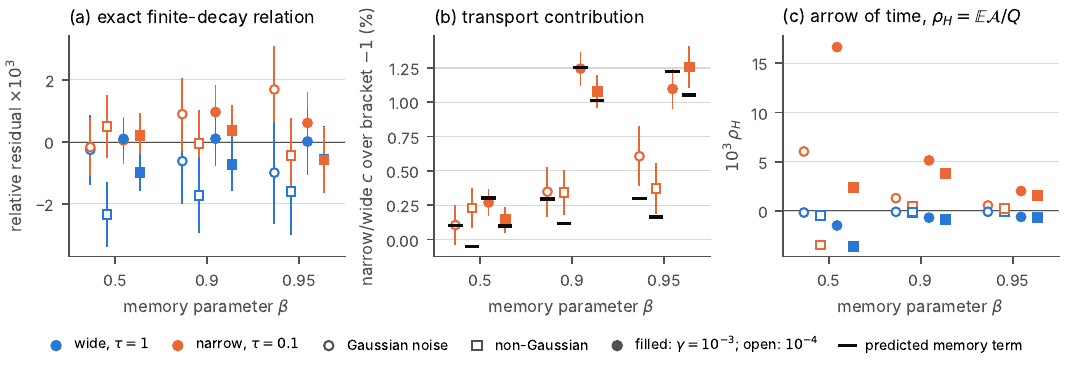}
\caption{Teacher--student network with gradient memory (24 runs; graded output weights, $B=16$, $\sigma=0.3$, 128 chains, 20,000 measured steps), against the memory parameter $\beta$. (a)~Relative residual of the exact finite-decay relation \eqref{eq:F2}: the measured $c$ divided by the right-hand side of \eqref{eq:F2}, evaluated with the measured $U$, $\rho_0$ and $\rho_1$, minus one. (b)~Curvature-transport contribution: the narrow-to-wide ratio $c(\tau{=}0.1)/c(\tau{=}1)$ divided by the ratio of the single-step remainder brackets of \eqref{eq:F2}. It is compared with the predicted memory term $[1-b_2\rhoH(1)]/[1-b_2\rhoH(0.1)]$ (bars). (c)~The signed arrow-of-time observable $\rhoH=\E\cA/Q$. It is negative and geometric for the wide activation and positive for the narrow one at $\beta\ge0.9$. Error bars are one standard error over chains. The full table is in \cref{app:tstable}; \ncheck{ts}.}
\label{fig:ts-momentum}
\end{figure}

\begin{observation}[Gradient memory: exact relations, geometric decay effects and transport]\label{obs:ts-ema}
With memory ($\beta=0.5$, 0.9, 0.95; $\gamma=10^{-4}$ and $10^{-3}$; \cref{fig:ts-momentum}) the linear-FDR ratio ranges from 0.977 to 1.020.
(i)~\emph{Exact relations.} \eqref{eq:F2} holds: measured over predicted $c$ is 0.9998 (s.d.\ 0.0009, range 0.9977--1.0017, RMS 0.8 standard errors over the 24 runs). The transport identity \eqref{eq:transport} holds to within 0.16\%. The norm balance holds to within 0.03\%.
(ii)~\emph{Geometric decay effects.} With the wide activation the remainders are small ($|\rho_0|\le6\cdot10^{-4}$) and largely geometric. Here $c$ falls below 1 as $\beta\gamma/(1-\beta)$ grows, to 0.977--0.980 at $\beta=0.95$, $\gamma=10^{-3}$. The per-mode factor $(1-\beta)/(1-q)$ of \cref{lem:geo} gives the size of the effect (0.981 there). However, $c$ lies below it in all 12 wide runs, by 0.1--4.1 standard errors. With Gaussian label noise at $\gamma=10^{-3}$ the difference is $-(0.1\text{--}1.5)\gamma$. This is the size of the terms that the factor does not capture (\cref{sec:geo}) and that \eqref{eq:F2} accounts for. The curvature change of these runs is itself geometric, $\rhoH\approx-(0.58\text{--}0.99)\,\gamma$ for $\beta\ge0.9$. It comes from the change of the radius along the step. The radial part of $\delta_X$ follows \eqref{eq:dX-rad} at every $\beta$, exceeding it by at most 5.3\% (narrow activation: 14\%). With the symmetric non-Gaussian label noise $c$ is 0.15--0.29\% lower than with Gaussian noise of the same variance in all six wide pairs. This is an effect of higher noise moments. \eqref{eq:F2} holds for both.
(iii)~\emph{Curvature transport.} The narrow activation gives a positive remainder ($\rho_0=1.3\cdot10^{-4}$--$1.8\cdot10^{-2}$). For $\beta\ge0.9$ the ratios are $\rho_1/\rho_0=2.3$--$2.6$ and $\rhoH/\rho_0=1.7$--$3.4$, not the cubic values 3 and 6. Comparing the narrow with the wide activation at the same $(\beta,\gamma,\text{noise})$ approximately removes the leading geometric contribution. The ratio $c(0.1)/c(1)$ should then equal the ratio of the remainder brackets of \eqref{eq:F2} times the memory term $[1-b_2\rhoH(1)]/[1-b_2\rhoH(0.1)]$, $b_2=\beta^2/(2(1-\beta^2))$. This term uses only the difference of the two curvature changes. The measured ratio equals this product within 1.9 standard errors in all 12 pairs (RMS 1.0). At $\gamma=10^{-3}$ and $\beta\ge0.9$ the memory term is 1.01--1.25\%, 4.3--10.6 times the single-step bracket. In these runs omitting it misses the measured ratio by 7.4--10.0 standard errors. At $\beta=0.5$ the memory term is $-0.05$ to $+0.30\%$.
\end{observation}
The model thus resolves the curvature-transport contribution and its $\beta$-dependence on a small positive remainder. For $\beta\ge0.9$, the anharmonic transport correction isolated by the narrow--wide comparison is positive. The total ratio can nevertheless lie below one because of geometric decay effects. It tests the endpoint relations, not the cubic-remainder approximation. It does not satisfy the endpoint ratios of that approximation. With equal output weights the mode variances follow \eqref{eq:modes-q} to within 1\% (stiff modes) and 2--4\% (soft modes). Exploratory runs with an unrealizable teacher, widths down to 0.04 and $\gamma$ up to $10^{-2}$ agree with the per-mode factor to within 1.5\% (\ncheck{ts}).

\begin{observation}[Coupled groups with unequal radii]\label{obs:ts-coupled}
With graded output weights $1,\frac12,\frac14,\frac18$ the stationary radii of the four hidden units differ by a factor of 8. For SGD (six runs: four with a realizable teacher at $\gamma=10^{-3}$ and $10^{-2}$, two with an unrealizable one, $N_{\mathrm t}=8$, at $\gamma=10^{-3}$) a common angular step fails. The single-$\eta_0$ mode prediction is off by factors 0.42--0.60 for the stiff modes and 2.1--9.3 for the soft ones. The coupled surrogate of \cref{prop:coupled} is evaluated with the measured radii and the measured $60\times60$ angular noise covariance. It gives the stiff-mode variances within 0.6\% and the soft ones within 0.5--1.2\% at $\gamma=10^{-3}$ (4--9\% at $\gamma=10^{-2}$, the higher-order geometric terms of \cref{obs:ts-sgd}). It also gives the full covariance matrix within 1.3--7.4\% in Frobenius norm at $\gamma=10^{-3}$ (6.7\% and 16.1\% at $10^{-2}$). Near the realizable teacher the noise is Hessian-aligned. By \cref{prop:coupled}(iii) the surrogate is therefore nearly reversible despite the unequal radii (predicted relative current $3\cdot10^{-4}$--$6\cdot10^{-3}$). For the unrealizable teacher the measured noise is aligned as well (predicted $1.5\cdot10^{-5}$). The measured relative current $\|J_1\|/\|\E[x_+x^\top+xx_+^\top]\|$ (here $x$ are the tangent coordinates of all units, not the inputs) is 0.002--0.007. It was recorded in these runs without an estimate of its sampling floor. It is 0.003 also where $1.5\cdot10^{-5}$ is predicted. We therefore do not regard it as resolved.
With memory, the same analysis in phase space $(x,m)$ was applied to the runs of \cref{fig:ts-momentum}. These runs were repeated with the full covariances and a noise-floor estimate of the current recorded (same trajectories). Of the 24 runs, 22 completed. The two non-Gaussian runs at $\beta=0.95$, $\gamma=10^{-3}$ did not finish within the compute window. With the wide activation the stiff modes agree within 0.7\%. The soft ones agree within 1.4\% at $\gamma=10^{-3}$ and are 8\% below the prediction at $\gamma=10^{-4}$, where their relaxation times approach the measurement window. The single-$\eta_0$ prediction is off by factors 0.61--0.62 and 1.9--2.1. The predicted current is at most $3\cdot10^{-4}$. The measured one equals its sampling floor (ratio 0.97--1.04). With the narrow activation the loss is anharmonic on the scale of the fluctuations. The linear surrogate fails for the variances (stiff modes 3--19\% and soft modes 15--77\% above it). However, the predicted current at $\gamma=10^{-3}$, $1.7$--$2.5\cdot10^{-3}$, is resolved. The measured one is $1.2$--$2.1\cdot10^{-3}$, 1.8--2.6 times its floor and 0.74--0.85 times the prediction. At $\gamma=10^{-4}$ the predicted current, $3$--$4\cdot10^{-4}$, is below the floor. The measured one is 1.0--1.2 times the floor.
\end{observation}
A minimal model makes the distinction concrete. It has two coupled spheres with two tangent coordinates each, a cross-coupled quadratic loss and isotropic noise of angular variance $s_1$, $s_2$ (4,000 chains, $4\cdot10^4$ measured steps; \ncheck{coupled}). With nearly equal noise temperatures ($\alpha_1s_1=\alpha_2s_2$ to within 1\%) the antisymmetric lag covariance between the coupled coordinates is $-4.6\cdot10^{-6}$ against $-5.0\cdot10^{-6}$ predicted for the residual mismatch. With temperature ratios 2.1 and 4.2 it is $-2.96\cdot10^{-4}$ and $-5.02\cdot10^{-4}$, against $-2.98\cdot10^{-4}$ and $-5.11\cdot10^{-4}$ from \cref{prop:coupled}. Both groups satisfy their norm balance exactly. The stationary state nonetheless carries a probability current set by the mismatch of the local noise temperatures.

\paragraph{Arrow of time.} The mean of the pair observable, $\E\cA=Q\rhoH$, differs from zero in every run, by 13 to 277 standard errors (\cref{fig:ts-momentum}c). The parameter trajectory thus shows a resolved time-reversal asymmetry. With the wide activation $\E\cA$ is negative and geometric (\cref{obs:ts-ema}\,(ii)). With the narrow activation it is positive for $\beta\ge0.9$. Its excess over the wide value is the transport of \cref{obs:ts-ema}\,(iii). For SGD the radius of each hidden unit changes with probability one. The pair law of the parameters is then singular with respect to its reversal, and \eqref{eq:I2} carries no information (\cref{app:proofs}). With memory the radial part of the step contains the memory, $w_t\cdot w_{t+1}=(1-\gamma)\|w_t\|^2-\eta\beta\,w_t\cdot m_t$ per group. That argument therefore does not apply. Here \eqref{eq:I2} gives $I_2\ge8\cdot10^{-5}$--$3\cdot10^{-2}$ per step in the 22 repeated runs, $1.7\cdot10^{2}$ to $7.6\cdot10^{4}$ times its noise floor. The two activations give bounds of the same size. The bound therefore does not separate geometric from anharmonic irreversibility. In addition, it concerns the parameter pairs alone, not the entropy production of the full optimizer.

\paragraph{ReLU.} With ReLU units the singular curvature along the step is measured as the population Hessian minus the regular Hessian of an independent batch. Along an SGD trajectory at $\gamma=10^{-3}$ it is $-2.127\%\pm0.018\%$ of $\Delta w^\top H\Delta w$, and \cref{prop:relu} gives $-2.131\%\pm0.003\%$. It is negative, as in the batch-normalized ResNets of \cref{sec:networks}, but 3.3--5.6 times smaller.

The smooth teacher--student model is therefore nearly harmonic per step. The departures of the networks below need an ingredient it lacks. Batch statistics, cross-entropy far from its minimum, depth and a finite training set are candidates. The explored family does not establish which of them is necessary.

\subsection{Scale-invariant convolutional networks on CIFAR-10 and CIFAR-100}\label{sec:networks}
\paragraph{Setup.}\label{sec:setup}
CIFAR-10 and, in campaign v4, CIFAR-100 \citep{krizhevsky2009cifar} (each normalized with its own channel statistics), batch size 128 (each batch drawn uniformly with replacement from the training set, independently at every step, with independent random crops where used; there are no epochs), batch normalization with batch statistics. \emph{Thin SI ResNet-18}: pre-activation ResNet-18 \citep{he2016identity} with widths $(4,8,16,32)$, every convolution followed by an affine-free batch normalization, no biases and a fixed random head. It has 43,692 parameters in 300 channel groups and is the fully invariant model of \citet{sadrtdinov2025idealgas}. \emph{SI ConvNet}: four $3\times3$ convolutions of widths $(8,16,32,64)$ with affine-free batch normalization, max-pooling and a fixed head. It has 24,408 parameters in 120 groups. \emph{Practical ResNet-18} of \citep{sadrtdinov2025idealgas}: width 8, normalization with affine parameters, trainable head, random crops. Of its 176,178 parameters, 73,728 form 240 per-channel groups, and 100,824 are invariant only jointly (one group per stage's residual stream). The remaining 1,626 are Euclidean (\cref{sec:mixture}). Activations are ReLU or the shifted Softplus $\phi_\tau$, applied to normalized pre-activations. Hence $\tau$ is in units of the pre-activation scale. Scale invariance is then preserved up to the normalization constant $\epsilon_{\mathrm{BN}}$ (\cref{sec:mixture}). We measure the invariance directly at every diagnostic point (\ncheck{invariance}). The first measure is the relative response factor $s=\epsilon_{\mathrm{BN}}/(v_B+\epsilon_{\mathrm{BN}})$ of the weakest channel on the diagnostic batch. A rescaling of the channel by $c$ changes a normalized output $\hat z$ by $d\hat z/d\log c=s\,\hat z$. The second is the Euler residual $w\cdot g$ of the invariant parameters, which vanishes for exact invariance. In the scale-invariant models of v1--v3, the response factor of the weakest channel has median $3\cdot10^{-4}$ and maximum $1.6\cdot10^{-2}$ over 2,176 diagnostic points of 128 runs. In these models $|w\cdot g|\le5\cdot10^{-5}\,\|w\|\|g\|$, and the Euler residual contributes at most $2.3\cdot10^{-3}$ of the norm balance. In the CIFAR-100 runs the response factor stays below $1.1\cdot10^{-3}$, and the Euler residual of each group is at most $2.4\cdot10^{-3}\,\|w_a\|\|g_a\|$ (median $3\cdot10^{-7}$). In the practical ResNet the weakest channel is far from invariant. Its response factor exceeds $10^{-2}$ at 97\% of the points (median 0.36; dead channels in 8 of 31 runs). The joint Euler residual of all invariant groups stays below $5\cdot10^{-3}\,\|w\|\|g\|$ and contributes up to 12\% (median 0.8\%) of the norm balance. This aggregate can hide cancellations between groups. The response factor is not a bound on the residual of $w_a\cdot g_a=0$, which also involves downstream derivatives. We therefore report both.

Unless stated otherwise, a run of v1--v3 has $\mathrm{clip}(10/\gamma,\,2\cdot10^5,\,10^6)$ steps (the radius relaxes at rate $\approx4\gamma$). Diagnostics are taken at 16 points (32 in v3 and v4), evenly spaced in the last 40\% of the run. Campaign~v1 (119 runs): thin ResNet with SGD ($\lambda=0.01$, $\eta\in[10^{-3},10]$, 13 values, 2 seeds), EMA momentum ($\lambda=10^{-3}$, $\eta\in[10^{-2},10^2]$), AdamW ($\lambda=0.1$, $\eta\in[10^{-4},10^{-1}]$), symmetry pairs at fixed $\eta\lambda$ for SGD and AdamW, the SI ConvNet ($2\cdot10^6$ steps), and the practical ResNet with SGD and AdamW. Campaign~v2 (28 runs) addresses the finite-step remainder, with ReLU against Softplus $\tau\in\{0.5,0.05\}$. It comprises the thin ResNet with SGD at $\eta\lambda\in\{3\cdot10^{-5},10^{-4},10^{-3}\}$ (2 seeds), EMA and AdamW at $\eta\lambda\in\{10^{-4},10^{-3}\}$, and the practical ResNet with SGD at $\eta\lambda=10^{-4}$. Campaign~v3 (12 runs): EMA momentum in the thin ResNet ($\lambda=10^{-3}$, $\eta\in\{0.1,1\}$, ReLU and Softplus 0.5, 2 seeds) and in the practical ResNet ($\lambda=10^{-2}$, $\eta\in\{0.01,0.1\}$), $2\cdot10^5$ steps. Campaign~v4 (16 runs): the thin-ResNet settings of v3 on CIFAR-100, with the fixed head of Frobenius norm 10 of \citet{sadrtdinov2025idealgas}. The 100 rows of this head have norm $\approx1$ instead of $\approx3.2$. This keeps the logits small (training loss 3.5--3.9, test accuracy 19--29\%). We record the head scale rather than retune it. Differences between v3 and v4 therefore reflect the dataset and the head scale together. The accuracies are context, not a benchmark. Runs of $4\cdot10^5$ steps (seed 0) showed ReLU at $\gamma=10^{-4}$ still sharpening. In this setting $\|g\|^2$, $\Xi_H$ and $B_H$ rose by 10--59\% over the window (3.1--3.5 standard errors each). All settings were therefore rerun with $2\cdot10^6$ steps (3 seeds), the CIFAR-100 length of \citet{sadrtdinov2025idealgas}. Below, v4 denotes these runs. Altogether the four campaigns comprise 175 runs. In the 28 runs of v3 and v4, all with gradient memory, every term of the exact identities is measured directly: 24 production runs, twelve on each dataset, and the four shorter CIFAR-100 runs.

The estimators are described in \cref{app:estimators}. Here $\|g\|^2$, $\Tr\Sigma$, $\Tr\Sigma H$ and the Hessian Gram matrix in $\mathrm{span}(g,m,\theta)$ come from 400 independent batch gradients and 64--128 Hessian--vector products on fresh batches. The finite-sample biases are removed. As a check, the Gram entry $g^\top Hw$ equals $-\|g\|^2$ by \eqref{eq:si}. Over 1,920 diagnostic points of the SI models, the ratio has median 0.998 (10--90\%: 0.955--1.045). The top three Hessian eigenpairs come from Lanczos on 8 fixed batches. Their Rayleigh quotients are re-measured on 32 fresh batches, and $\Rtop$ uses the re-measured value. This value is about 0.79 of the fixed-batch value (median). Mode variances come from 1,000 weight snapshots, taken every 25 steps, projected on the top eigenvectors. The remainder \eqref{eq:remainder} is computed in float64 for $f\in\{1,\frac12,\frac14,\frac18,\frac1{16}\}$, with an evaluation batch independent of the step. It is averaged over 128 draws of the step per diagnostic point. In v3 and v4, for 64 draws of the step per point, the loss, gradient and Hessian--vector products of one evaluation batch are computed in float64 at \emph{both} ends of the step the optimizer would take. This gives $\Rem(1)$, $\Rem'(1)$, $\cA$, $\cK$ and $\DH$ directly. The drifts come from the trajectory (\cref{app:estimators}). The v3 and v4 uncertainties are joint (block bootstrap, \cref{app:estimators}).

\subsubsection{Finite steps without memory: edge term, radius and mode temperatures}\label{sec:v1}

\begin{figure}[t]
\centering
\includegraphics[width=\linewidth]{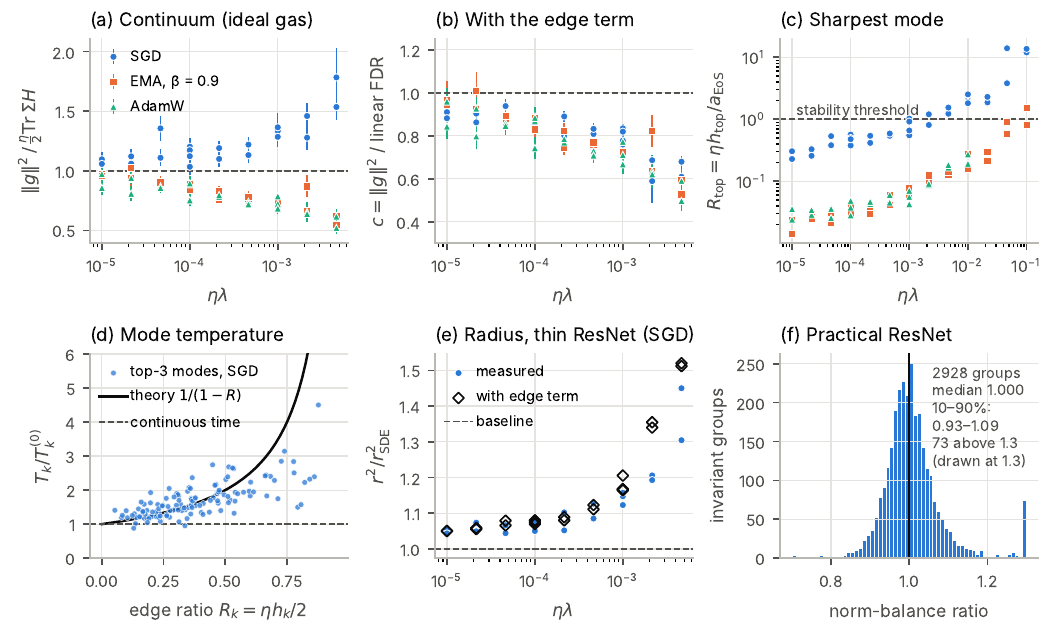}
\caption{Campaign v1 (thin SI ResNet-18 unless stated). (a)~$\|g\|^2$ relative to the continuous-time (ideal-gas) value $\frac\eta2\Tr\Sigma H$ (``SDE''). It is above 1 for SGD and below 1 with memory and AdamW. (b)~Relative to the linear FDR \eqref{eq:fdr} with the edge term. SGD comes down to 0.8--0.9 at small $\eta\lambda$. Memory and AdamW are unchanged because their edge term is negligible. (c)~Edge ratio of the sharpest mode (curvature re-measured on fresh batches). (d)~Variance of the top-3 modes (three SGD models) relative to the continuous-time value, against $1/(1-R)$. (e)~Radius of the thin ResNet with SGD, measured and predicted with the edge term, relative to the measured-noise baseline $r^2_{\mathrm{SDE}}$. (f)~Practical ResNet-18: ratio of the two sides of the norm balance \eqref{eq:radial-sgd} for 2,928 invariant groups (per channel and joint) in 12 runs.}
\label{fig:v1}
\end{figure}

\begin{table}[t]
\centering
\caption{Thin SI ResNet-18, medians per $\eta\lambda$ bin. ``SDE'': $\|g\|^2/(\frac\eta2\Tr\Sigma H)$. ``Edge'': share of the edge term $\frac\eta\aeos g^\top Hg$ in the linear FDR. $c$: $\|g\|^2$ relative to the linear FDR \eqref{eq:fdr}. $\Rtop$: edge ratio of the sharpest mode. Runs per bin: SGD 4, 6, 4, 4, 4; EMA 4, 4, 4, 2, 4; AdamW 4, 5, 4, 3, 4.}
\label{tab:v1}
\small
\begin{tabular}{lcccccccccc}
\toprule
& \multicolumn{4}{c}{SGD} & \multicolumn{3}{c}{EMA momentum} & \multicolumn{3}{c}{AdamW}\\
\cmidrule(lr){2-5}\cmidrule(lr){6-8}\cmidrule(lr){9-11}
$\eta\lambda$ & SDE & Edge & $c$ & $\Rtop$ & SDE & $c$ & $\Rtop$ & SDE & $c$ & $\Rtop$\\
\midrule
$1$--$2\cdot10^{-5}$ & 1.08 & 17\% & 0.89 & 0.28 & 0.97 & 0.95 & 0.02 & 0.90 & 0.89 & 0.03\\
$5\cdot10^{-5}$--$10^{-4}$ & 1.15 & 23\% & 0.88 & 0.49 & 0.89 & 0.88 & 0.03 & 0.87 & 0.86 & 0.04\\
$2$--$5\cdot10^{-4}$ & 1.16 & 26\% & 0.82 & 0.54 & 0.78 & 0.76 & 0.04 & 0.76 & 0.75 & 0.05\\
$10^{-3}$ & 1.35 & 36\% & 0.80 & 0.79 & 0.74 & 0.72 & 0.07 & 0.72 & 0.71 & 0.06\\
$2$--$5\cdot10^{-3}$ & 1.50 & 50\% & 0.64 & 1.22 & 0.64 & 0.61 & 0.13 & 0.63 & 0.61 & 0.13\\
\bottomrule
\end{tabular}
\end{table}

\begin{observation}[Continuum limit and edge term]\label{obs:v1}
(i)~With SGD, $\|g\|^2$ exceeds the continuous-time value by 8--50\%, increasingly with $\eta\lambda$. At the same time the sharpest mode moves towards the threshold ($\Rtop=0.28\to1.22$; \cref{fig:v1}, \cref{tab:v1}). The edge term carries 17--50\% of the linear FDR. With it, $c=0.80$--$0.89$ for $\eta\lambda\le10^{-3}$. The edge term thus removes the excess. The remaining 10--20\% is the finite-step remainder (\cref{sec:v2}). (ii)~With memory and with AdamW the sharpest mode stays far from the threshold ($\Rtop\le0.13$), and the edge term is negligible. In these runs $\|g\|^2$ is 3--37\% \emph{below} the continuous-time value. The ideal gas thus fails in opposite directions with and without memory.
\end{observation}

\paragraph{Radius.} For the thin ResNet with SGD, $r^2/r^2_{\mathrm{SDE}}=1.04$--$1.09$ for $\eta\lambda\le10^{-4}$, $1.12$--$1.16$ at $10^{-3}$ and $1.19$--$1.45$ at $2$--$5\cdot10^{-3}$. The prediction with the edge term is \eqref{eq:radial-sgd} with $\|g\|^2$ from \eqref{eq:fdr}. This is a consistency check at the level of second moments. The prediction is within 0.97--1.01 of the measurement up to $5\cdot10^{-4}$, 0.93--0.99 at $10^{-3}$ and 0.86--0.96 beyond. The exact one-step norm balance (\cref{prop:radial}) holds with median 1.004 (10--90\%: 0.981--1.033) over the 114 runs that learn (test accuracy $>0.2$). For the practical model it takes the general form \eqref{eq:radial-adam} with the cross term. The balance also holds for each invariant group of the practical ResNet (median 1.000, 10--90\% 0.93--1.09 over 2,928 groups: 240 channels and 4 jointly invariant groups per run).

\paragraph{Symmetry.} At fixed $\eta\lambda$ the radius scales as \cref{prop:symmetry} predicts. The ratios $r^2/\eta$ for SGD and $r^2/\eta^2$ for AdamW agree within 0.8--4.3\% over $\eta$ ranging by a factor 100. All these runs use seed 0, so the test is that the radius forgets its initial scale. Runs at the \emph{same} $(\eta,\lambda)$ with different seeds differ by 3--10\% in $r^2$ with SGD and memory (13\% for SGD at $\eta\lambda=10^{-2}$), against standard errors of 0.1--1\%. AdamW seeds agree within 2\%. Either the windows are too short for full equilibration, or independently initialized runs reach different stationary states. In either case \cref{cor:maxwell} applies to corresponding states.

\begin{table}[t]
\centering
\caption{Mode test, SGD, $\eta\lambda\le10^{-3}$ (thin ResNet, ConvNet and practical ResNet). Variance of the weight snapshots along the top-3 Hessian eigenvectors, relative to the discrete prediction $\eta\sigma_k^2/(2\mu_k(1-R_k))$ (``theory'') and to the continuous-time value $\eta\sigma_k^2/(2\mu_k)$. Here $\sigma_k^2=v_k^\top\Sigma v_k$.}
\label{tab:modes}
\small
\begin{tabular}{lcccc}
\toprule
$R_k$ & modes & var/theory, median & mean $\pm$ s.e. & var/continuum, median\\
\midrule
$[0,0.3)$ & 46 & 1.07 & $1.09\pm0.03$ & 1.32\\
$[0.3,0.5)$ & 47 & 0.98 & $1.03\pm0.03$ & 1.68\\
$[0.5,0.7)$ & 24 & 0.82 & $0.87\pm0.04$ & 1.96\\
$[0.7,0.95)$ & 11 & 0.50 & $0.50\pm0.06$ & 2.39\\
\bottomrule
\end{tabular}
\end{table}

\paragraph{Mode temperatures.} The top modes are hotter than the continuous-time prediction, by factors 1.3 to 2.4 that grow with $R$ (\cref{tab:modes}). The discrete law $T^{(0)}/(1-R)$ accounts for this up to $R\approx0.5$. Above, the modes are colder than predicted, by a factor 2 at $R=0.7$--$0.95$. The harmonic theory does not contain this saturation.

\paragraph{Other models.} In the SI ConvNet, $\|g\|^2/\mathrm{SDE}=1.18$--$2.51$. Its linear-FDR ratio falls from 1.10 to 0.46 as $\Rtop$ grows from 0.29 to 1.89. In the practical ResNet with SGD, $\|g\|^2/\mathrm{SDE}=0.68$--$1.51$ and the FDR ratio is 0.32--0.89. Up to 77 of 244 groups are dead there at the largest decay rates. In the thin ResNet with SGD, test accuracy is 0.66--0.71 for $\eta\lambda\le10^{-3}$ and collapses beyond. The range beyond is the third regime of \citet{kodryan2022three}, where $\Rtop>1$ and the harmonic theory does not apply.

\subsubsection{The finite-step remainder: activation boundaries and smooth nonlinearity}\label{sec:v2}
\begin{figure}[t]
\centering
\includegraphics[width=\linewidth]{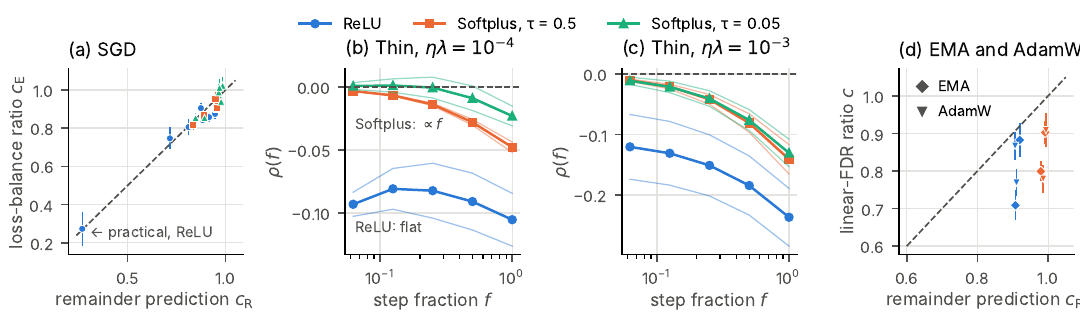}
\caption{Campaign v2. (a)~SGD: the remainder prediction $\cR$ against the measured loss-balance ratio $\cE$, for all 20 SGD runs (thin and practical ResNet; ReLU and Softplus). (b,c)~Remainder profiles $\rho(f)=\Rem(f)/(\frac{f^2}2\Delta\theta^\top H\Delta\theta)$ of the thin ResNet with SGD at $\eta\lambda=10^{-4}$ and $10^{-3}$ (thin lines: individual runs; thick: means over seeds). ReLU: $\rho$ flat in $f$ (activation boundaries, $\Rem\propto f^2$). Softplus: $\rho\propto f$ ($\Rem\propto f^3$). (d)~Memory and AdamW: the linear-FDR ratio lies well below the single-step remainder prediction.}
\label{fig:v2}
\end{figure}

\begin{table}[t]
\centering
\caption{Campaign v2, SGD, seed means (2 seeds for the thin ResNet, 1 for the practical one). $\cE$: loss-balance ratio. $\cR=1+\rho(1)$: its prediction from the remainder. $\rho(1)$, $\rho(\frac1{16})$: remainder profile. Flips: fraction of units that change sign along the full step. $\chi$: RMS pre-activation change in units of $\tau$.}
\label{tab:v2}
\small\setlength{\tabcolsep}{4.5pt}
\begin{tabular}{llccccccccc}
\toprule
model & $\eta\lambda$ & activation & $\cE$ & $\cR$ & $\rho(1)$ & $\rho(\frac1{16})$ & $\Rtop$ & flips & $\chi$ & test acc.\\
\midrule
thin & $3\cdot10^{-5}$ & ReLU & 0.858 & 0.914 & $-0.086$ & $-0.071$ & 0.36 & 1.3\% & -- & 0.678\\
 & & Softplus 0.5 & 0.909 & 0.959 & $-0.041$ & $-0.004$ & 0.31 & 1.5\% & 0.11 & 0.714\\
 & & Softplus 0.05 & 0.963 & 0.971 & $-0.029$ & $-0.002$ & 0.31 & 1.4\% & 1.02 & 0.680\\
\addlinespace
 & $10^{-4}$ & ReLU & 0.881 & 0.895 & $-0.105$ & $-0.093$ & 0.48 & 2.1\% & -- & 0.706\\
 & & Softplus 0.5 & 0.929 & 0.952 & $-0.048$ & $-0.003$ & 0.37 & 2.5\% & 0.17 & 0.731\\
 & & Softplus 0.05 & 1.017 & 0.977 & $-0.023$ & $+0.001$ & 0.33 & 2.4\% & 1.71 & 0.703\\
\addlinespace
 & $10^{-3}$ & ReLU & 0.776 & 0.763 & $-0.237$ & $-0.120$ & 0.91 & 6.3\% & -- & 0.699\\
 & & Softplus 0.5 & 0.838 & 0.859 & $-0.141$ & $-0.010$ & 0.70 & 6.5\% & 0.42 & 0.732\\
 & & Softplus 0.05 & 0.854 & 0.870 & $-0.130$ & $-0.011$ & 0.50 & 5.6\% & 3.81 & 0.710\\
\midrule
practical & $10^{-4}$ & ReLU & 0.273 & 0.268 & $-0.732$ & $-0.600$ & 0.24 & 3.4\% & -- & 0.801\\
 & & Softplus 0.5 & 0.866 & 0.891 & $-0.109$ & $-0.008$ & 0.45 & 5.2\% & 0.10 & 0.752\\
\bottomrule
\end{tabular}
\end{table}

\begin{observation}[The remainder accounts for the SGD loss balance]\label{obs:closure}
Over the 20 SGD runs (\cref{tab:v2}, \cref{fig:v2}a), the independently estimated terms satisfy the loss balance with an RMS discrepancy $\cR-\cE$ of 0.038 (mean $+0.015$). The median $|z|$-score is 1.12 ($\chi^2/n=1.85$). The correlation is 0.974 over the range $\cE=0.27$--$1.02$ (0.87 without the practical ReLU run). For the two practical runs, which are not fully scale-invariant, both ratios use the general form \eqref{eq:energy-general}. Since $\chi^2/n=1.85$, either our standard errors are about 35\% too small or a small additional effect is present.
\end{observation}

\begin{observation}[Activation boundaries and smooth nonlinearity]\label{obs:kink}
In the thin ResNet with ReLU, $\rho(f)$ is flat in $f$ at small $\eta\lambda$ (seed means $-0.07$ to $-0.09$ at $f=\frac1{16}$ and $-0.09$ to $-0.11$ at $f=1$). The remainder is thus $\propto f^2$, an activation-boundary term. With Softplus, $\rho$ halves with every halving of $f$ ($\Rem\propto f^3$) and vanishes as $f\to0$ ($|\rho(\frac1{16})|\le0.011$). This holds for $\tau=0.5$ ($\chi\approx0.1$--$0.4$) and for $\tau=0.05$ ($\chi\approx1$--$3.8$, steps much larger than $\tau$). \Cref{lem:kink} predicts this for any fixed width. At $\eta\lambda=10^{-3}$ ReLU has both: a boundary plateau $-0.12$ and a smooth part reaching $-0.24$ at $f=1$. The singular curvature is \emph{negative} here. The quadratic model thus overestimates the curvature along the step by 7--12\%.
\end{observation}

In the practical ResNet with ReLU the effect is extreme. Here $\rho$ lies between $-0.56$ and $-0.73$ at every $f$, and $\cE=0.273$ against $\cR=0.268$. The quadratic model overestimates the curvature along the step by a factor of about four, and the remainder accounts for it. With Softplus, $\rho(\frac1{16})=-0.008$ and $\cE=0.866$ against $\cR=0.891$. In this ReLU run the Rayleigh quotient along the mean gradient corresponds to $\eta h_g/2=2.66$, while Lanczos gives $\Rtop=0.24$. This is impossible for a fixed matrix. Lanczos on eight batches may miss the top of the averaged spectrum, or $\|g\|^2$ is underestimated. We exclude this run from spectral statements.

\subsubsection{Gradient memory: the departure and its direct accounting}\label{sec:momentum}
With memory and with AdamW the linear-FDR ratio lies well below the single-step remainder prediction (\cref{fig:v2}d). For the eight v2 runs, $\cR=0.91$--$0.99$ against $c=0.71$--$0.91$. Over the 30 EMA runs of v1 and v2, the first-moment relations $\E[g\cdot m]=-\kappa\E\|g\|^2$ (\cref{thm:finite}) and $\E[g^\top Hm]=-\kappa\E g^\top Hg$ hold within the noise (median ratios 0.98 and 0.95). In contrast, $\E[m^\top Hm]$ falls below its linear value by about the factor by which $\|g\|^2$ falls below the linear FDR (correlation 0.88 over the 24 runs with $\eta\lambda\le10^{-2}$). In addition, the per-mode factor of \cref{lem:geo} is $\ge0.991$ for $\eta\lambda\le10^{-3}$. By \eqref{eq:F2}, the departure therefore sits in the update curvature $U$.

\begin{figure}[t]
\centering
\includegraphics[width=\linewidth]{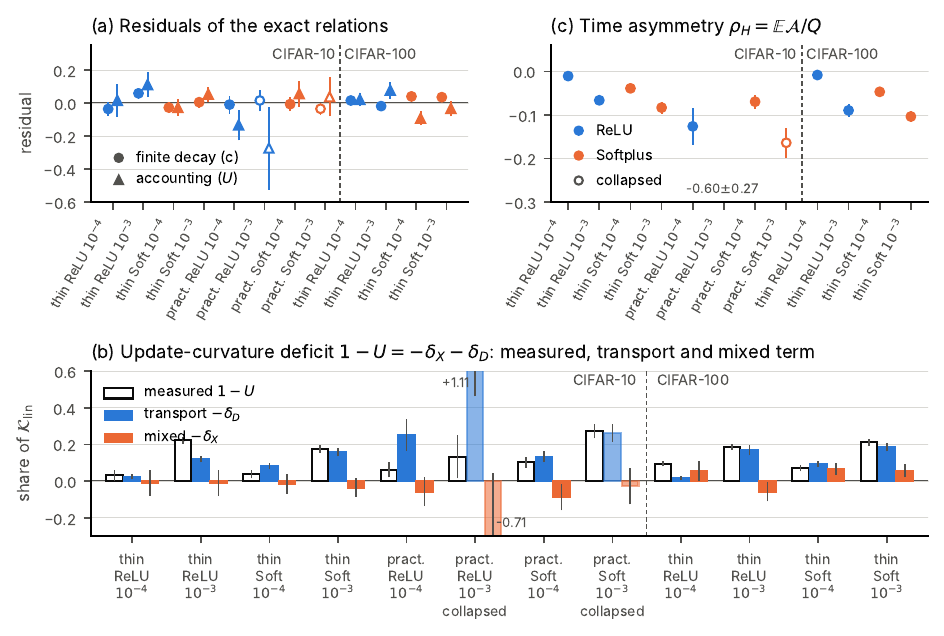}
\caption{Direct endpoint measurements with gradient memory (campaigns v3, CIFAR-10, and v4, CIFAR-100; $\beta=0.9$). One point per setting; seed means for the thin ResNet. Open symbols and faded bars: runs with collapsed channels. (a)~Residuals of the finite-decay relation in window form (as in \cref{tab:endpoint}) and of the accounting \eqref{eq:accounting} in units of $\cK_{\mathrm{lin}}$. (b)~Accounting of the update-curvature deficit: measured $1-U$, transport $-\delta_D$ and mixed term $-\delta_X$. (c)~Signed time asymmetry $\rhoH=\E\cA/Q$. Error bars: one standard error, as in \cref{tab:endpoint} (joint bootstrap). For seed means, the error is the larger of the propagated and the seed-scatter error.}
\label{fig:endpoint}
\end{figure}

\begin{table}[t]
\centering
\caption{Direct endpoint measurements with gradient memory ($\beta=0.9$; thin SI ResNet-18: means of two seeds on CIFAR-10, three on CIFAR-100, with the larger of the propagated and the seed-scatter error; practical ResNet-18: one run each; joint bootstrap standard errors in parentheses, in units of the last digit). Residual of \eqref{eq:F2} (thin) or \eqref{eq:finite} with $\lambda\bar e$ (practical), relative to its left-hand side; deficit $1-U$, transport $-\delta_D$ of \eqref{eq:accounting} and the gap $(1-U)-(-\delta_D)$ it leaves; $\rho_0$, $\rho_1$, $\rhoH$ of \eqref{eq:rhos}, ratios averaged over runs with $\rho_0$ resolved ($>3$ s.e.; -- if the mean $\rho_0$ is not resolved or changes sign between seeds). Softpl.: Softplus 0.5. $^*$Collapsed channels.}
\label{tab:endpoint}
\footnotesize\setlength{\tabcolsep}{2.4pt}
\begin{tabular}{lcccccccc}
\toprule
setting, $\gamma$ & $c$ & residual & $1-U$ & $-\delta_D$ & gap & $\rho_1/\rho_0$ & $\rhoH/\rho_0$ & $\rhoH$\\
\midrule
\multicolumn{9}{l}{\emph{CIFAR-10}}\\
thin ReLU $10^{-4}$ & $0.830(25)$ & $-0.035(41)$ & $0.032(29)$ & $0.024(14)$ & $0.008(30)$ & $2.02(8)$ & $0.1(1)$ & $-0.011(7)$\\
thin ReLU $10^{-3}$ & $0.715(21)$ & $0.059(33)$ & $0.224(20)$ & $0.122(18)$ & $0.102(27)$ & $2.17(6)$ & $0.5(1)$ & $-0.066(11)$\\
thin Softpl. $10^{-4}$ & $0.925(23)$ & $-0.027(33)$ & $0.039(23)$ & $0.081(16)$ & $-0.042(34)$ & $2.84(3)$ & $5.0(2)$ & $-0.039(8)$\\
thin Softpl. $10^{-3}$ & $0.808(18)$ & $0.006(33)$ & $0.178(22)$ & $0.159(23)$ & $0.019(34)$ & $2.90(2)$ & $5.4(2)$ & $-0.083(14)$\\
pract. ReLU $10^{-4}$ & $0.791(23)$ & $-0.009(54)$ & $0.062(41)$ & $0.253(88)$ & $-0.191(112)$ & $2.12(10)$ & $0.9(2)$ & $-0.126(42)$\\
pract. ReLU$^*$ $10^{-3}$ & $0.651(40)$ & $0.016(66)$ & $0.134(118)$ & $1.114(651)$ & $-0.980(768)$ & -- & -- & $-0.599(274)$\\
pract. Softpl. $10^{-4}$ & $0.875(22)$ & $-0.006(43)$ & $0.103(30)$ & $0.134(31)$ & $-0.031(46)$ & $2.83(3)$ & $5.0(2)$ & $-0.069(16)$\\
pract. Softpl.$^*$ $10^{-3}$ & $0.678(49)$ & $-0.034(38)$ & $0.273(39)$ & $0.262(49)$ & $0.011(58)$ & $3.13(18)$ & $6.7(11)$ & $-0.164(34)$\\
\midrule
\multicolumn{9}{l}{\emph{CIFAR-100}}\\
thin ReLU $10^{-4}$ & $0.910(26)$ & $0.016(21)$ & $0.094(13)$ & $0.017(13)$ & $0.077(19)$ & -- & -- & $-0.008(7)$\\
thin ReLU $10^{-3}$ & $0.782(21)$ & $-0.018(28)$ & $0.188(17)$ & $0.168(27)$ & $0.020(38)$ & -- & -- & $-0.089(15)$\\
thin Softpl. $10^{-4}$ & $0.957(18)$ & $0.040(29)$ & $0.070(17)$ & $0.094(15)$ & $-0.024(26)$ & $3.40(7)$ & $8.4(4)$ & $-0.047(7)$\\
thin Softpl. $10^{-3}$ & $0.798(19)$ & $0.036(31)$ & $0.212(21)$ & $0.187(24)$ & $0.025(40)$ & $3.38(7)$ & $8.3(4)$ & $-0.103(12)$\\
\bottomrule
\end{tabular}
\end{table}

\paragraph{Direct endpoint measurements.} Campaigns v3 (CIFAR-10) and v4 (CIFAR-100) evaluate every term of \eqref{eq:finite} and \eqref{eq:transport} at both ends of updates sampled from the optimizer's conditional transition law at the saved states (\cref{fig:endpoint}, \cref{tab:endpoint}).

\begin{observation}[Exact relations and the update curvature]\label{obs:endpoint-exact}
(i)~The finite-decay relation holds in window form, with the drifts of the actual trajectory. All twelve v3 residuals lie within 1.6 standard errors ($\chi^2=8.5$ for 12 runs). The v4 residuals give $\chi^2=12.5$ for 12 runs (largest 2.3 standard errors).
(ii)~The update curvature measured at the endpoints agrees with the Gram estimate of the earlier campaigns. Their difference lies between $-0.007$ and $+0.017$ in v3 and within $\pm0.006$ in v4 ($\chi^2=7.0$ and 4.9 for 12 runs each). The two estimators thus agree within the measured uncertainty, and we find no evidence of a bias of the Gram estimator at this resolution. This supports the interpretation that the discrepancy found in campaign v2 arose from the cubic-remainder reconstruction of the profile rather than from the Gram estimate (\cref{app:cubic}). The v2 runs themselves were not re-measured.
(iii)~The accounting \eqref{eq:accounting} holds with $\chi^2=16.6$ and $20.6$ for 12 runs each (largest residual 2.5 standard errors). Its drift term uses the window drift $d_{\cM}$, a least-squares slope over the diagnostic states (\cref{app:estimators}).
(iv)~The windows are compatible with stationarity at the measured resolution, with the qualifications that follow. The finite-window identities do not require stationarity, so their accuracy does not establish it. We therefore regress the quantities entering the identities ($G$, $\Xi_H$, $B_H$, $\cM$, $Q$, $\cK$, $\E\cA$, $c$, $U$) on time. The largest single drift is 2.7 standard errors. The sum of the 108 squared standardized drifts is 109 in v3 and 124 in v4. These drifts are measured on the same states and are correlated, so the sum is a diagnostic score, not a $\chi^2$ statistic. Calibrated with the joint bootstrap, its 95\% point under the hypothesis of no drift is about 200 ($p=0.37$ and 0.26; \cref{app:estimators}). In v4, however, $G$, $\Xi_H$, $B_H$, $\cM$, $Q$ and $\cK$ rise together. Combined over the 12 independent runs, each of these drifts lies 2.2--3.1 standard errors above zero, with a median rise of 4--19\% over the window. This residual sharpening is a quarter to a half of that in the $4\cdot10^5$-step run. The drifts of $c$ and $U$ are not resolved (1.1 and 0.8 standard errors, combined). The CIFAR-100 seeds also scatter beyond their errors (gap: $\chi^2=18.7$ for 8 degrees of freedom); the seed-mean errors include this scatter. In v3 the change of $U$ between the two halves of the windows exceeds its errors ($\chi^2=23$ for 12 runs, $p\approx0.03$). This is a $2\sigma$ indication that $U$ varies slowly within the windows. In v4 neither $U$ nor $c$ changes between the halves ($\chi^2=6.6$ and 9.5).
\end{observation}

\begin{observation}[Accounting of the update-curvature deficit]\label{obs:endpoint-account}
The transport contribution is negative in all twelve settings. In seven of the ten settings without collapsed channels it lies 4.4--7.8 standard errors from zero. It lies at 2.9 standard errors in the practical ReLU network at $\gamma=10^{-4}$, and at 1.7 and 1.4 in the thin ReLU networks at $\gamma=10^{-4}$ (\cref{fig:endpoint}b). In the Softplus networks the transport is consistent with the deficit on both datasets. The gap $(1-U)-(-\delta_D)$ of \cref{tab:endpoint} lies within 1.3 joint standard errors of zero in all six Softplus settings. In the thin network on CIFAR-10 at $\gamma=10^{-4}$ the deficit itself is not resolved, and the transport is about twice it. In the ReLU networks a resolved gap remains in one setting per dataset. It is $0.102\pm0.027$ at $\gamma=10^{-3}$ on CIFAR-10 (transport $0.122\pm0.018$ of $0.224\pm0.020$) and $0.077\pm0.019$ at $\gamma=10^{-4}$ on CIFAR-100 (transport $0.017\pm0.013$ of $0.094\pm0.013$). By \eqref{eq:accounting}, this gap equals the mixed contribution $-\delta_X$. On CIFAR-100 its direct estimate is $-\delta_X=0.055\pm0.053$, consistent with the gap but not resolved. The other thin ReLU settings close within the errors. In the practical ReLU network at $\gamma=10^{-4}$ the transport exceeds the deficit. The mixed term is not resolved beyond 2.2 standard errors in any setting (standard errors 0.03--0.10; 0.75 in the collapsed ReLU run), and it limits the accounting. The drift term $d_{\cM}$ contributes less than $10^{-4}$. The individual runs show the same pattern, within the seed scatter noted in \cref{obs:endpoint-exact}. Eleven of the twelve Softplus runs have gaps within 1.5 standard errors of zero, and the twelfth lies at 2.2. The CIFAR-10 ReLU gap lies at 3.3 and 2.0 standard errors in its two seeds. The CIFAR-100 ReLU gap is resolved in two of the three seeds (4.6 and 5.1 standard errors; 1.8 in the third). Stationarity matters here. For ReLU at $\gamma=10^{-4}$ on CIFAR-100 the gap was $-0.035\pm0.045$ in the $4\cdot10^5$-step run of seed~0, which was still sharpening, and $0.092\pm0.020$ after $2\cdot10^6$ steps.
\end{observation}

\begin{observation}[Smooth and boundary-dominated remainders; time asymmetry]\label{obs:endpoint-remainder}
The endpoint ratios separate the two activations (\cref{tab:endpoint}, \cref{fig:overview}c). With Softplus, the runs with $\rho_0$ resolved give $\rho_1/\rho_0=2.83$--$2.91$ and $\rhoH/\rho_0=5.0$--$5.5$ on CIFAR-10 (3.1 and 6.7 in the collapsed run). On CIFAR-100 they give $3.35$--$3.46$ and $8.1$--$8.7$. The runs without collapsed channels obey $\rhoH=6\rho_1-12\rho_0$ to within 0.9\% of $\rhoH$ (\cref{app:cubic}). This is the relation of a remainder $af^3+bf^4$ along the step. Within this approximation the effective quartic fraction $b/(a+b)=\rho_1/\rho_0-3$ is $-0.17$ to $-0.09$ on CIFAR-10 and $0.35$--$0.46$ on CIFAR-100. Accordingly, the $U$ implied by the cubic-remainder approximation differs from the measured one by $-0.05$ to $+0.02$ on CIFAR-10 (at most 1.6 joint standard errors). On CIFAR-100 it differs by $0.077\pm0.030$ at $\gamma=10^{-3}$.
With ReLU on CIFAR-10, $\rho_1/\rho_0=2.0$--$2.2$ (setting means). This is close to the value 2 of a boundary remainder $\propto f^2$ (\cref{lem:kink}). The change of the regular curvature along the step is small against it ($\rhoH/\rho_0=0.1$--$0.9$). The cubic approximation, derived for smooth remainders, would give $U=0.36$--$0.42$ instead of the measured $0.78$--$0.97$. On CIFAR-100 the ReLU remainder differs in sign between two seeds at $\gamma=10^{-4}$ ($\rho_0=+0.034\pm0.008$ and $-0.033\pm0.012$, at 4.1 and 2.7 standard errors). It is not resolved in the other runs. At $\gamma=10^{-3}$ there, $|\rhoH|\gg|\rho_0|$, so boundary and smooth parts are not separated. In all ReLU runs the cubic-plus-quartic relation fails draw by draw.
The normalized mean $\rhoH=\E\cA/Q$ of the time-antisymmetric observable is negative in all twelve v3 runs and in eleven of the twelve v4 runs. It is resolved (3.3--9.2 standard errors) in seven of the ten v3 runs without collapsed channels and in nine v4 runs. The others are four of the five thin ReLU runs at $\gamma=10^{-4}$, the practical ReLU run at $\gamma=10^{-4}$ (2.98 standard errors) and one Softplus run at $\gamma=10^{-4}$ (\cref{fig:endpoint}c). The radial terms of \eqref{eq:A-decay} have the opposite sign and at most 13\% of the magnitude of $\E\cA$. For the bound \eqref{eq:I2}, $I_2\ge B=2(\E\cA)^2/\E\cA^2$, we take the second moment over individual transitions. Each transition uses one evaluation batch, whose noise adds to $\E\cA^2$ and makes $B$ conservative. We remove the upward bias of the squared mean. The resulting estimate of $B$ is 0.07--0.22 (v3) and 0.06--0.54 (v4) in the resolved runs without collapsed channels. A nominal one-sided 95\% lower confidence bound, built from two 97.5\% bootstrap bounds (\cref{app:estimators}), is 0.01--0.09 and 0.01--0.26.
\end{observation}

In the practical ResNet the head term $\lambda\bar e/G$ is $-0.10\%$ and $-0.09\%$ without collapsed channels ($-0.22\%$ and $-0.33\%$ with them). Half to two thirds of it comes from the normalization coefficients. It lies far below the 4--6\% precision of the window test. The non-invariant sector is therefore not resolved in these networks. The teacher--student test of \cref{thm:finite} resolves it. The two practical runs at $\eta=0.1$ have collapsed channels (8 and 119 groups with $r_a\to0$). There \eqref{eq:F2} does not apply, but \eqref{eq:finite} and \eqref{eq:transport} do. These two runs test the general identities only, not exact scale invariance of the hidden groups. Why the transport is negative in these networks but positive in the teacher--student network (\cref{sec:ts}) remains open.
\section{Discussion}\label{sec:status}

The thermodynamic-looking relations considered here have different origins. The entropy-scaling relation follows from an exact rescaling of the stochastic dynamics (\cref{prop:symmetry,cor:maxwell}). The equation of state and the fluctuation spectrum require additional assumptions about the noise and the local loss landscape. This distinction answers the question posed in the Introduction. The entropy relation can remain valid when the fixed-noise ideal-gas description is no longer adequate, because its validity does not rest on equilibrium. For corresponding stationary states, the scaling argument requires neither a Gibbs distribution nor detailed balance. It also survives a trained non-invariant output layer when the hyperparameters of that layer are held fixed under the rescaling.

The harmonic approximation provides a calculable reference for the fluctuations, rather than an equilibrium assumption for the full dynamics. Finite update steps modify the mode variances, and the resulting fluctuations feed back on the weight norms. For SGD with isotropic angular noise, this gives the spectral equation of state of \cref{sec:eos}. Its virial coefficients are set by moments of the curvature spectrum, and it contains no attractive term within the stated approximation. In the SGD networks with $\eta\lambda\le10^{-3}$, the discrete mode temperatures and the edge term account for the excess of the gradient power and of the radius over their ideal-gas values. Interacting groups require a joint covariance description. In the Gaussian harmonic surrogate of SGD with frozen radii, the law of consecutive states is reversible if and only if the transformed curvature and noise covariance commute (\cref{prop:coupled}). Exact norm balances can therefore coexist with stationary probability currents, and reversibility alone does not imply a common mode temperature.

Beyond the harmonic approximation, the principal result is a measurable decomposition of the fluctuation correction. The loss balance retains the finite-step remainder. With gradient memory, the curvature-transport identity separates the change of curvature along successive updates from the mixed gradient--memory contribution and the finite-window drift. These identities determine how the terms combine, but not their signs or magnitudes. Numerical agreement with an identity tests the estimators and the accounting of its terms. The physical information comes from measuring the individual contributions and from testing approximations that predict them. The identities thus determine how the correction is partitioned, the measurements determine which contribution matters, and the approximations determine whether that contribution can be predicted from fewer observables.

\begin{table}[t]
\centering\small
\caption{Main statements, the level at which they hold, and the numerical evidence.}
\label{tab:status}\footnotesize
\setlength{\tabcolsep}{4pt}\begin{tabular}{p{0.31\linewidth}p{0.25\linewidth}p{0.37\linewidth}}
\toprule
Statement & Level & Evidence\\
\midrule
Norm balance per invariant group (\cref{prop:radial}) & exact (invariant group) & networks: median 1.004; per group 1.000\\
Scaling symmetry, entropy relation (\cref{prop:symmetry}, \cref{cor:maxwell}) & exact (invariant groups; trained head at fixed hyperparameters) & sphere 4.000/16.007; networks 0.8--4.3\%; trained head $10^{-15}$\\
Loss balance, finite-decay identity (\cref{prop:energy}, \cref{thm:finite}) & exact (any network; common $\eta,\lambda,\beta$ for \cref{thm:finite}) & SGD networks: RMS 0.038 (20 runs); EMA networks, window form: $\chi^2=8.5$, 12.5 (12 runs each); teacher--student $0.9998\pm0.0009$ (24 runs); trained head $1.0000$--$1.0001$\\
Curvature transport, arrow of time (\cref{prop:transport,prop:antisym}) & exact (any network; common $\beta$) & teacher--student: within 0.16\%, $\E\cA$ at 13--277 s.e.; EMA networks: accounting $\chi^2=17$, 21 (12 runs each), $\rhoH<0$ resolved in 16 of 22 runs\\
Activation-boundary remainder $\propto f^2$ (\cref{lem:kink}) & asymptotic (small steps, smooth pre-activation density) & thin ResNet profiles (all but one small, non-monotone Softplus profile)\\
Singular curvature of ReLU committees (\cref{prop:relu}) & exact (Gaussian inputs) & fixed states within $4\cdot10^{-3}$; teacher--student $-2.131\%$ vs $-2.127\%$\\
Mode temperatures, spectral equation of state (\cref{prop:modes}, \cref{sec:eos}) & harmonic approximation & sphere radius 0.03\%; teacher--student 0.1--2\%; networks up to $R\approx0.5$\\
Geometric decay effects, coupled groups (\cref{approx:geo}, \cref{prop:coupled}) & approximation & sphere modes 2.4\%; coupled variances 0.6--1.4\% at $\gamma=10^{-3}$ (up to 9\% at $10^{-2}$, fails with the narrow activation); currents resolved where predicted large\\
Remainder closures (\cref{sec:closure}) & approximation, smooth remainders & toy (soft modes); Softplus networks: \eqref{eq:quartic} within 1\%, $\rho_1/\rho_0=2.8$--$2.9$ (CIFAR-10), 3.4 (CIFAR-100); fails per draw for ReLU\\
\bottomrule
\end{tabular}
\end{table}

The endpoint measurements identify different contributions in the tested networks. In the convolutional networks the curvature transport is negative in every setting and resolved in most of them. In the Softplus networks it is consistent with the update-curvature deficit on both CIFAR datasets, within the measured uncertainty. In one ReLU setting per dataset, a resolved part of the deficit remains after the transport contribution is subtracted. The identity assigns this residual to the mixed term, whose direct estimate is less precise. The teacher--student comparison instead resolves a small positive anharmonic transport correction for long memory ($\beta\ge0.9$), after the narrow--wide comparison has approximately removed the leading geometric contribution. The sign of the nonlinear correction is therefore not universal. The present measurements identify its contribution without yet determining which architectural or statistical features set its sign.

The comparison between CIFAR-10 and CIFAR-100 also separates the exact transport relation from the simpler remainder approximations. The cubic-remainder approximation is approximately consistent with the smooth-network measurements on CIFAR-10. On CIFAR-100, the endpoint ratios instead indicate a substantial quartic contribution within the cubic-plus-quartic approximation, \eqref{eq:quartic}. This does not establish that the remainder is a polynomial of degree four: the endpoint relation constrains a particular combination of higher-order contributions (\cref{app:cubic}). It nevertheless provides a direct test of the approximation and identifies a limitation that is not visible from the final fluctuation ratio alone.

Activation boundaries require separate treatment. In the Gaussian committee calculation, their curvature contribution is set by the residual on the activation boundaries (\cref{prop:relu}). Competing student and teacher terms give it no universal sign, and it vanishes at an exactly matched teacher. It is negative in the batch-normalized ReLU networks on CIFAR-10 and in the ReLU teacher--student network, and positive in the weight-normalized pilot (\cref{app:pilot}). On CIFAR-100, the full finite-step remainder changes sign between two ReLU runs, but its smooth and boundary contributions have not been separated. The regular-Hessian transport also combines motion within activation regions with changes of the activation pattern. These observations limit what can be inferred from a quadratic model evaluated at a single state.

The pair observable $\cA$ supplies a distinct diagnostic of time-reversal asymmetry. Its mean is resolved in the teacher--student runs and in most of the network runs. Its mean and second moment give a lower bound on the relative entropy between forward and reversed pairs. At zero decay $\E\cA=\eta^2\DH$; at finite decay, the terms of \eqref{eq:A-decay} distinguish the two quantities. The bound applies to the sampled pair distribution, including a distribution averaged over a finite observation window. It is not a measurement of the thermodynamic entropy production of the complete optimizer. In particular, finite-window asymmetry does not by itself show that a fully stationary distribution has been reached.

\paragraph{Limitations.} The convolutional networks are small (24k--176k parameters), and the experiments cover two image-classification datasets with specific normalization and output-layer conventions. The fixed head on CIFAR-100 keeps the loss high, and the change between the datasets also involves a change in the typical head-row scale. The network mode tests use only the three sharpest measured directions. The endpoint measurements cover EMA momentum only. For AdamW the identities hold in the frozen-preconditioner approximation. The head term of the practical network is below the measurement resolution, although it is resolved in the teacher--student test. \Cref{prop:spec-entropy} gives an analytical Gaussian entropy reference, but the entropy of the sampled network distributions, central in \citet{sadrtdinov2025idealgas}, was not estimated here.

The numerical uncertainty also limits the conclusions. The early campaigns v1 and v2 use delta-method error estimates over correlated diagnostic states, whereas the endpoint campaigns v3 and v4 use a joint block bootstrap. The normalized ratios $c$ and $U$ show little resolved drift over the windows. Several unnormalized curvature statistics, however, continue to increase on CIFAR-100, and the CIFAR-100 seeds scatter beyond their errors. We therefore retain the finite-window drift terms and do not claim a fully stationary distribution in those runs. The contributions to the accounting are more stable under the sensitivity checks of \cref{app:estimators} than the standardized residuals of the identities, which reach 3.6 and 5.5 standard errors in individual variants. The mixed term of \eqref{eq:accounting} remains the least precisely measured contribution. Two of the 24 teacher--student runs of the coupled-group test did not finish within the compute window (\cref{sec:ts}).

\paragraph{Open questions.} (1)~\emph{What sets the sign and size of the curvature transport?} Batch statistics, the cross-entropy loss far from its minimum, depth and a finite training set are possible contributors, but the present models do not isolate their roles. For ReLU networks, separating the transport within smooth regions from the transport associated with changes of the activation pattern would sharpen the interpretation. A more precise direct estimate of the mixed term would further test the contribution inferred from the residual. (2)~\emph{How do modes saturate near the threshold?} In the networks, modes with $R>0.5$ are colder than the harmonic prediction. The harmonic prediction loses applicability when the fluctuations explore nonlinear regions of the loss. Which nonlinear terms determine the saturation remains to be established. The normal form of a noisy period-doubling bifurcation is one candidate description; its applicability and coefficients remain to be measured. (3)~\emph{How large is the irreversibility in reduced and full coordinates?} The pair observable for the directions alone, and a reversal of the full state $(\theta,m)$, would relate the bound \eqref{eq:I2} to the entropy production of the optimizer. Projecting to directions removes the radial support obstruction of the SGD pair law, but a finite reduced relative entropy still requires sufficient overlap of the forward and reversed laws. Extending the analysis to the full optimizer requires a specified reversal operation for the memory variable. Probability currents between coupled groups also remain to be measured directly in the convolutional networks.

Taken together, the results support a description of stationary and slowly evolving learning dynamics that does not assume thermal equilibrium. Scale symmetry fixes some relations exactly. Harmonic theory predicts the fluctuations within a stated regime. Curvature transport, finite-step remainders and activation-boundary contributions quantify the departures from that regime, and the pair asymmetry provides a separate test of irreversibility. This distinction identifies both the usefulness and the limits of a thermodynamic description of learning.

\section*{Data, software and code availability}
The source code, configurations, numerical results and reproduction scripts underlying this study are maintained in a versioned private repository. A reproducibility archive corresponding to the submitted manuscript is available confidentially to the editors and referees on request through the editorial office. A curated public release will be made available upon acceptance. Its permanent identifier will be added to the accepted manuscript and to an updated arXiv version. \Cref{app:numerics} lists every numerical calculation behind the results, with its setup and the script that produced it. The network campaigns ran on Snellius (A100 GPUs). The teacher--student runs, the sphere models and the pilot of \cref{app:pilot} ran on CPUs.

\section*{Acknowledgements}
\paragraph{Use of generative AI tools.} The author used Claude Opus 5.5 (Anthropic) and GPT-6 Astra Pro (OpenAI) to assist with software development, symbolic and numerical checks, exploration of theoretical arguments, running and analysis of the numerical experiments, and preparation and revision of the manuscript. The author reviewed, checked and revised the AI-assisted material and takes responsibility for the derivations, numerical results, references and interpretations in the article. The verification procedures are documented in \cref{app:numerics} and in the accompanying code.

\paragraph{Computing resources.} This work used the Dutch national e-infrastructure with the support of the SURF Cooperative using grant no.~EINF-19390. Access to the supercomputer Snellius was provided under the NWO call ``Computation time national computer systems 2026''.

\renewcommand*{\bibfont}{\footnotesize}\setlength{\bibsep}{1pt}

\appendix
\section{Linear theory of gradient memory}\label{app:lyapunov}
One mode with curvature $h$ and noise variance $s$, state $(x_t,m_t)$:
\[
m_{t+1}=\beta m_t+(1-\beta)(hx_t+\zeta_t),\qquad x_{t+1}=x_t-\eta m_{t+1},
\]
so $z_{t+1}=Fz_t+b\,\zeta_t$ with $F=\begin{pmatrix}1-\eta(1-\beta)h&-\eta\beta\\(1-\beta)h&\beta\end{pmatrix}$ and $b=(-\eta(1-\beta),\,1-\beta)^\top$. The stationary covariance solves $C=FCF^\top+s\,bb^\top$. Solving symbolically,
\[
h^2C_{xx}=\frac{\eta hs(1+\beta)}{2(1+\beta)-\eta h(1-\beta)}=\frac{\eta hs/2}{1-R},\qquad R=\frac{\eta h}{\aeos},
\]
which gives \eqref{eq:modes} and, with $g=hx$, \eqref{eq:fdr}. Further, $\E[g(g-m)]=h^2C_{xx}-hC_{xm}=\frac{2}{1+\beta}\E g^2$. Hence $\E[gm]=-\kappa\,\E g^2$ and $\E[gHm]=-\kappa\,h^3C_{xx}$. Finally, $hC_{mm}=\kappa\,hs+\kappa^2\,h^3C_{xx}$. These are \eqref{eq:ema-identities}. They hold for every $0<R<1$, with $\gamma\to0$. For $\gamma>0$ the same computation for \eqref{eq:geo-ema} gives \eqref{eq:modes-q}.

\paragraph{Radial step and the remainder.} For a fully scale-invariant network without memory and $0\le\gamma<1$, $w_{t+1}=(1-\gamma)(w_t+\Delta_{\tan})$ with $\Delta_{\tan}=-\eta g_B/(1-\gamma)$. Scale invariance gives $L(w_{t+1})=L(w_t+\Delta_{\tan})$. The loss balance along $\Delta_{\tan}$ reads $d_L=-\eta G/(1-\gamma)+Q_{\tan}(1+\rho_{\tan})$ with $Q_{\tan}=\eta^2\E[g_B^\top Hg_B]/(2(1-\gamma)^2)$. At stationarity this is \eqref{eq:c-sgd}. With drift, $\rho_{\tan}$ is replaced by $(\E\Rem_{\tan}-d_L)/Q_{\tan}$. Expressed through the full step, $w+f\Delta w=(1-f\gamma)[w-\eta_fg_B]$ with $\eta_f=f\eta/(1-f\gamma)$. We insert $g\cdot w=0$, $w^\top H=-g^\top$ and $w^\top Hw=0$ into the Taylor terms of \eqref{eq:remainder}. The terms linear in $g\cdot g_B$ then cancel to $O(\gamma^2)$. The quadratic ones leave $\frac12(\eta_f^2-f^2\eta^2)\,g_B^\top Hg_B$. Therefore $\Rem(f)=f^3\gamma\,\eta^2\,g_B^\top Hg_B+\Rem_{\tan}(f)+O(\gamma^2)$ and $\rho_0=2\gamma+O(\gamma\rho_{\tan},\gamma^2)+\rho_{\tan}$. The full-step remainder thus contains a geometric part of order $\gamma$, even for a loss that is harmonic along the step. In the exact sphere simulation of \cref{sec:geo} ($\beta=0$), $\rho_0=(2.01\text{--}2.03)\gamma$. Hence $\rho_{\tan}$ is negligible there. In the same simulation, $c-1=(1.00\pm0.02)\gamma$ at $\gamma=10^{-2}$ (also with twice the edge ratios) and $(0.96\pm0.07)\gamma$ at $3\cdot10^{-3}$. With memory the radial factor becomes $1-\varepsilon_t$ with $\varepsilon_t=\gamma+\eta\beta\,\varpi_t/r_t^2$, $\varpi_t=w_t\cdot m_t$. With $a_t=\eta\varpi_t/r_t^2$, the update gives the exact recursion $(1+a_{t+1})\,r_{t+1}^2/r_t^2=(1-\gamma)^2-qa_t$. At stationarity, taking logarithms and expanding to leading order in $\gamma$ and in the relative step gives $\E a_t=-2\gamma/(1+\beta)+O(\gamma^2)$, so that $\varepsilon_t$ has mean $\kappa\gamma+O(\gamma^2)$. This expansion is an approximation; \cref{prop:radial} alone fixes $\E\varpi_t/\E r_t^2$, not $\E[\varpi_t/r_t^2]$. However, $\varepsilon_t$ is correlated with the step length. The estimate $\rho_0\approx2\kappa\gamma$ holds in the simulation for $\beta\le0.5$ ($1.9$--$2.0\,\kappa\gamma$). It fails for longer memory, in a way that depends on the dimension. With 8 tangent dimensions the simulation gives $1.4\,\kappa\gamma$ at $\beta=0.75$ and $|\rho_0|\le0.06\gamma$ at $\beta\ge0.9$. With 32 it gives $1.5\,\kappa\gamma$ and $0.95\,\kappa\gamma$ at $\beta=0.9$ and 0.95. This is why we do not give a closed first-order formula for $c$ with memory.

\paragraph{Radial part of the memory.} $\E[g\cdot m]$ sees only the tangent part of the memory, because $g\perp w$ exactly. It agrees with \eqref{eq:ema-identities} to 0.3\% in all teacher--student runs. $\E[g^\top Hm]$ also sees the radial part, because $Hw_a=-g_a$ for every scale-invariant group. The radial part contributes $-\sum_a\E[(w_a\cdot m_a)\|g_a\|^2/r_a^2]$. Here $\E[w_a\cdot m_a]=-\eta\E\|m_a\|^2/(1-q)$ is fixed by the proof of \cref{prop:radial}. We neglect the fluctuations of $r_a$ and the correlation of $w_a\cdot m_a$ with $\|g_a\|^2$, i.e.\ we factorize the expectation as $\E[w_a\cdot m_a]\,\E\|g_a\|^2/\E r_a^2$. With the norm balance and $\E\|g\|^2\approx\frac\eta2\Tr\Sigma H$ this gives the approximate radial part of the mixed term
\begin{equation}\label{eq:dX-rad}
\delta_X^{\mathrm{rad}}\approx\frac{2\beta\gamma}{(1-\beta)(1+q)},
\end{equation}
0.0095 at $\beta=0.9$, $\gamma=10^{-3}$. The measured values (\cref{tab:ts-ema}) exceed it by 0.4--5.3\% with the wide and 3.6--14\% with the narrow activation. The tangent part of $\E[g^\top Hm]$ agrees with \eqref{eq:ema-identities} to within 1\% with the wide activation and to 0.8--3.1\% with the narrow one at $\gamma=10^{-3}$.

\section{Proofs}\label{app:proofs}

\paragraph{\Cref{thm:finite}.} Write $Y_0=\E[g\cdot m_t]$ and $Y_1=\E[g\cdot m_{t+1}]$, with $g=g(\theta_t)$, and $\Delta=\theta_{t+1}-\theta_t=-\gamma\theta-\eta m_{t+1}$. Unbiasedness of the fresh noise gives $Y_1=\beta Y_0+(1-\beta)G$. Expanding the population loss along $\Delta$ and using $g\cdot\theta=e$,
\[
d_L=-\gamma\bar e-\eta Y_1+Q+\bar\Rem_0,\qquad\text{i.e.}\qquad \eta Y_1=Q+\bar\Rem_0-d_L-\gamma\bar e .
\]
(For an evaluation batch independent of the step, the expectations of the batch remainder and of its endpoint derivative equal those of the population objective.) Let $\delta=g(\theta_{t+1})-g-H\Delta$, so that $\Rem'(1)=\delta\cdot\Delta$. Since $\eta m_{t+1}=-(\Delta+\gamma\theta)$,
\[
\eta\big(\Delta^\top Hm_{t+1}+\delta\cdot m_{t+1}\big)=-\Delta^\top H\Delta-\delta\cdot\Delta-\gamma\,(g_{t+1}-g)\cdot\theta,
\]
and $g_{t+1}\cdot\theta=e(\theta_{t+1})-g_{t+1}\cdot\Delta=e(\theta_{t+1})-g\cdot\Delta-\Delta^\top H\Delta-\delta\cdot\Delta$ with $g\cdot\Delta=-\gamma e-\eta g\cdot m_{t+1}$. The increment of $g\cdot m$ is $d_{gm}=Y_1+\E[\Delta^\top Hm_{t+1}+\delta\cdot m_{t+1}]-Y_0$. Inserting the last two relations gives
\[
\eta d_{gm}=\eta(1-\gamma)Y_1-\eta Y_0-(1-\gamma)(2Q+\bar\Rem_1)-\gamma\,d_e-\gamma^2\bar e .
\]
Eliminating $Y_0$ with $Y_1=\beta Y_0+(1-\beta)G$ and then $Y_1$ with the expansion of $d_L$ yields \eqref{eq:finite}. The coefficient of $\bar e$ is $-\gamma[(1-q)-\beta\gamma]=-\gamma(1-\beta)$. No Hessian identity is used. For a fully scale-invariant network $e\equiv0$, $\theta^\top H=-g^\top$ and $\theta^\top H\theta=0$. Hence $\E[\Delta^\top H\Delta]=\eta^2\cK-2\gamma\eta Y_1$. In a stationary state this gives $2Q=\eta^2\cK-2\gamma(Q+\bar\Rem_0)$, the second relation of \eqref{eq:F2}. Dividing \eqref{eq:finite} by $(1-\beta)\frac{\eta}{2}(\Xi_H+\kappa B_H)=\frac{\eta}{2\kappa}(1-\beta)\cK_{\mathrm{lin}}$ gives the first. For a realized trajectory the exact pathwise version replaces $G$ by $g\cdot g_B$ (the difference is a martingale term with zero mean) and the drifts by the telescoped differences. The pathwise identity holds to rounding error on a nonlinear loss without scale invariance and, in its scale-invariant form, over finite windows from a non-stationary start (\ncheck{identity}).

\paragraph{\Cref{approx:geo}: the value of $\E\omega$.} Project the update on $n^*$. This gives $n^*\!\cdot n_{t+1}=[(1-\gamma)\,n^*\!\cdot n_t-\eta_0\,m_{t+1}\cdot n^*]/\upsilon_t$ with $\upsilon_t=r_{t+1}/r_t$ (with $m=g_B$ for SGD). \Cref{prop:radial} gives $\E[r_t^2\upsilon_t^2]=\E r_t^2$. This fixes one moment, not the distribution of $\upsilon_t$. Under the approximation of narrow radius fluctuations and small steps, $\upsilon_t=1$ up to fluctuations of order $\theta^2$, the squared angle of one step. The memory is an average of past gradients. In a stationary state, therefore, $\E[m\cdot n^*]=\E[g\cdot n^*]=-\E\omega$. Stationarity of $\E[n\cdot n^*]$ then gives $\gamma\,\E[n\cdot n^*]=\eta_0\E\omega$. This holds up to the correlations of $\upsilon_t$ with $n_t$, which are of higher order. For SGD the same value follows from \cref{prop:modes,prop:radial}: $\eta_0\sum_jh_j\mathrm{Var}\,x_j=\frac12\eta_0^2\sum_js_j/(1-R_j)=\frac12\eta_0^2(\sum_js_j+\E\|g\|^2)=\frac12\E\theta^2$.

\paragraph{\Cref{prop:antisym}~(iv).} Let $T_\pi=\int(\pi-\pi^{\mathrm R})^2/(\pi+\pi^{\mathrm R})$. Exchanging the two arguments shows $D_{\mathrm{KL}}(\pi\|\pi^{\mathrm R})=D_{\mathrm{KL}}(\pi^{\mathrm R}\|\pi)$. Hence $I_2=\frac12\int(\pi-\pi^{\mathrm R})\log(\pi/\pi^{\mathrm R})\ge T_\pi$ by $(x-y)\log(x/y)\ge2(x-y)^2/(x+y)$. Antisymmetry gives $\int \pi^{\mathrm R}\cA=-\E\cA$, hence $\E\cA=\frac12\int(\pi-\pi^{\mathrm R})\cA$. Cauchy--Schwarz with the weight $\pi+\pi^{\mathrm R}$ then gives $(\E\cA)^2\le\frac14T_\pi\int(\pi+\pi^{\mathrm R})\cA^2=\frac12T_\pi\,\E\cA^2$. The derivation uses only the normalization of $\pi$ and the antisymmetry of $\cA$, not stationarity. \emph{Support constraint for scale-invariant SGD.} Let $0<\gamma<1$. Since $g_B\perp w_t$, every forward pair satisfies $w_t\cdot w_{t+1}=(1-\gamma)\|w_t\|^2$ (per group). Every pair in the support of the reversed law satisfies $w_t\cdot w_{t+1}=(1-\gamma)\|w_{t+1}\|^2$. The two sets meet only where $\|w_{t+1}\|=\|w_t\|$. If the radius changes with positive probability, $\pi$ puts positive mass where $\pi^{\mathrm R}$ vanishes. Then $\pi$ is not absolutely continuous with respect to $\pi^{\mathrm R}$, and $I_2=\infty$. If the radius changes with probability one, the two laws are mutually singular. Dynamics with constant radius are not covered. For example, the loss $L=y^2/(x^2+y^2)$ with $\gamma=2\sin^2u$, $0<u<\pi/4$, has the deterministic two-cycle $w_\pm=\sqrt\eta(\cos u,\pm\sin u)$, which is reversible and has $I_2=0$. Projecting to directions removes this radial obstruction. A finite reduced relative entropy still requires sufficient overlap of the forward and reversed laws. In addition, $\E\cA$ itself contains a contribution of order $\gamma$ from the correlation between the radius change and the size of the step.

\paragraph{\Cref{prop:relu}.} The population Hessian of $L=\frac12\E(f_w-f_*)^2$ is $\E[J_wJ_w^\top]+\E[(f_w-f_*)\nabla^2f_w]$. Here $\nabla^2f_w$ contains, besides the regular terms, $o_j\,\delta(n_j\!\cdot x)\,\nabla(n_j\!\cdot x)\nabla(n_j\!\cdot x)^\top$ with $\nabla_{w_j}(n_j\!\cdot x)=P_jx/r_j$. Hence $(\Hsing)_{jj}=(o_j/r_j^2)\,\E[(f_w-f_*)\,\delta(n_j\!\cdot x)\,(P_jx)(P_jx)^\top]$, and the off-diagonal blocks vanish. On $n_j\cdot x=0$ the $j$-th student unit vanishes, so $f_w-f_*=\sum_{i\ne j}o_i\mathrm{ReLU}(n_i\!\cdot x)-\sum_lo^*_l\mathrm{ReLU}(v_l\!\cdot x)$. Take a unit vector $v$ with $P_jv=\varsigma\,e$, $\varsigma=\varsigma_j(v)$. Condition on $n_j\cdot x=0$ (density $1/\sqrt{2\pi}$) and write $x=ye+z_\perp$ with $y\sim\mathcal N(0,1)$. Then $v\cdot x=\varsigma y$. With $\E[y\mathbf1_{y>0}]=1/\sqrt{2\pi}$ and $\E[y^3\mathbf1_{y>0}]=2/\sqrt{2\pi}$ it follows that
\begin{align*}
\E\big[\delta(n_j\!\cdot x)\,\mathrm{ReLU}(v\!\cdot x)\,(P_jx)(P_jx)^\top\big]&=\frac{\varsigma}{\sqrt{2\pi}}\Big(\E[y^3\mathbf1_{y>0}]\,ee^\top+\E[y\mathbf1_{y>0}]\,(P_j-ee^\top)\Big)\\
&=\frac{\varsigma}{2\pi}\big(P_j+ee^\top\big)=\frac{F_j(v)}{2\pi}.
\end{align*}
Summing over the units and the teacher gives \eqref{eq:R2}. A single unit gives \eqref{eq:R1}. An independent check differentiates the population loss $\frac{o^2}4+\frac{o_*^2}4-oo_*\,k(n\!\cdot v)$ with the arc-cosine kernel $k(\cos\theta)=[\sin\theta+(\pi-\theta)\cos\theta]/(2\pi)$ \citep{cho2009kernel} and subtracts the expected regular tangent Hessian.

\paragraph{Concentration at the threshold.} We state the result announced in \cref{sec:eos}.
\begin{proposition}[Concentration in the self-consistent harmonic approximation]\label{prop:condensation}
Consider \eqref{eq:thermostat} with $D$ tangent modes and angular noise $\sigma_{\mathrm a}^2=\sigma_0^2/D$, bulk curvatures $h_1\le\dots\le h_{D-1}$ and top curvature $h_D=h_*$. Assume (a) that the empirical distribution of the bulk curvatures converges weakly to a density $\varrho$. Assume also (b) that the top mode is uniformly isolated, $h_{D-1}\le h_*-\delta$ with $\delta>0$ independent of $D$. Let $I_c=\int\varrho(h)(1-h/h_*)^{-1}dh$ (finite by (b)) and $\gamma_c=\sigma_0^2\aeos^2I_c/(2h_*^2)$. The following holds as $D\to\infty$. For $\gamma<\gamma_c$, every mode carries a share $O(D^{-1})$ of $\sum_k(1-R_k)^{-1}$. For $\gamma>\gamma_c$, $1-\Rtop=O(D^{-1})$ and the top mode carries the share $1-\gamma_c/\gamma$. Consequently, $T_{\mathrm{top}}/T^{(0)}\simeq(1-\gamma_c/\gamma)\,2\gamma/(\sigma_{\mathrm a}^2\eta_{0*}^2)$ with $\eta_{0*}=\aeos/h_*$.
\end{proposition}
\begin{proof}
Divide \eqref{eq:thermostat} by $D$. The left-hand side $\ell(\eta_0)=2\gamma/(\sigma_0^2\eta_0^2)$ decreases. The right-hand side is $\Phi_D(\eta_0)=B_D(\eta_0)+D^{-1}(1-\eta_0h_*/\aeos)^{-1}$, with the bulk part $B_D(\eta_0)=D^{-1}\sum_{k<D}(1-\eta_0h_k/\aeos)^{-1}$. The right-hand side increases on $(0,\eta_{0*})$ and diverges at $\eta_{0*}$. The solution is therefore unique. By (b) the bulk integrand is bounded by $h_*/\delta$ on $[0,\eta_{0*}]$. By (a), therefore, $B_D\to B(\eta_0)=\int\varrho(h)(1-\eta_0h/\aeos)^{-1}dh$ uniformly there, with $B(\eta_{0*})=I_c$. If $\gamma<\gamma_c$, i.e.\ $\ell(\eta_{0*})<I_c$, the solution converges to the root $\bar\eta_0<\eta_{0*}$ of $\ell=B$. At this root every term of $\Phi_D$ is $O(D^{-1})$. If $\gamma>\gamma_c$, then $\ell(\eta_0)-B_D(\eta_0)\ge\ell(\eta_{0*})-I_c-o(1)>0$ for all $\eta_0\le\eta_{0*}$. The top term must therefore be of order one. Hence $1-\Rtop=O(D^{-1})$ and $\eta_0\to\eta_{0*}$. The share of the top term tends to $1-I_c/\ell(\eta_{0*})=1-\gamma_c/\gamma$.
\end{proof}
Without the gap (b), a limiting density alone does not single out one mode. Eigenvalues within $O(D^{-1})$ of $h_*$ contribute no mass to $\varrho$ but can share the excess. \emph{Remark.} Suppose that the bulk density extends to $h_*$ as $\varrho(h)\sim C(h_*-h)^{a_\varrho}$ with $a_\varrho>0$, so that $I_c<\infty$. Suppose also that the bulk sum converges to its integral up to the modes with $1-R_k=O(D^{-1})$. The same argument then shows that for $\gamma>\gamma_c$ the share $1-\gamma_c/\gamma$ is carried by this top spectral subspace. This subspace need not consist of a single eigenvector. For $a_\varrho\le0$, $I_c=\infty$ and there is no transition. The sphere model of \cref{sec:eos} has an isolated top mode.

\section{Estimators}\label{app:estimators}
\paragraph{Gradient moments.} At each diagnostic point the optimizer state is frozen and $N_B=400$ batch gradients $g_{B_i}$ are drawn. Here $\hat g=\frac1{N_B}\sum_ig_{B_i}$, and $\widehat{\Tr\Sigma}$ is the sample variance. We estimate $\|g\|^2$ by $\|\hat g\|^2-\widehat{\Tr\Sigma}/N_B$.
\paragraph{Curvature terms.} The term $\Tr\Sigma H=\E[\xi^\top H_{B'}\xi]$ is estimated with centred batch gradients and Hessian--vector products on independent batches, with the factor $N_B/(N_B-1)$. This is a Hutchinson-type estimator with the noise as the probe \citep{hutchinson1989trace}. The Gram matrix of $H$ in $\mathrm{span}(\hat g,m,\theta)$ (and $\hat g/P$ for AdamW) uses 64--128 Hessian--vector products on fresh batches. The bias $\Tr\Sigma H/N_B$ of $\hat g^\top H\hat g$ is subtracted.
\paragraph{Ratios.} In v1 and v2, $\cE$, the FDR ratio and the SDE ratio are ratios of the means over the 16 diagnostic points. Standard errors are from the point-to-point scatter (delta method). This error estimate ignores correlations between consecutive points.
\paragraph{Remainder.} At each diagnostic point the optimizer state is frozen, and 128 steps $\Delta\theta$ that the optimizer would take are drawn, each with its own minibatch. With memory and for AdamW, the step includes the current memory. For each step, $\Rem(f)$ is evaluated in float64 on an independent evaluation batch, together with $g_{B'}\cdot\Delta\theta$ and $\Delta\theta^\top H_{B'}\Delta\theta$ (forward-over-reverse differentiation). The float32--float64 discrepancy of the loss is also evaluated, as a numerical floor. For 16 of the steps the pre-activations at $w$ and $w+\Delta\theta$ give the flip fraction (units whose pre-activation changes sign), the RMS change $\Delta z$ and, for Softplus, $\chi$. The remainder profile is $\rho(f)=\overline{\Rem(f)}/(\frac{f^2}2\overline{\Delta\theta^\top H_{B'}\Delta\theta})$, with the means over steps and diagnostic points.
\paragraph{Endpoint quantities (campaigns v3 and v4).} For 64 draws per diagnostic point, a training batch $B$ gives $m_{t+1}$ and $\Delta=-\gamma\theta-\eta m_{t+1}$. One evaluation batch $B'$, independent of $B$, is used at $\theta$ and at $\theta+\Delta$. Losses, gradients and the quadratic forms of $H_{B'}$ in $\mathrm{span}(\theta,m_{t+1})$ at both ends, in float64, give $\Rem(1)$, $\Rem'(1)$, $\cA$, $\cK$, $\DH$, $G=g_{B'}\cdot g_B$ and $e$ by block. With these, \eqref{eq:finite} holds draw by draw to $10^{-15}$, a check of the code. The window tests use $G$, $\bar e$, $\Xi_H$, $B_H$, $\E[g^\top Hm_t]$ and $\cM$ from the stationary estimators.
\paragraph{Window drifts.} The exact finite-window identities contain telescoped increments $(\E X_{\mathrm{end}}-\E X_{\mathrm{start}})/T$ of the loss, of $g\cdot m$, of $e$ and of $\cM$. We estimate each increment by the least-squares slope of the state values over the 32 diagnostic states. The slope estimates the increment when the drift is approximately linear over the window. The drift terms are small; $d_{\cM}$, for example, contributes less than $10^{-4}$ to \eqref{eq:accounting}.
\paragraph{Uncertainties.} All v3 and v4 errors come from one joint bootstrap per run with 2,000 replicates. Each replicate resamples moving blocks of four consecutive diagnostic states, and the draws within each state. The resampled states keep their original times, and the window drifts are refitted to the resampled (time, value) pairs. The drifts thus keep their covariance with the other terms. Every quantity of a replicate, including the gap $(1-U)-(-\delta_D)$ and $X_4$ (\cref{app:cubic}), is computed from the same resample. A seed mean of $n$ runs carries the larger of the propagated error $\sqrt{\sum_i\mathrm{se}_i^2}/n$ and the standard error of the mean between seeds, $\mathrm{sd}/\sqrt n$.
\paragraph{Time-asymmetry bound.} For \eqref{eq:I2}, $\E\cA^2$ is averaged over individual draws. Each draw uses one evaluation batch, and the evaluation noise adds to $\E\cA^2$. This lowers $B$, so the bound is conservative. The squared mean is corrected by the bootstrap variance of $\E\cA$; this removes a bias but is not a confidence statement. For the lower confidence bound we combine two one-sided 97.5\% bounds: a lower bound for $|\E\cA|$, set to zero if its interval contains zero, and an upper bound for $\E\cA^2$. By the Bonferroni inequality, $B_{\mathrm{low}}=2\,(|\E\cA|_{\mathrm{low}})^2/(\E\cA^2)_{\mathrm{up}}$ is then a nominal one-sided 95\% lower confidence bound. Its actual coverage depends on the adequacy of the block-bootstrap approximation.
\paragraph{Stationarity diagnostics.} For each of $G$, $\Xi_H$, $B_H$, $\cM$, $Q$, $\cK$, $\E\cA$, $c$ and $U$ we compute the relative change over the window from the least-squares slope, with its standard error from the same bootstrap. The sum of the squared standardized drifts over these nine quantities and the runs of a campaign is a diagnostic score. The nine drifts of one run are correlated, so the score is not $\chi^2$-distributed. We calibrate it with the bootstrap: the null distribution of the score of one run is that of its centred replicates, and different runs are independent. A $\chi^2$ statistic is quoted only for a single quantity over independent runs.
\paragraph{Sensitivity.} We repeated the analysis with blocks of two and of eight states, and on the second half of each window (16 states, blocks of four). The variants and a summary table are in \path{paper-jstat/data/sensitivity} (\ncheck{endpoint}). With blocks of two or eight, the setting means of \cref{tab:endpoint} and the contribution estimates are stable. The Softplus gaps stay within 1.5 standard errors of zero, the two resolved ReLU gaps at 3.1--4.4, and the transport at 2.5 or more in the same eight settings. The endpoint ratios move by at most 0.02 ($\rho_1/\rho_0$) and 0.13 ($\rhoH/\rho_0$). The accounting residuals are more sensitive: with blocks of eight their $\chi^2$ is 22.3 (v3) and 36.3 (v4) for 12 runs each, and the largest is 3.6 standard errors, in one CIFAR-100 ReLU run. On the half windows the errors grow. The CIFAR-10 ReLU gap falls to 2.4 standard errors, and the CIFAR-100 one stays at 4.0. There the accounting residual of one CIFAR-100 ReLU run reaches 5.5 standard errors, and the finite-decay residual of the collapsed practical ReLU run 4.3. Blocks of eight and the half windows leave only four blocks per run, so the bootstrap errors are themselves uncertain. The CIFAR-100 seeds then scatter far beyond them (gap: $\chi^2=25$ and 55 for 8 degrees of freedom).
\paragraph{Eigenpairs and modes.} The eigenpairs come from Lanczos (24 iterations, top 3) on 8 fixed batches. Their Rayleigh quotients are re-measured on 32 fresh batches. Per-mode variances come from 1,000 snapshots taken every 25 steps in the stationary window. The snapshots are projected on the eigenvectors at the last diagnostic point (for AdamW in the preconditioned coordinates). The predicted variance uses $\sigma_k^2=v_k^\top\Sigma v_k$, estimated from the 400 batch gradients.

\section{Curvature transport: a toy check}\label{app:toy}
We run gradient memory ($\beta=0.9$, $\eta=1$, no weight decay) on the one-dimensional potential $L(x)=h\tau^2[\log\cosh(x/\tau)-\varepsilon x/\tau]$, with the noise set so that $\mathrm{Var}\,x\approx0.08\tau^2$. We use 20,000 chains (8,000 for skewed noise) of 20,000 steps after a burn-in of 4,000. For $\varepsilon=0$ the potential is symmetric (leading anharmonicity quartic). For $\varepsilon=0.3$ the minimum is shifted and a cubic term appears. The noise is Gaussian, or a centred exponential of the same variance (skewed). \Cref{tab:toy} compares the measured linear-FDR ratio $c$ with the exact $\gamma=0$ relation (\eqref{eq:decomp} with $U$ from \eqref{eq:transport}), evaluated with the measured $\rho_0$, $\rho_1$, $\rhoH$ and $\delta_X$. It also compares the measured ratio with the single-power form \eqref{eq:closure} ($n=3$; $n=4$ for $\varepsilon=0$).
\begin{table}[tb]\centering\small
\caption{Curvature-transport toy: measured $c$, the transport identity, the exact $\gamma=0$ relation and the single-power form.}\label{tab:toy}
\begin{tabular}{llccccccccc}
\toprule
noise & $\varepsilon$ & $\eta h$ & $c$ & $\cK$/rhs \eqref{eq:transport} & \eqref{eq:decomp} & $\rho_0$ & $\rho_1/\rho_0$ & $\rhoH/\rho_0$ & \eqref{eq:closure}\\
\midrule
Gaussian & 0.3 & 0.10 & 1.0003 & 1.0002 & 1.0002 & $+0.0003$ & 1.99 & $-0.06$ & 1.0037\\
 & & 0.76 & 1.0015 & 0.9998 & 1.0017 & $+0.0020$ & 1.99 & $-0.07$ & 1.0292\\
 & & 2.00 & 1.0042 & 0.9997 & 1.0045 & $+0.0054$ & 1.99 & $-0.08$ & 1.0825\\
 & & 4.00 & 1.0095 & 1.0001 & 1.0094 & $+0.0113$ & 1.98 & $-0.09$ & 1.1886\\
Gaussian & 0 & 0.76 & 1.0025 & 0.9999 & 1.0026 & $+0.0026$ & 2.00 & $-0.01$ & 1.0774\\
 & & 4.00 & 1.0144 & 1.0001 & 1.0143 & $+0.0149$ & 2.00 & $-0.02$ & 1.6641\\
\addlinespace
skewed & 0.3 & 0.03 & 1.0199 & 1.0004 & 1.0194 & $+0.0014$ & 2.93 & 5.57 & 1.0208\\
 & & 0.10 & 1.0278 & 1.0003 & 1.0272 & $+0.0022$ & 2.84 & 5.03 & 1.0323\\
 & & 0.30 & 1.0288 & 1.0003 & 1.0282 & $+0.0029$ & 2.64 & 3.82 & 1.0435\\
 & & 0.76 & 1.0226 & 1.0001 & 1.0227 & $+0.0041$ & 2.35 & 2.09 & 1.0611\\
 & & 2.00 & 1.0152 & 0.9992 & 1.0160 & $+0.0075$ & 2.09 & 0.53 & 1.1184\\
\bottomrule
\end{tabular}
\end{table}
The identities hold to $10^{-3}$ at every $\eta h$ and for both noises. With Gaussian noise, $\rho_1\approx2\rho_0$ and the transport is small but not zero ($\rhoH/\rho_0$ between $-0.01$ and $-0.09$; uncertainties not estimated). Here the single-power form fails. With skewed noise the ratios approach the cubic values for soft modes, where the approximation is accurate. It overshoots as $\eta h$ grows (\ncheck{toy}).
\section{Remainder closures in the networks}\label{app:cubic}
\paragraph{Cubic plus quartic.} If the remainder along the step is $\Rem(f)=af^3+bf^4$, then $\rho_1/\rho_0=3+\varphi$ and $\rhoH/\rho_0=6+6\varphi$ with $\varphi=b/(a+b)$. Hence $X_4\equiv\rhoH-6\rho_1+12\rho_0=0$ for every $\varphi$. $X_4$ is dimensionless. Its unnormalized counterpart for a single transition is, with $F(f)=L_{B'}(\theta+f\Delta)$ for the evaluation batch and a smooth loss,
\[
E_4=12[F(1)-F(0)]-6[F'(0)+F'(1)]+F''(1)-F''(0)=\frac12\int_0^1f^2(1-f)^2F^{(5)}(f)\,df ,
\]
in units of the loss. Its mean is $QX_4$. $E_4$ is twelve times the error of the corrected trapezoidal rule for $\int_0^1F'(f)\,df$, and $|E_4|\le\sup|F^{(5)}|/60$. A small $X_4$ therefore constrains one weighted moment of the higher derivatives; it does not show that the remainder is a polynomial of degree four.
In every Softplus run of v3 and v4 without collapsed channels, the window means give $|X_4/\rhoH|<0.009$. One run has $X_4$ nonzero at 4 standard errors, so the relation is accurate but not exact. For individual sampled transitions, the median of $|E_4|/|\cA|$ is 0.4--1.4\%; this ratio compares two single-transition quantities in the same units. In ReLU runs the window means also give $X_4\approx0$, within their larger errors. The pure boundary point $(\rho_1/\rho_0,\rhoH/\rho_0)=(2,0)$ also lies on the line. For individual transitions, however, the median of $|E_4|/|\cA|$ is 4--20 in ReLU runs. The relation is thus a property of smooth remainders.
The endpoint ratios of smooth networks lie on the one-parameter line (\cref{fig:overview}c). The effective quartic fraction $\varphi$ is between $-0.17$ and $-0.09$ on CIFAR-10 and between $0.35$ and $0.46$ on CIFAR-100 (runs with $\rho_0$ resolved). We use $\varphi$ only within this approximation; it is not an independently identified Taylor coefficient. The cubic approximation is the point $\varphi=0$. With a measured $\rho_1$, the relation $\rhoH=6\rho_1-12\rho_0$ enters $c$ through \eqref{eq:quartic-c} (fully scale-invariant networks, $\gamma\to0$).
\paragraph{Campaign v2.} Before the endpoint measurements, \eqref{eq:closure} was tested on the remainder profiles of v2. With Softplus 0.5 the profile is cubic up to $f=1$. A fit $\E\Rem(f)/Q=a_3f^3+a_4f^4$ to the five fractions gives $\rho_1/\rho_0=2.80$--$2.95$ and $\rhoH/\rho_0=4.8$--$5.7$. With $n=3$ and the measured $\rho(1)$, \eqref{eq:closure} gives $c=0.891$ and $0.768$ for EMA at $\eta\lambda=10^{-4}$ and $10^{-3}$ (measured $0.903\pm0.031$ and $0.799\pm0.027$); the first-order form \eqref{eq:amplification} overshoots at $10^{-3}$ ($c=0.708$). The implied $U=0.903$ and $0.792$ differed from the Gram estimates 0.954 and 0.865. In v3 the Gram estimate agrees with the endpoint $U$ within the errors, while the $U$ of the approximation differs from it by $-0.05$ to $+0.02$ with Softplus (\cref{sec:momentum}). Offsets of 0.05--0.07 are therefore plausible for the profile route. The v2 runs were not re-measured, so this explanation is supported but not proven.
\section{Teacher--student runs with gradient memory}\label{app:tstable}
\Cref{tab:ts-ema} lists the runs of \cref{fig:ts-momentum} (graded output weights $1,\frac12,\frac14,\frac18$; $B=16$, $\sigma=0.3$, 128 chains, 20,000 measured steps). The two rows of each $(\beta,\gamma)$ are Gaussian (G) and symmetric non-Gaussian (nG) label noise.
\begin{table}[h]
\centering\footnotesize
\caption{Teacher--student, gradient memory. Linear-FDR ratio $c$ for activation widths $\tau=1$ and $0.1$, measured $c$ over the exact relation \eqref{eq:F2} ($\tau=1$/$\tau=0.1$), and the per-mode factor of \cref{lem:geo}. The ratio $c_{\tau=0.1}/c_{\tau=1}$, measured and predicted from the remainder brackets and the memory term (\cref{obs:ts-ema}\,(iii)), and the memory term in percent. Radial part of $\delta_X$ ($\times10^3$, $\tau=1$/$\tau=0.1$), with \eqref{eq:dX-rad} in parentheses.}
\label{tab:ts-ema}
\footnotesize
\setlength{\tabcolsep}{3pt}
\begin{tabular}{lllccccccll}
\toprule
$\beta$ & $\gamma$ & noise & $c_{\tau=1}$ & $c_{\tau=0.1}$ & $c$/\eqref{eq:F2} & $\frac{1-\beta}{1-q}$ & ratio & pred. & mem.\,(\%) & $10^3\delta_X^{\rm rad}$\\
\midrule
0.5 & $10^{-4}$ & G & 0.9996 & 1.0040 & 0.9997/0.9998 & 0.9999 & 1.0045 & 1.0044 & +0.10 & 0.1/0.1 (0.1) \\
 &  & nG & 0.9974 & 1.0017 & 0.9977/1.0005 &  & 1.0043 & 1.0015 & -0.05 & 0.1/0.2 (0.1) \\
0.5 & $10^{-3}$ & G & 0.9989 & 1.0202 & 1.0001/1.0001 & 0.9990 & 1.0213 & 1.0217 & +0.30 & 1.3/1.5 (1.3) \\
 &  & nG & 0.9974 & 1.0153 & 0.9990/1.0002 &  & 1.0180 & 1.0175 & +0.10 & 1.4/1.5 (1.3) \\
0.9 & $10^{-4}$ & G & 0.9987 & 1.0027 & 0.9994/1.0009 & 0.9991 & 1.0040 & 1.0035 & +0.30 & 1.0/1.0 (0.9) \\
 &  & nG & 0.9963 & 1.0000 & 0.9983/0.9999 &  & 1.0037 & 1.0015 & +0.12 & 1.0/1.0 (0.9) \\
0.9 & $10^{-3}$ & G & 0.9896 & 1.0045 & 1.0001/1.0010 & 0.9911 & 1.0151 & 1.0151 & +1.25 & 9.6/10.2 (9.5) \\
 &  & nG & 0.9881 & 1.0010 & 0.9993/1.0004 &  & 1.0131 & 1.0125 & +1.01 & 9.7/10.3 (9.5) \\
0.95 & $10^{-4}$ & G & 0.9971 & 1.0034 & 0.9990/1.0017 & 0.9981 & 1.0063 & 1.0032 & +0.30 & 2.0/2.1 (1.9) \\
 &  & nG & 0.9948 & 0.9986 & 0.9984/0.9996 &  & 1.0039 & 1.0018 & +0.17 & 2.0/2.1 (1.9) \\
0.95 & $10^{-3}$ & G & 0.9799 & 0.9918 & 1.0000/1.0006 & 0.9814 & 1.0121 & 1.0134 & +1.22 & 19.6/20.6 (19.5) \\
 &  & nG & 0.9770 & 0.9904 & 0.9994/0.9994 &  & 1.0137 & 1.0116 & +1.05 & 19.7/20.8 (19.5) \\
\bottomrule
\end{tabular}
\end{table}
\section{A weight-normalized CNN with exact Hessian}\label{app:pilot}
A small CNN with every filter normalized in the forward pass \citep{salimans2016wn} gives the exact gradient, Hessian and noise covariance in float64. It uses two $3\times3$ convolutions, $3\to8\to16$ channels, average pooling, a fixed readout and a fixed CIFAR-10 subset (1,368 parameters in 24 groups, no batch normalization; SGD, $\eta=\lambda=0.01$). The test of the full angular covariance is mostly negative. At batch size 128 the sharpest mode exceeds the threshold ($\Rtop=1.24$ ReLU, 1.20 Softplus). Slow modes, whose relaxation times exceed the window, dominate there (median variance ratio 13.5 and 12.7). At batch size 16 (ReLU, $\Rtop=0.89$) the modes with $0.05\le R<0.5$ have variance ratios near one (median 1.20) and lag-1 autocorrelations close to $1-\eta\mu_k$. The ReLU remainder has a \emph{positive} boundary plateau, $\rho=+0.06$ to $+0.09$ at small $f$. However, the Softplus remainder vanishes as $f\to0$. The sign of the singular curvature is therefore model-dependent.
\section{Numerical checks and code}\label{app:numerics}
{\sloppy
Each numerical result comes from one of the calculations below. Paths are relative to the root of the reproducibility archive. \path{paper-jstat/tools/reproduce.sh} regenerates every table, figure and check from the stored data and compares them with this manuscript. \path{paper-jstat/REPRODUCE.md} lists the data, the code version of every campaign and the settings.

\small
\begin{enumerate}[label=N\arabic*,ref=N\arabic*,leftmargin=*,itemsep=1pt]
\item\label{num:lyap} Covariance of the surrogate (\cref{lem:geo}) against a numerical Lyapunov solution, 2,000 random stable parameter sets, largest relative difference $10^{-13}$: \path{paper-jstat/tools/check_geo_lyapunov.py}.
\item\label{num:identity} Pathwise identities: \eqref{eq:finite} on a nonlinear loss without scale invariance, $\beta\le0.95$, $\gamma\le0.2$ (residual $<5\cdot10^{-15}$), and the radial separation \eqref{eq:c-sgd}: \path{check_general_identity.py}. The scale-invariant form over 400-step windows from a non-stationary start: \path{check_finite_identity.py} (both in \path{paper-jstat/tools}).
\item\label{num:sphere} Sphere models (\cref{sec:sphere}, \cref{fig:toy}): \path{thermo-opt/sim} (\path{sphere.py}, \path{exp_*.py}, \path{geometry_check.py}); theory values recomputed from the stored chains by \path{repredict.py}.
\item\label{num:condense} Concentration (\cref{prop:condensation}): \path{thermo-opt/sim/exp_condense.py}. \path{repredict.py} prints the large-$D$ predictions ($131$, $458$, $1335$, $3687\,T^{(0)}$ at $\gamma=1.6\cdot10^{-3}$--$3\cdot10^{-2}$, $\gamma_c=5.0\cdot10^{-4}$) and the simulated values ($134$, $460$, $1336$, $3685$).
\item\label{num:zero} Finite decay on a quadratic sphere ($D$ tangent curvatures, isotropic tangent noise, $\Rtop\le0.1$, $\gamma=3\cdot10^{-3}$): \path{thermo-opt/sim/zero_memory_check.py}. At $\beta=0$, 0.5, 0.9, 0.95 the per-mode factor of \cref{lem:geo} departs from 1 by $(0,-1.0,-8.8,-18.0)\gamma$. The measured $c$ differs from it by $(+0.96,+0.02,-1.76,-3.12)\gamma$ for $D=8$ and $(+1.03,+0.32,-0.32,-0.66)\gamma$ for $D=32$ (standard errors 0.04--0.16 in units of $\gamma$). For SGD at $\gamma=10^{-2}$, $c-1=(1.00\pm0.02)\gamma$. In all runs \eqref{eq:F2} holds to $2.3\cdot10^{-4}$.
\item\label{num:head} Trained head (\cref{prop:symmetry}, \cref{thm:finite}): \path{teacher-student/tsm/head_checks.py}. Blockwise rescaling of four hidden units (2,000 steps with memory) maps trajectories to $10^{-15}$, while the head moves by $O(1)$. By comparison, a common rescaling of all learning rates changes the trajectories by $5\cdot10^{-2}$. For \eqref{eq:finite} with the head term, measured over predicted is $1.0001(3)$, $1.0000(3)$, $1.0001(2)$ at $\gamma=0.03$, $0.01$, $0.003$. Without the head term it is off by 1.5\%, 0.40\%, 0.06\%.
\item\label{num:ts} Teacher--student runs (\cref{sec:ts}, \cref{fig:ts-momentum}, \cref{tab:ts,tab:ts-ema}): \path{teacher-student/tsm} (\path{sweep.py} plans; exact population quantities at both ends of every step); figure, table and text by \path{paper-jstat/tools/ts_figure.py}, \path{ts_ema_table.py}, \path{fill_ts.py}.
\item\label{num:relu} Singular curvature of the ReLU committee (\cref{prop:relu}) against Monte Carlo: \path{paper-jstat/tools/check_relu_singular.py}.
\item\label{num:coupled} Coupled spheres (\cref{prop:coupled}): \path{thermo-opt/sim/coupled_current.py}.
\item\label{num:toy} One-dimensional anharmonic potentials (\cref{app:toy}): \path{thermo-opt/sim/transport_toy.py}.
\item\label{num:networks} Convolutional networks (\cref{sec:networks}): models, optimizers and diagnostics in \path{thermo-campaign} (configurations of v1--v4; A100-hours 73.7, 15.6, 9.3 and 51.7). Per-run summaries by \path{thermo-opt/analysis/campaign.py}. \Cref{tab:v1,tab:modes,tab:v2} by \path{paper/tools/tables.py}; memory statistics by \path{thermo-opt/analysis/momentum.py}. Six practical v1 runs were repeated with a guarded Lanczos projector after failing at dead channels.
\item\label{num:endpoint} Endpoint analysis of v3 and v4 (\cref{fig:overview,fig:endpoint}, \cref{tab:endpoint}): \path{thermo-opt/analysis/endpoint.py}, \path{endpoint_gap.py}; \path{paper-jstat/tools/endpoint_*.py}, \path{overview_figure.py}. \path{endpoint_numbers.py} prints every quoted statistic. Raw v4 results: \path{thermo-v4}, with a SHA-256 manifest.
\item\label{num:invariance} Invariance audit (\cref{sec:setup}): \path{thermo-opt/analysis/invariance.py}. Weight-normalized CNN (\cref{app:pilot}): \path{thermo-opt/pilot}.
\end{enumerate}
}

\end{document}